\documentclass{article}
\pdfoutput=1

\usepackage{lmodern}
\usepackage{cutwin}
\newcommand{\institute}[1]{}

\usepackage[export]{adjustbox}%

\usepackage{wrapfig}
\newcommand{\todo}[1]{{\color{blue}#1}}

\usepackage[utf8]{inputenc}

\usepackage{amssymb,amsmath}
\usepackage{enumerate} %
\usepackage[shortlabels]{enumitem} %
\usepackage{booktabs} %
\usepackage{graphicx} %

\usepackage{comment}
\newcommand*{\myTitle}{Deterministic Sparse Suffix Sorting in the Restore Model}
\usepackage[natbib,cref,pdf,theorems,ruby,color]{nikis} 
\hypersetup{%
    pdftitle={\myTitle},%
	pdfauthor={Johannes Fischer, Tomohiro I, and Dominik Koeppl},
}
\usepackage{fullpage}

\usepackage{etoolbox} %
\makeatletter
\pretocmd{\NAT@citexnum}{\@ifnum{\NAT@ctype>\z@}{\let\NAT@hyper@\relax}{}}{}{}
\makeatother

\newenvironment{textminipage}[1]{%
	\noindent%
	\edef\myindent{\the\parindent}%
	\begin{minipage}{#1}
	\setlength{\parindent}{\myindent}
}{\end{minipage}}

\usepackage{tabularx}
\newcolumntype{R}[1]{>{\raggedleft\arraybackslash}p{#1}}
\newcolumntype{C}[1]{>{\centering\arraybackslash}p{#1}} %

\usepackage{tikz}
\usepackage{xstring} %
\usetikzlibrary{matrix,backgrounds,positioning,arrows}
\usetikzlibrary{decorations.pathreplacing,fit,calc}

\newcount\Length
\pgfmathsetcount{\Length}{3}

\usetikzlibrary{arrows.meta}
\tikzset{>={Latex[width=2mm,length=1mm]}}

\newcommand{\RefLabels}{}
\newcommand{\Ref}[8]{%
\pgfmathsetcount{\Length}{#2-#1+1}

	\ifthenelse{\equal{#8}{}}{%
		\newcommand{\Label}{\texttt{(#1,\the\Length)}}
		\gappto{\RefLabels}{\begingroup \ttfamily\color{#6}(#1,}
		\xappto{\RefLabels}{\the\Length}
		\gappto{\RefLabels}{)\endgroup}
	}{%
		\newcommand{\Label}{\texttt{#8}}
		\gappto{\RefLabels}{\texttt{\color{#6}#8}}
	}

	\begin{pgfonlayer}{bg} 
	\node [fit=(warray-1-#3.south) (warray-1-#4.north), inner sep=0, inner xsep=5, fill=#6,opacity=0.4, rounded corners=4pt] (src) {};
\end{pgfonlayer}

\ifthenelse{\equal{#7}{l}}{%
	\draw [->, rounded corners, color=#6, opacity=0.5] (warray-1-#3.north) -- +(0,#5)-| node[above] {\Label{}} (warray-1-#1.north);
}{%
	\ifthenelse{\equal{#7}{s}}{%
		\draw [->, rounded corners, color=#6, opacity=0.5] (warray-2-#3.south) -- +(0,-#5) node[below] {\Label{}} -| (warray-2-#1.south);
	}{%
	\draw [->, rounded corners, color=#6, opacity=0.5] (warray-1-#3.north) -- +(0,#5) node[above] {\Label{}} -| (warray-1-#1.north);
}%
}%
}

\makeatletter
\newcommand{\GetTextChar}[1]{%
  \ensuremath{%
    \begingroup\fullexpandarg
    \let\$\relax \let\text\unexpanded
	\texttt{\StrChar{\exampleStringPlain}{#1}}
	\endgroup
  }%
}
\makeatother

\tikzstyle{array} = [matrix of nodes,font=\ttfamily, column sep=0.5\pgflinewidth, row sep=0.5mm, nodes in empty cells,
row 1/.style={nodes={scale=0.9, fill=none, minimum size=0mm, text height=0.5em}},
row 2/.style={nodes={font=\tiny,fill=none, minimum size=0mm, text height=0.33em}},
]

\usepackage{enumerate}
\usepackage{enumitem}
\newlist{passes}{enumerate}{10}
\setlist[passes]{label*=(\alph*)}
\Crefname{passesi}{Pass}{Passes}

\usepackage[linesnumbered,ruled,vlined]{algorithm2e}
\SetKwProg{Function}{function}{}{}
\SetKw{Continue}{continue}
\SetKw{Break}{break}
\SetKw{INCR}{incr}
\DontPrintSemicolon{}
	\SetKwInOut{Input}{input}
	\SetKwInOut{REQUIRE}{require}
	\SetKwInOut{Output}{output}
	\SetKwInOut{Result}{result}

\SetKwComment{Comment}{$\triangleright$\ }{}

\SetKwSty{MyKwFont}
\SetNlSty{bfseries}{\color{gray}}{}

\SetCommentSty{MyCommFont}

\newcommand{\leaf}     {\ensuremath{\lambda}}
\newcommand{\bv}[1]{\ensuremath{B_{\mathup{#1}}}}

\newcommand{\Time}[1]{\ensuremath{t_{\textup{#1}}}}

\newcommand{\SA}           {\instancename{SA}}

\newcommand{\SSA}  [1][]{\UnaryOperator[#1]{\instancename{SSA}}}
\newcommand{\SLCP} [1][]{\UnaryOperator[#1]{\instancename{SLCP}}}

\newcommand{\predecessor}[1][]{\UnaryOperator[#1]{\operatorname{predecessor}}}
\newcommand{\successor}[1][]{\UnaryOperator[#1]{\operatorname{successor}}}

\newcommand{\lcp}         [1][]{\UnaryOperator[#1]{\operatorname{lcp}}}

\newcommand{\lce}         [1][]{\UnaryOperator[#1]{\operatorname{lce}}}

\newcommand{\ibeg}[1]{\ensuremath{\mathsf{b}(#1)}}%
\newcommand{\iend}[1]{\ensuremath{\mathsf{e}(#1)}}%

\newcommand{\textS}{\ensuremath{S}}
\newcommand{\textT}{\ensuremath{T}}
\newcommand{\textX}{\ensuremath{X}}
\newcommand{\textY}{\ensuremath{Y}}
\newcommand{\textZ}{\ensuremath{Z}}

\newcommand{\period}{\ensuremath{p}}

\newcommand{\intervalI}{\ensuremath{\mathcal{I}}} %
\newcommand{\intervalJ}{\ensuremath{\mathcal{J}}} %
\newcommand{\intervalK}{\ensuremath{\mathcal{K}}} %

\newcommand{\IC}{..} %
\newcommand{\Int}[2]{\ensuremath{\lbrack#1\IC#2\rbrack}}
\newcommand{\substr}[3]{\ensuremath{\mathop{}\mathopen{}#1\mathopen{}\lbrack#2\IC#3\rbrack}} %
\newcommand{\substrI}[2]{\ensuremath{\mathop{}\mathopen{}#1\mathopen{}\lbrack#2\rbrack}} %

\newcommand{\exampleStringPlain}{aaababaaabaaba\$}

\newenvironment{subproof}{\block{Sub-Proof}}{\hfill\ensuremath{\blacksquare}}
\newenvironment{subclaim}{\block{Sub-Claim}}{}

\newcommand{\V}[1]{\ensuremath{#1}} %

\usepackage{tikz}
\usetikzlibrary{patterns,matrix}
\usetikzlibrary{shapes.geometric,positioning}
\tikzstyle{textstring} = 
	[ matrix of nodes
	, font=\Large\ttfamily
	, inner sep=0pt
	, minimum width=1.5em
	, minimum height=1em
	, nodes={draw=none, fill=none,text height=1.25em}
	, node distance=0mm
	, nodes in empty cells,
	, row 1/.style={nodes={draw=none, fill=none, color=gray, font=\tiny}}
	]
	\tikzstyle{arm} = [pattern=horizontal lines, pattern color=solarizedBlue!20] %
\tikzstyle{subarm} = [pattern=vertical lines, pattern color=green!20]
\tikzstyle{gap} = [pattern=north east lines, pattern color=gray!20]
\tikzstyle{rep} = [pattern=north west lines, pattern color=purple!20]
\tikzstyle{diff} = [pattern=dots, pattern color=cyan!20]

\newcount\ourRow%
\pgfmathsetcount{\ourRow}{0}

\newcommand{\SAVL}{\mathsf{SAVL}}
\newcommand{\SAVLT}{suffix AVL tree}
\newcommand{\eTree}{ET}
\newcommand{\hTree}{HT}
\newcommand{\thTree}{tHT}
\newcommand{\thTreeF}{$\eta$-truncated HSP tree}
\newcommand{\SubTree}[1]{\ensuremath{\mathcal{T}_{#1}}}

\newcommand{\etInst}[1][]{\UnaryOperator[#1]{\instancename{\eTree}}}
\newcommand{\hInst}[1][]{\UnaryOperator[#1]{\instancename{\hTree}}}
\newcommand{\thInst}[1][]{\UnaryOperator[#1]{\instancename{\thTree}_\eta}}
\newcommand{\DynLCE}{dynLCE}
\newcommand{\DynLCEInst}[1][]{\UnaryOperator[#1]{\instancename{\DynLCE}}}

\newcommand{\esp}[1][]{\UnaryOperator[#1]{\functionname{esp}}}
\newcommand{\espExp}[2]{\UnaryOperator[#2]{\functionname{esp}^{(#1)}}}
\newcommand{\mlcp}[1]{\UnaryOperator[#1]{\text{mlcp}}}
\newcommand{\mlcparg}[1]{\UnaryOperator[#1]{\text{mlcparg}}}

\newcommand{\LCEI}{LCE interval}
\newcommand{\timeSuffix}      {\Time{S}}
\newcommand{\timeDynLCE}[1]      {\UnaryOperator[#1]{\Time{L}}}
\newcommand{\timeDynMerge}[1]    {\UnaryOperator[#1]{\Time{M}}}
\newcommand{\timeDynConstruct}[1]{\UnaryOperator[#1]{\Time{C}}}

\newcommand{\memref}{\mathup{\tt ref}}
\newcommand{\memtext}{\mathup{\tt ival}}

\newcommand{\generated}[1]{\UnaryOperator[{#1}]{\mathit{string}}}
\newcommand{\invalid}{\ensuremath{\bot}} %

\newcommand{\numNodes}[1]{#1/2^\eta}

\newcommand{\Type}[1]{\texttt{type\,#1}}

\newcommand{\idset}{\mathcal{V}}
\newcommand{\idsetsize}[1]{\abs{\idset}} %
\newcommand{\dic}{\mathfrak{D}} %
\newcommand{\Comlcp}{\mathcal{C}} %
\newcommand{\Maxlcp}{\ensuremath{\ell}} %

\newcommand{\jcontext}{\ensuremath{\mathit{\Delta}}}
\newcommand{\lcontext}[1][]{\ensuremath{\jcontext_{\mathup{L}#1}}}
\newcommand{\rcontext}[1][]{\ensuremath{\jcontext_{\mathup{R}#1}}}

\newcommand{\espnode}[1]{\langle #1 \rangle}

\newcommand{\marginLCE}{\ensuremath{f}}
\newcommand{\gapLCE}{\ensuremath{g}}

\newcommand{\lookupTime}{\Time{look}}
\newcommand{\lookupTimeBest}{\lg n}
\newcommand{\lookupEta}{\lookupTime}
\newcommand{\lookupUpper}{\lookupTime}

\newcommand{\LCEITree}{\ensuremath{\mathcal{L}}} %
\newcommand{\Pos}{\mathcal{P}} %
\newcommand{\Suf}{\mathit{Suf}} %

\newcommand{\SetStr}{\mathcal{S}} %

\newcommand{\wordpack}{/ \log_{\sigma}n} %
\newcommand{\wordpackI}{\log_{\sigma}n} %

\newcommand{\MakeHSPName}[2]{\ensuremath{#1_{#2}}}
\newcommand{\Name}[1]{\texttt{#1}} %
\newcommand{\Char}[1]{\texttt{#1}} %

\usepackage{ifthen}
\newcommand{\C}[2][]{#2} %
\newboolean{withAlgo} 
\setboolean{withAlgo}{true}
\newboolean{withDiss}
\setboolean{withDiss}{false}

\newcommand{\ConstraintBox}[1]{\fcolorbox{white}{yellow!15}{\parbox{0.9\linewidth}{#1}}}
\newcommand{\RuleBox}[1]{\fcolorbox{white}{green!10}{\parbox{0.9\linewidth}{#1}}}

\newlist{LCEproperty}{enumerate}{1}
\setlist[LCEproperty]{leftmargin=*,label={Property~\arabic*:},ref={\arabic*}}
\Crefname{LCEpropertyi}{Property}{Properties}

\newlist{Facts}{enumerate}{1}
\setlist[Facts]{leftmargin=*,label={Fact~\arabic*:},ref={\arabic*}}
\Crefname{Factsi}{Fact}{Facts}

\newlist{Rules}{enumerate}{1}
\setlist[Rules]{leftmargin=*,label={Rule~\arabic*:},ref={\arabic*}}
\Crefname{Rulesi}{Rule}{Rules}

\newlist{Goals}{enumerate}{1}
\setlist[Goals]{leftmargin=*,label={Goal~\arabic*:},ref={\arabic*}}
\Crefname{Goalsi}{Goal}{Goals}

\newlist{Subsets}{enumerate}{1}
\setlist[Subsets]{leftmargin=*,label={Set~\arabic*:},ref={\arabic*}}
\Crefname{Subsetsi}{Set}{Sets}

\makeatletter
\newcommand{\CustomLabel}[2]{%
   \protected@write \@auxout {}{\string \newlabel {#1}{{#2}{\thepage}{#2}{#1}{}} }%
   \hypertarget{#1}{#2}%
}

\makeatother

\usepackage{calc}
\newlength{\lenTextFigureBlock}
\newcommand{\TextFigureBlock}[3]{%
	\setlength{\lenTextFigureBlock}{\linewidth-#1\linewidth}
	\begin{textminipage}{#1\linewidth}
		#2
	\end{textminipage}
	\begin{minipage}{\lenTextFigureBlock}
		\hfill
		#3
	\end{minipage}
}

\newboolean{isColor} 
\setboolean{isColor}{false} 
\setboolean{withAlgo}{false}
\ifthenelse{\boolean{withAlgo}}{%
	\usepackage[nohyperlinks, nolist]{acronym}
	\usepackage[linesnumbered,ruled,vlined]{algorithm2e}
	\SetKwComment{Comment}{$\triangleright$\ }{}
\SetKwProg{Function}{function}{}{}
}{}%

\ifthenelse{\boolean{isColor}}{%
	\renewcommand{\C}[2][]{#2}
}{%
	\renewcommand{\C}[2][]{#1}
}

\title{%
	\myTitle{}\footnote{%
		Parts of this work have already been presented at
		the 12th Latin American Symposium~\cite{fischer16determinstic}.
	}
}
\usepackage{authblk}
\author[1]{Johannes Fischer}
\author[2]{Tomohiro I}
\author[1]{Dominik K\"{o}ppl}%
\affil[1]{Department of Computer Science, TU Dortmund, Germany}
\affil[2]{Kyushu Institute of Technology, Japan}
\date{}

\ifthenelse{\boolean{withAlgo}}{}{%
\usepackage[nohyperlinks, nolist]{acronym}
}

\begin{acronym}
	\acro{DynLCE}[dynLCE]{dynamic LCE data type}
	\acroplural{DynLCE}[dynLCEs]{dynamic LCEs}
	\acro{LCE}[LCE]{longest common extension}
	\acroplural{LCE}[LCEs]{longest common extensions}
	\acro{HSP}[HSP]{hierarchical stable parsing}
	\acro{LCP}[LCP]{longest common prefix}
	\acroplural{LCP}[LCPs]{longest common prefixes}
\acro{CFG}{context free grammar}

        \acro{ST}{suffix tree}
        \acro{CST}{compressed suffix tree}
        \acro{SST}{succinct suffix tree}
	\acro{MAST}{minimal augmented suffix tree}

        \acro{BWT}{Burrows-Wheeler Transform}
        \acro{LZ}{Lempel-Ziv}
        \acro{BSPT}{binary search prefix tree}
        \acro{LZW}{Lempel-Ziv-Welch}
        \acro{LZ78}{Lempel-Ziv-78}
        \acro{LZ77}{Lempel-Ziv-77}
		\acro{preapp}[pre-/appending]{prepending or appending}

		\acro{LCG}{linear congruential generator}
		\acroplural{LCG}[LCGs]{linear congruential generators}
\end{acronym}

\input{./papermeta/HSPname.texdic}
\input{./papermeta/ESPname.texdic}
\begin{document}

\maketitle

\begin{abstract}
Given a text $T$ of length $n$,
we propose a deterministic online algorithm computing the sparse suffix array and the sparse longest common prefix array of $T$
in
$\Oh{c \sqrt{\lg n} + m \lg m \lg n \lg^* n}$ time 
with $\Oh{m}$ words of space
under the premise that the space of~$T$ is rewritable, where $m \le n$ is the number of suffixes to be sorted (provided online and arbitrarily), 
and $c$ is the number of characters with $m \le c \le n$ that must be compared for distinguishing the designated suffixes.
\end{abstract}

\ifthenelse{\boolean{withAlgo}}{%
\newcommand{\Bild}[2][]{\includegraphics[#1]{sparsesuf/images/#2}}
	\def\Jvar{}
	\newcommand{\JJ}[2][]{#2}
	\newcommand{\JO}[2][]{#2}
	\newcommand{\Jitemize}[1]{\sitemize{#1}}
}{%
\newcommand{\Bild}[2][]{\includegraphics[#1]{images/#2}}
	\def\Jvar{%
		\newpage
		\appendix
	}
	\newcommand{\JJ}[2][]{\appto\Jvar{#1#2}}
	\newcommand{\JO}[2][]{#1}
	\newcommand{\Jitemize}[1]{\sitemize*{#1}}
}%
\newcommand{\SemiOrStable}{\mbox{(semi-)}stable}

\JO{%
\chapter{Sparse Suffix Sorting}\label{chSparseSuf}

\defcitealias{vannevar45asWeMayThink}{Vannevar Bush}
\Zitation{%
A record, if it is to be useful to science, \\
must be continuously extended, \\
it must be stored, \\
and above all it must be consulted.
}{vannevar45asWeMayThink}
}%

\section{Introduction}
Sorting suffixes of a long text lexicographically is an important first step for many text processing algorithms~\cite{puglisi07taxonomy}.
The complexity of the problem is quite well understood, as for integer alphabets
suffix sorting can be done in optimal linear time and in-place~\cite{LiLH16a,goto2017}.
In this article, we consider a variant of the problem:
instead of computing the order of \emph{every} suffix, we address the \intWort{sparse suffix sorting problem}.
Given a text $T[1..n]$ of length $n$ and a set $\Pos\subseteq [1..n]$ of~$m$ arbitrary positions in $T$,
the problem asks for the (lexicographic) order of the suffixes starting at the positions in $\Pos$.
The answer is encoded by a permutation of $\Pos$, which
is called the \intWort{sparse suffix array~(SSA)} of~$T$ (with respect to $\Pos$) and denoted by \SSA[T,\Pos].

Applications are found 
in external memory LCP-array construction algorithms~\cite{KarkkainenK16},
and in the search of maximal exact matches~\cite{btp275,btt042}, 
i.e., substrings found in two given strings that can be extended neither to their left nor to their right without getting a mismatch.

Like the ``full'' suffix arrays,
we can enhance \SSA[T,\Pos] with the lengths of the \acp{LCP} between adjacent suffixes in \SSA[T,\Pos].
These lengths are stored in the \intWort{sparse longest common prefix array~(SLCP)}, which we denote by \SLCP[T,\Pos].
In combination, \SSA[T,\Pos] and \SLCP[T,\Pos] store the same information as the \intWort{sparse suffix tree}, i.e.,
they implicitly represent a compacted trie over all suffixes starting at the positions in $\Pos$.
The sparse suffix tree is an efficient index for pattern matching~\cite{Kolpakov11sparse}.

Based on classic suffix array construction algorithms~\cite{kaerkkaeinen06linearSA,nong11two}, 
sparse suffix sorting is easily conducted in $\Oh{n}$ time if $\Oh{n}$ words of additional working space is available.
For $m = \oh{n}$, however, the working space may be too large, compared to the final space requirement of~\SSA[T,\Pos]\@.
Although some special choices of~$\Pos$ admit space-optimal $\Oh{m}$-words construction algorithms (e.g.~\cite{Karkkainen96sparse}, see also the related work listed in~\cite{Bille0GKSV16}),
the problem of sorting arbitrary suffixes in small space seems to be much harder. We are aware of the following results:
As a deterministic algorithm, \citet{kaerkkaeinen06linearSA}
gave a trade-off using $\Oh{\tau m +n\sqrt{\tau}}$ time and $\Oh{m+n/\sqrt{\tau}}$ words of working space 
with a parameter $\tau \in [1..\sqrt{n}]$.
If randomization is allowed, there is a technique based on Karp-Rabin fingerprints,
first proposed by \citet{Bille0GKSV16} and later improved by \citet{I2014FSS}.
\citet{GawrychowskiK17} presented an algorithm running with \Oh{m} words of additional space in either \Oh{n\sqrt{\lg m}} expected time, or in $\Oh{n}$ time as a Monte Carlo algorithm (i.e., the output is correct only with high probability).
Most recently, \citet{Prezza18sparse} presented a Monte Carlo algorithm in the restore model~\cite{Chan14restoreModel} that 
runs with \Oh{m} words of space in \Oh{n + m \lg^2 n} expected time.
\JO{\Cref{figSparseSufRelatedWork} summarizes the running times and the memory usage of the listed algorithms.
\begin{figure}[t]
	\centerline{%
		\begin{tabular}{llll}
			\toprule
			Time                  & Space              & Restriction           & Reference
			\\
	{\Oh{n \lg n}}                & {\Oh{n}}           &                       & \cite{manber93suffix}
			\\
			{\Oh{n}}                & {\Oh{n}}           &                       & \cite{KoA05}
			\\
	{\Oh{n \lg n}}                & {$n+\Oh{1}$} space &                       & \cite{Franceschini07sa}
			\\
	{\Oh{n}}                      & {$n+\Oh{1}$} space &                       & \cite{goto2017,LiLH16a}
			\\
	{\Oh{n}}                      & {\Oh{m}}           & $\Pos$ evenly spaced  & \cite{Karkkainen96sparse}
			\\
	{\Oh{n \lg^2 m}}              & {\Oh{m}}           & MC           & \cite{Bille0GKSV16}
			\\
	{\Oh{n \lg^2 m + m^2 \lg m}}  & {\Oh{m}}           & LV             & \cite{Bille0GKSV16}
			\\
	{\Oh{n}}                      & {\Oh{m}}           & MC           & \cite{GawrychowskiK17}
			\\
			{\Oh{n \sqrt{\lg m}}} & {\Oh{m}}           & LV             & \cite{GawrychowskiK17}
			\\
	{\Oh{n \lg m}}                & {\Oh{m}}           & MC or LV & \cite{I2014FSS}
			\\
	{\Oh{n}}                      & {\Oh{m \lg m}}     & MC           & \cite{I2014FSS}
			\\
			{\Oh{n + m \lg^2 n}}  & {\Oh{1}}             & MC, restore model     & \cite{Prezza18sparse}
			\\\bottomrule
		\end{tabular}
	}%
	\caption{Sparse suffix sorting algorithms. MC and LV denote Monte Carlo and Las Vegas algorithms. }
	\label{figSparseSufRelatedWork}
\end{figure}
}%

\subsection{Computational Model}
Let $\lg$ and $\log_x$ denote the logarithm to the base two and to the base~$x$ for a real number~$x$, respectively.
Our computational model is the word RAM model with word size $\Om{\lg n}$.
Here, characters use $\upgauss{\lg\sigma}$ bits, where $\sigma$ is the alphabet size; hence, $\gauss{\log_\sigma n}$ characters can be packed into one word. Comparing two strings $\textX$ and~$\textY$ therefore takes 
$\Oh{\lcp(\textX, \textY)\wordpack}$
time, where $\lcp(\textX, \textY)$ denotes the length of the \ac{LCP} of $\textX$ and $\textY$.

We assume that the text~\textT{} of length $n$ is loaded into RAM\@.
We work with the restore model~\cite{Chan14restoreModel}, where algorithms are allowed to overwrite parts of~$\textT$, as long as they can restore~$\textT$ to its original form at termination. 
Apart from this space, we are only allowed to use $\Oh{m}$ words.
The positions in $\Pos$ are assumed to arrive on-line, implying in particular that they need not be sorted.
We aim at worst-case efficient \emph{deterministic} algorithms.

\subsection{Algorithm Outline and Our Results}\label{secSparseSuffAlgoOutline}

Our main algorithmic idea is to insert the suffixes starting at the positions of $\Pos$ into a self-balancing binary search tree~\cite{Irving2003sbs}; since each insertion invokes $\Oh{\lg m}$ suffix-to-suffix comparisons, 
the time complexity is $\Oh{\timeSuffix m \lg m}$, where $\timeSuffix$ is the cost for a suffix-to-suffix comparison.
If all suffix-to-suffix comparisons are conducted naively by comparing the characters ($\timeSuffix = \Oh{n \wordpack}$),
the resulting worst case time complexity is $\Oh{n m \lg m \wordpack}$.
In order to speed this up,
our algorithm identifies large identical substrings at different positions
during different suffix-to-suffix comparisons.
Instead of performing naive comparisons on identical parts over and over again,
we build a data structure (stored in redundant text space) to accelerate 
subsequent suffix-to-suffix comparisons. 
Informally, when two (possibly overlapping) substrings in the text are detected to be the same, one of them can be overwritten.

To accelerate suffix-to-suffix comparisons,
we devise a new data structure called \intWort{\ac{HSP} tree} that is based on \intWort{edit sensitive parsing (ESP)}~\cite{Cormode2007sed}.
The HSP tree supports \intWort{\ac{LCE}} queries.
An \ac{LCE} query~$\lce(i,j)$ on an HSP tree asks for the length~$\lcp(T[i..],T[j..])$ 
of the \ac{LCP} of two suffixes starting at the respective positions~$i$ and~$j$ of the text~$T$ on which the tree is built.
Besides answering LCE queries, HSP trees are \emph{mergeable}, 
allowing us to build a dynamically growing \ac{LCE} index on substrings read in the process of the sparse suffix sorting.
Consequently, comparing two already indexed substrings is done by a single \ac{LCE} query.

In their plain form, HSP trees need more space than the text itself;
to overcome this space problem, we devise a \emph{truncated} version of the HSP tree,
yielding a trade-off parameter between space consumption and \ac{LCE} query time.
By choosing this parameter appropriately, the truncated HSP tree fits into the text space.
With a text space management specialized on the properties of the HSP, we achieve the result of \cref{thmSparseSuffixSorting} below.

We make the following definition that allows us to analyze the running time more accurately.
Define $\Comlcp := \bigcup_{p,p' \in \Pos, p \neq p'} [p..p+\lcp(T[p..], T[p'..])]$ as the set of positions that must be compared for distinguishing the suffixes starting at the positions of~$\Pos$.
Then sparse suffix sorting is trivially lower bounded by $\Om{\abs{\Comlcp}\wordpack}$ time.
With the definition of~$\Comlcp$, we now can state the main result of this article as follows:

\begin{theorem}\label{thmSparseSuffixSorting}
	Given a text $T$ of length $n$ that is loaded into RAM,
  the SSA and SLCP of $T$ for a set of $m$ arbitrary positions can be computed deterministically in 
  $\Oh{\abs{\Comlcp} (\sqrt{\lg \sigma} + \lg \lg n) +  m \lg m \lg n \lg^* n} =
 \Oh{\abs{\Comlcp} \sqrt{\lg n} +  m \lg m \lg n \lg^* n}$ time, using $\Oh{m}$ words of additional working space.
\end{theorem}

Excluding the loading cost for the text, the running time can be sublinear (when $\abs{\Comlcp} = \oh{n/\sqrt{\lg n}}$ and $m \lg m = \oh{n / \lg n \lg^* n}$).
To the best of our knowledge, this is the first algorithm that refines the worst-case performance guarantee. %
All previously mentioned (deterministic and randomized) algorithms take $\Om{n}$ time even if we exclude the loading cost for the text.
Also, general string sorters (e.g., forward radix sort~\cite{forwardradixsort} or multikey quicksort~\cite{BentleyS97}),
which do not take advantage of the overlapping of suffixes,
suffer from the lower bound of $\Om{\Maxlcp\wordpack}$ time,
where $\Maxlcp$ is the sum of all LCP values in the SLCP, which is always at least $\abs{\Comlcp}$, but can in fact be $\Theta(nm)$.

\subsection{Relationship Between Suffix Sorting and LCE Queries}
The LCE-problem is to preprocess a text $T$ such that subsequent LCE queries can be answered efficiently.
Data structures for LCE and sparse suffix sorting are closely related, as shown in the 
following \lcnamecref{obsLCEconnection}:

\begin{observation}\label{obsLCEconnection}
	Given a data structure that answers LCE queries in $\Oh{\tau}$ time for $\tau > 0$,
	we can compute sparse suffix sorting for $m$ positions in $\Oh{\tau m \lg m}$ time
	by inserting suffixes into a balanced binary search tree~\cite{Irving2003sbs}. %
	Conversely, given an algorithm computing the SSA and the SLCP of a text $T$ of length $n$ for $m$ positions in $\Oh{f(n,m)}$ time
	with $\Oh{m}$ words of space for a function $f$, 
	we can construct a data structure in $\Oh{\max(f(n,m), n/m})$ time with $\Oh{m}$ words of space, answering LCE queries on $T$ in $\Oh{n^2/m^2}$ time.
\end{observation}
	\begin{proof}
		The first claim is trivial.
		For the second claim, we use the data structure of~\cite[Theorem 1a]{billeLCEPre} that 
		answers LCE queries in~\Oh{\tau} time.
The data structure uses the SSA and SLCP values of those suffixes whose starting positions are in a difference cover sampling modulo~$\tau$.
This difference cover consists of \Oh{n/\sqrt{\tau}} text positions, and can be computed in~\Oh{\sqrt{\tau}} time~\cite{ColbournL00}.
We obtain the claimed bounds on time and space by setting~$\tau := n^2/m^2$.
	\end{proof}

\begin{figure}
	\centering{%
		\small
	\begin{tabular}{lllll}
		\toprule
		\multicolumn{2}{c}{Construction}  & \multicolumn{2}{c}{Data Structure} \\ \cmidrule(lr){1-2} \cmidrule(lr){3-4}
		Time & Working Space   & Space                                  & Query Time                       & Ref \\ \midrule
		{\OhS{n\tau}}     & {\OhS{\frac{n}{\tau}}} & ${\OhS{\frac{n}{\tau}}}$ & {\OhS{\tau \lg\min(\tau, \frac{n}{\tau}}} & \cite{TanimuraIBIPT16} \\
		{\OhS{n^{2+\epsilon}}}       & {\OhS{\frac{n}{\tau}}}   & {\OhS{\frac{n}{\tau}}}                          & {\OhS{\tau}}                      & \cite{billeLCE} \\
		${\OhS{n \tupleS{\lg^*n + \frac{\lookupTimeBest}{\tau} + \frac{\lg\tau}{\wordpackI} }}}$ & 
		${\OhS{\max\tupleS{\frac{n}{\lg n}, \tau^{\lg 3} \lg^* n} }}$ 	& 
		${\OhS{\frac{n}{\tau}}}$ & 
		${\OhS{\lg^* n \tupleS{\lg\tupleS{\frac{\ell}{\tau}} + \frac{\tau^{\lg 3}}{\wordpackI} }}}$ &
		Thm.~\ref{thmLCEtradeoff}
		\\
		${\OhS{n \tupleS{\lg^*n + \frac{\lookupTimeBest}{\tau} + \frac{\lg\tau}{\wordpackI} }}}$ & 
		${\OhS{\tau^{\lg 3} \lg^* n }}$ 	& 
		${\OhS{\frac{n}{\tau}}}$ & 
		${\OhS{\lg^* n \tupleS{\lg\tupleS{\frac{n}{\tau}} + \frac{\tau^{\lg 3}}{\wordpackI} }}}$ &
		Cor.~\ref{lemmaTruncatedLCEPointer}\\
		\bottomrule
	\end{tabular}
}%
	\caption{Deterministic LCE data structures with trade-off parameters. 
		The length returned by an LCE query is denoted by $\ell$.
		$\epsilon$ and $\tau$ with $\epsilon >0$ and $1 \le \tau \le n$ are constants. Space is measured in \emph{words}.
	The column \emph{Working Space} lists the working space needed to construct a data structure, 
	whereas the column \emph{Space} lists the final space needed by a data structure. 
}
	\label{tableDeterministic}
\end{figure}

There has been a great interest in devising deterministic LCE data structures with trade-off parameters~(see \cref{tableDeterministic}),
or in compressed space~\cite{tanimura17lce,NishimotoIIBT16,I16}.
One of the currently best data structures with a trade-off parameter is due to
\citet{TanimuraIBIPT16}, using ${\Oh{n/\tau}}$ words of space and answering LCE queries in ${\Oh{\tau \lg\min(\tau,n/\tau)}}$ time, for a trade-off parameter~$\tau$ with $1 \le \tau \le n$. 
However, this data structure has a preprocessing time of $\Oh{n\tau}$, and is thus not helpful for sparse suffix sorting. We develop a new data structure for LCE with the following properties.

\begin{theorem}\label{thmLCEtradeoff}
	There is a \emph{deterministic} data structure
	using $\Oh{n/\tau}$ words of space
	that answers an LCE query~$\ell := \lce(i,j)$ for two text positions~$i$ and $j$ with $1 \le i, j \le n$ on a text of length~$n$ in $\Oh{\lg^* n \tuple{\lg\tuple{\ell/\tau} + \tau^{\lg 2} \wordpack }}$ time, 
	where $1 \le \tau \le n$.
	We can build the data structure in $\Oh{n \tuple{\lg^*n + (\lookupTimeBest) / \tau + (\lg\tau) \wordpack }}$ time
	with additional $\Oh{\max(n/\lg n, \tau^{\lg 3} \lg^* n)}$ words during construction.
\end{theorem}

The construction time of our data structure is upper bounded by $\Oh{n \lg n}$, and hence it can be constructed faster than the 
deterministic data structures in~\cite{TanimuraIBIPT16} when $\tau = \Om{\lg n}$.

\subsection{Outline of this Article}
We start with \cref{secESP} introducing the edit sensitive parsing, and giving a motivation for our hierarchical stable parsing
whose description follows in \cref{secHSP}.
\Cref{secLCE} shows the general techniques for answering LCE queries with the HSP tree.
Subsequently, \cref{secSSS} introduces our algorithm for the sparse suffix sorting problem
with an abstract data type \emph{\DynLCE{}} that supports LCE queries and a merging operation.
The remainder of that section shows that the HSP tree from \cref{secHSP} fulfills all properties of a \DynLCE{};
in particular, HSP trees support the merging operation.
The last part of this article is dedicated to the study on how the text space can be exploited with the HSP technique to improve the memory footprint.
This leads us to truncated HSP trees with a merging operation that is tailored to working in text space (\cref{secSparseSuffixSortingTextSpace}).
With the truncated HSP trees we finally solve the sparse suffix sorting problem in the time and space as claimed in \cref{thmSparseSuffixSorting}.

\subsection{Preliminaries}
Let $\Sigma$ be an ordered alphabet of size $\sigma$ whose characters are represented by integers.
For a string $\textX \in \Sigma^{*}$, let $|\textX|$ denote the length of $\textX$.
For a position $1 \le i \le \abs{\textX}$ in $\textX$, let $\textX[i]$ denote the $i$-th character of $\textX$.
For positions $i$ and $j$ with $1 \le i,j \le \abs{\textX}$, let $\textX[i..j] = \textX[i] \textX[i+1] \cdots \textX[j]$.
Given $\textT = \textX\textY\textZ$ with $\textX,\textY,\textZ \in \Sigma^*$, $\textX$, $\textY$ and $\textZ$ are called a \intWort{prefix}, \intWort{substring}, \intWort{suffix} of $\textT$, respectively.
In particular, the suffix beginning at position $i$ is denoted by $\textT[i..]$.
A \intWort{period} of a string $\textY$ is a positive integer $p < \abs{\textY}$ such that $\textY[i]=\textY[i+p]$ for all integers~$i$ 
with $1 \le i \le \abs{\textY}-p$.

For a binary string~$T \in \menge{0,1}^*$
we are interested in the operation $T.\rank[1](j)$ that counts the number of~\bsq{1}s in $T[1..j]$.
This operation can be performed in constant time by a data structure~\cite{jacobson89space} that takes $\oh{\abs{T}}$ extra bits of space, and can be constructed in time linear in~$\abs{T}$.

An \intWort{interval} $\intervalI=[b..e]$ is the set of consecutive integers from $b$ to $e$, for $b\le e$. 
For an interval $\intervalI$, we use the notations $\ibeg{\intervalI}$ and $\iend{\intervalI}$ to denote the beginning and the end of $\intervalI$;
i.e., $\intervalI = [\ibeg{\intervalI}..\iend{\intervalI}]$.
We write $\abs{\intervalI}$ to denote the length of $\intervalI$; i.e., $\abs{\intervalI}=\iend{\intervalI}-\ibeg{\intervalI}+1$.

\section{Edit Sensitive Parsing}\label{secESP}
The crucial technique used in this article is the so-called alphabet reduction.
The alphabet reduction is used to partition a string deterministically into blocks.
The first work introducing the alphabet reduction technique to the string context was done by~\citet{Mehlhorn94sigenc}.
They presented the so-called \emph{signature encoding}.
The signature encoding is derived from a tree coloring approach~\cite{Goldberg87treecolor}.
It supports string equality checks in the scenario where strings can be dynamically concatenated or split.
In the same context,
\citet{sahinalp94stoc} studied the maximal number of characters to the left and to the right of a substring~$\textZ$ of~$\textY$ such that
changing one of these characters to the left or to the right of~$\textZ$ can affect how~$\textZ$ is parsed by the signature encoding of~$\textY$.
In a later work, \citet{Alstrup2000Pmi} enhanced signature encoding with additional queries like LCE\@.
Recently, an LCE data structure using signature encoding in compressed space was shown by \citet{NishimotoIIBT16}.
A slightly modified version of signature encoding is proposed by \citet{SakamotoMKS09},
showing that alphabet reduction can be used to build a grammar compressor whose approximation ratio to the size of the smallest grammar is \Oh{\lg^*n \lg n}.

A modified parsing was introduced by \citet{Cormode2007sed}. 
They modified the parsing by restricting the block size from two up to three characters, 
and named their technique \emph{edit sensitive parsing}~(ESP).
Initially used for approximating the edit distance with moves, the ESP technique has been found to be applicable to building self-indexes~\cite{TakabatakeNKTS16}. %
We stick to the ESP technique, because the size of the subtree of a node in the ESP tree is bounded.
In this section, we first introduce the ESP technique, and then give a motivation for a modification of the ESP technique, which we call hierarchical stable parsing (HSP).
Before that, we recall the alphabet reduction and the ESP trees.

\subsection{Alphabet Reduction}\label{secAlphabetReduction}
Given a string~$\textY$ in which no two adjacent characters are the same, i.e., 
$\textY[i-1] \not= \textY[i]$ for every integer~$i$ with $2 \le i \le \abs{\textY}$,
we can partition~$\textY$ (except at most the first $\lg^* \sigma$ positions) into \intWort{blocks} of size two or three with a technique called \intWort{alphabet reduction}~\cite[Section~2.1.1]{Cormode2007sed}.
It consists of three steps (see also \cref{figAlphabetReduction}):
First, it reduces the alphabet size to at most eight, in which every character has a rank from zero to seven.
Subsequently, it substitutes characters with ranks four to seven with characters having a rank between zero and two.
By doing so, it shrink the alphabet size to three.
Finally, it identifies certain text positions as landmarks that determine the block boundaries.

For reducing the alphabet size, we assume that $\sigma \ge 9$, otherwise we skip this step.
The task is to generate a surrogate string~$\textZ$ on the alphabet~$\menge{0,1,2}$ such that 
the entry $\textZ[i]$ depends only on the substring $\textY[i..i+\lg^*\sigma]$, for $1 \le i \le \abs{\textY}-\lg^*\sigma$.
To this end, we regard~$\textY$ as an array of binary numbers, i.e., 
$\textY[i][\ell] \in \menge{0,1}$ for an integer~$\ell$ with $1 \le \ell \le \upgauss{\lg \sigma}$.
We create an array~$\textZ$ of length~$\abs{\textY}-1$ storing integers of the domain~$[0..2\upgauss{\lg \sigma}-1]$.
For each text position~$i$ with $2 \le i \le \abs{\textY}$, we compare~$\textY[i]$ with~$\textY[i-1]$:
We compute $\ell := \lcp(\textY[i-1], \textY[i])$, and write $2 \ell + \textY[i][\ell+1]$ to $\textZ[i]$ (remember that we treat $\textY[i]$ as a binary string\footnote{Fix an arbitrary rule whether~$\textY[i][1]$ is the least significant or most significant bit.}).
By doing so, no two adjacent integers are the same in~$\textZ$~\cite[Lemma 1]{Cormode2007sed}.
Having computed $\textZ$, we recurse on~$\textZ$ until~$\textZ$ stores integers of the domain~$\menge{0,\ldots,5}$.
Note that the alphabet cannot be reduced further with this technique, since $2\upgauss{\lg x} \ge x$ for every integer~$x$ with $2 \le x \le 6$.
To obtain the final~$\textZ$, we recurse at most~$\lg^* \sigma$ times.
Let~$r$ be the number of recursions. Then we have $\abs{\textY} = \abs{\textZ}+r$.

If we skipped this step because of a small alphabet size ($\sigma \le 8$), 
then we set $\textZ[i]$ to the rank of $\textY[i]$ induced by the linear order of $\Sigma$ (e.g., $\textZ[i] = 0$ if $\textY[i]$ is the smallest character).
Since $\abs{\textY} = \abs{\textZ}$, we set $r$ to zero.

To reduce the domain further, we iterate over the values $j = 3,\ldots,8$ in ascending order, substituting
each $\textZ[i] = j$ with the lowest value of $\menge{0,1,2}$ that does not occur in its neighboring entries ($\textZ[i-1]$ and $\textZ[i+1]$, if they exist).
Finally, $\textZ$ contains only numbers between zero and two.

In the final step we create the landmarks that determine the block boundaries.
The landmarks obey the property that the distance between two subsequent landmarks is greater than one, but at most three.
They are determined by local maxima and minima:
First, each number~$\textZ[i]$ that is a local maximum is made into a landmark.
Second, each local minimum that is not yet neighbored by a landmark is made into a landmark.

Finally, we create blocks by associating each position in~$\textZ$ with its closest landmark.
Positions associated with the same landmark are put into the same block.
As a tie breaking rule we favor the right landmark in case that there are two closest landmarks.
The last thing to do is to map each block covering~$\textZ[i..j]$ to $\textY[i+r..j+r]$.

The tie breaking rule can cause a problem when~$\textZ[1]$ and $\textZ[3]$ are landmarks, i.e., the leftmost block contains only one character.
We circumvent this problem by fusing the blocks of the first and second landmark to a single block. If this block covers four characters, we split it evenly.

Altogether, the alphabet reduction needs \Oh{\abs{\textY} \lg^*\sigma} time, since we perform $r \le \lg^* \sigma$ reduction steps,
while determining the landmarks and computing the blocks take \Oh{\abs{\textY}} time.
The steps are summarized in the following \lcnamecref{lemAlphabetReductionSpeed}:

\begin{lemma} \label{lemAlphabetReductionSpeed}
	Given a string~$\textY$ in which no two adjacent characters are the same,
	the alphabet reduction applied on~$\textY$ partitions~$\textY$ into blocks, except at most $\upgauss{\lg^* \sigma}$ positions at the left.
It runs in~\Oh{\abs{\textY} \lg^* \sigma} time.
\end{lemma}
The main motivation of introducing the alphabet reduction is the following \lcnamecref{lemAlphabetReduction} that shows that 
applying the alphabet reduction on a text~$\textY$ and on a pattern~$\textX$ generates the same blocks in $\textX$ as 
in all occurrences of~$\textX$ in $\textY$, except at the left and right borders of a specific length:

\vspace{1em}
\TextFigureBlock{0.6}{%
\begin{lemma}[{\cite[Lemma 4]{Cormode2007sed}}] \label{lemAlphabetReduction}
	Given a substring~$\textX$ of a string~$\textY$ in which no two adjacent characters are the same,
	the alphabet reduction applied to~$\textX$ alone creates the same blocks as the blocks representing the substring~$\textX$ in $\textY$, 
	except for at most $\lcontext := \upgauss{\lg^* \sigma} + 5$ characters at the left border, and $\rcontext := 5$ characters at the right border.
\end{lemma}
}{%
\Bild[scale=1.0]{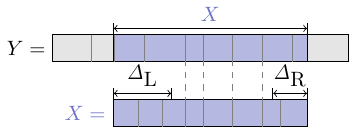}
}%
\vspace{0.5em}

\begin{figure}[t]
	\centering{%
		\Bild[width=\textwidth]{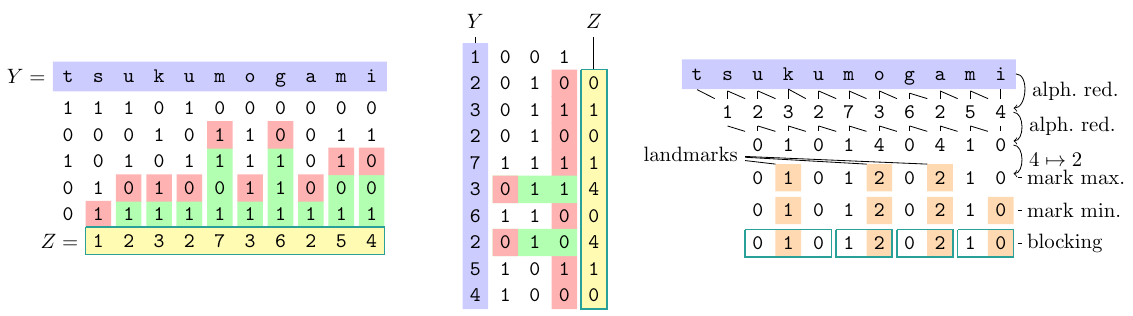}
	}
	\caption{Alphabet reduction applied on the string~$\textY = \Char{tsukumogami}$. 
		We represent the characters with the five lowest bits of the ASCII encoding. 
		\emph{Left:} A single step of the alphabet reduction.
		The bit representation of each character $\textY[i]$ is shown vertically in the left figure (the most significant bit is on the top).
	The alphabet reduction matches the least significant bits\C{~(shaded green)} of two adjacent entries, and returns twice the number of matched bits plus the mismatched bit of the right character (\C[highest shaded bit]{shaded red}).
		The resulting integer array~$\textZ$ is the last row.
		\emph{Middle:} A second step of the alphabet reduction, where the result of the first alphabet reduction stored in $\textZ$ is put into $\textY$.
		\emph{Right:} Computation of the blocks. Two steps of the alphabet reduction (seen in the left and in the middle image) yield a sequence consisting only of integers within the domain $\menge{0,\ldots,4}$.
		Subsequently, all \bsq{4}s are replaced (in this case by \bsq{2} since the neighboring values are \bsq{0} and \bsq{1} in both cases), and the maxima and certain minima are made into landmarks (\C[shaded]{shaded orange}).
		Finally, the\C{cyan} boxes in the last row are the computed blocks.
}
	\label{figAlphabetReduction}
\end{figure}

\begin{figure}[t]
	\centering{%
		\begin{adjustbox}{valign=t}
		\Bild[scale=1.0]{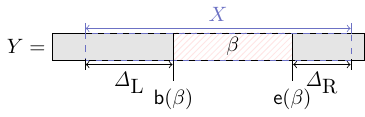}
	\end{adjustbox}
		\hspace{5em}
		\begin{adjustbox}{valign=t}
	\Bild[scale=1.0]{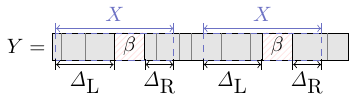}
	\end{adjustbox}
	}
	\caption{\emph{Left:} Surrounded block~$\beta$ with local surrounding~$\textX$ contained in a string~$\textY$. 
		\emph{Right:} Occurrences of the local surrounding~$\textX$ of a surrounded block~$\beta$ in the string~$\textY$, which is 
partitioned into blocks (gray rectangles) by the edit sensitive parsing. 
Although the occurrences of~$\textX$ can be differently blocked at their borders, they all have a block equal to~$\beta$ in common.
}
	\label{figSurroundedBlock}
\end{figure}

Given a block~$\beta$, we call the substring $\textY[\ibeg{\beta}-\lcontext..\iend{\beta}+\rcontext]$ the \intWort{local surrounding} of $\beta$, if it exists 
(i.e., $\ibeg{\beta} - \lcontext \ge 1$ and $\iend{\beta}+\rcontext \le \abs{\textY}$).
	Blocks whose local surroundings exist are also called \intWort{surrounded}.
	A consequence of \cref{lemAlphabetReduction} is the following:
	Given that $\textX$ is the local surrounding of a surrounded block~$\beta$, then
	the blocking of every occurrence of $\textX$ in $\textY$ is the same, except at most $\lcontext$ and $\rcontext$ characters at the left and right borders, respectively.
	We conclude that the blocking of every occurrence of $\textX$ has a block~$\textX[1+\lcontext..\lcontext+\abs{\beta}]$ that is equal to $\textY[\ibeg{\beta}..\iend{\beta}]$ (see \cref{figSurroundedBlock}).

\subsection{Edit Sensitive Parsing}\label{subsecESP}
Whenever a string~$\textY$ contains a repetition of a character at two adjacent positions, we cannot parse~$\textY$ with the 
alphabet reduction. A solution is to additionally use an auxiliary parsing specialized on repetitions of the same character.
With this auxiliary parsing, we can partition~$\textY$ into substrings, 
where each substring is either parsed with the alphabet reduction, or with the auxiliary parsing.
It is this auxiliary parsing where the aforementioned signature encoding and the ESP technique differ.
The main difference is that the ESP technique restricts the lengths of the blocks:
It first identifies so-called \intWort{meta-blocks} in $\textY$, and then further refines these meta-blocks into blocks of length~2 or~3.
The meta-blocks are created in the following 3-stage process~(see also \cref{figESP} for an example):
\begin{enumerate}[(1)]
	\item Identify maximal regions of repeated characters (i.e., maximal substrings of the form $c^\ell$ for $c\in\Sigma$ and $\ell\ge 2$). Such substrings form the \Type{1} meta-blocks.
\item Identify remaining substrings of length at least 2 (which must lie between two \Type{1} meta-blocks). Such substrings form the \Type{2} meta-blocks.
\item Every substring not yet covered by a meta-block consists of a single character and cannot have \Type{2} meta-blocks as its neighbors. 
	Such characters are fused with a neighboring meta-block. The meta-blocks emerging from this fusing are called \Type{M} (mixed). \label{itESPRuleMixed}
\end{enumerate}
Meta-blocks of \Type{1} and \Type{M} are collectively called \intWort{repeating meta-blocks}.
For~\ref{itESPRuleMixed}, we are free to choose whether a remaining character should be fused with its preceding or succeeding meta-block (both meta-blocks are repeating).
We stick to the following tie breaking rule%
\footnote{The original version~\cite{Cormode2007sed} prefers the left meta-block.}:

\RuleBox{%
\begin{Rules}[label={Rule~M:},ref={Rule~M}]
	\item
\label{stipulationM}
Fuse a remaining character $\textY[i]$ with its succeeding meta-block, or, if $i = \abs{\textY}$, with its preceding meta-block.
\end{Rules}
}

\begin{figure}[t]
	\centerline{%
		\Bild[scale=1.0]{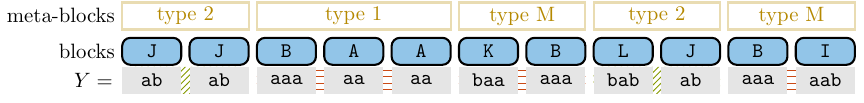}
	}
	\caption{ESP of the string $\textY = \texttt{ababaaaaaaabaaaaabababaaaaab}$. 
		The string is divided into blocks represented by the\C{~gray} rectangular boxes at the bottom.
		Each block gets assigned a new character represented by the capital letters in the rounded boxes.
		The white rectangular boxes on the top level represent the meta-blocks that group the blocks.
		The blocks are connected with\C{red} horizontal lines if they belong to a repeating meta-block,
		or by\C{green} diagonal lines if they belong to a \Type{2} meta-block.
	} %
	\label{figESP}
\end{figure}

Meta-blocks are further partitioned into \intWort{blocks}, each containing two or three characters from $\Sigma$.
Blocks inherit the type of the meta-block they are contained in.
How the blocks are partitioned depends on the type of the meta-block:
	
\begin{description}
	\item[Repeating meta-blocks.] 
		A repeating meta-block is partitioned greedily:
		create blocks of length three until there are at most four, but at least two characters left.
		If possible, create a single block of length two or three; otherwise (there are four characters remaining) create two blocks, each containing two characters.
	\item[Type-2 meta-blocks.]
		A \Type{2} meta-block $\mu$ is partitioned into blocks in $\Oh{\abs{\mu} \lg^* \sigma}$ time by the alphabet reduction (\cref{lemAlphabetReductionSpeed}).
		A block~$\beta$ generated by the alphabet reduction is determined by the characters
		$\textY[\max(\ibeg{\beta}-\lcontext, \ibeg{\mu})..\min(\iend{\beta}+\rcontext, \iend{\mu})]$ due to \cref{lemAlphabetReduction}.
		Given the number of reduction steps~$r$ in \cref{secAlphabetReduction}, the alphabet reduction does not create blocks 
		for the first~$r$ characters of each meta-block.
		The ESP technique blocks the first $r$ characters in the same way as a repeating meta-block. 
		The border case~$r=1$ (one character remaining) is treated by fusing the remaining character with the first block created by the alphabet reduction,
		possibly splitting this block in the case that its size is four.
\end{description}
A block is called \intWort{repetitive} if it contains the same characters.
All blocks of a \Type{1} meta-block and all blocks except at most the left- or rightmost block (these blocks can contain a fused character) in a \Type{M} meta-block are repetitive.

Let $\esp : \Sigma^* \rightarrow (\Sigma^2\cup \Sigma^3)^*$ denote the function 
that parses a string by the ESP technique. We regard the output of $\esp$ as a string of blocks.

\begin{figure}[t]
	\centerline{%
		\begin{tabular}{lr}
			\toprule
			\multicolumn{2}{c}{Common Dictionary} \\
			\midrule
			Rule & {$\generated{\cdot}$} \\ \cmidrule(lr){1-1} \cmidrule(lr){2-2}
			$\ESPname{$a^2b$}      \rightarrow \Char{aab}$ & ${\tt a}^2{\tt b}$ \\
			$\ESPname{$ab$}        \rightarrow \Char{ab}$         & ${\tt ab}$ \\
			$\ESPname{$ba^2$}      \rightarrow \Char{baa}$       & ${\tt ba}^2$ \\
			$\ESPname{$bab$}       \rightarrow \Char{bab}$        & ${\tt bab}$ \\
			$\ESPname{$ba$}        \rightarrow \Char{ba}$         & ${\tt ba}$ \\
			\bottomrule
		\end{tabular}
		\begin{tabular}{lr}
			\toprule
			\multicolumn{2}{c}{ESP Dictionary} \\
			\midrule
			Rule & {$\generated{\cdot}$} \\ \cmidrule(lr){1-1} \cmidrule(lr){2-2}
			$\ESPname{$a^2$}       \rightarrow \Char{aa}$        & ${\tt a}^2$ \\
			$\ESPname{$a^3$}       \rightarrow \Char{aaa}$        & ${\tt a}^3$ \\
			$\ESPname{$a^4$}       \rightarrow \ESPname{$a^2$}\ESPname{$a^2$}$         & ${\tt a}^4$ \\
			$\ESPname{$a^6$}       \rightarrow \ESPname{$a^3$}\ESPname{$a^3$}$         & ${\tt a}^6$ \\
			$\ESPname{$a^9$}       \rightarrow \ESPname{$a^3$}\ESPname{$a^3$}\ESPname{$a^3$}$        & ${\tt a}^9$ \\
			$\ESPname{$a^{12}$}    \rightarrow \ESPname{$a^6$}\ESPname{$a^6$}$         & ${\tt a}^{12}$ \\
			$\ESPname{$(ba)^2$}    \rightarrow \ESPname{$ba$}\ESPname{$ba$}$         & $({\tt ba})^2$ \\
			$\ESPname{$(ba)^3$}    \rightarrow \ESPname{$ba$}\ESPname{$ba$}\ESPname{$ba$}$        & $({\tt ba})^3$ \\
			$\ESPname{$a^4(ba)^2$} \rightarrow \ESPname{$a^4$}\ESPname{$(ba)^2$}$         & ${\tt a}^4({\tt ba})^2$ \\
			$\ESPname{$a^2(ba)^2$} \rightarrow \ESPname{$a^2$}\ESPname{$ba$}\ESPname{$ba$}$        & ${\tt a}^2({\tt ba})^2$ \\
			$\ESPname{$(ab)^3$}     \rightarrow \ESPname{$ab$}\ESPname{$ab$}\ESPname{$ab$}$ & $({\tt ab})^3$ \\
			\bottomrule
		\end{tabular}
		\begin{tabular}{lr}
			\toprule
			\multicolumn{2}{c}{HSP Dictionary} \\
			\midrule
			Rule & {$\generated{\cdot}$} \\ \cmidrule(lr){1-1} \cmidrule(lr){2-2}
			$\HSPname{$a^2$}       \rightarrow \Char{aa}$        & ${\tt a}^2$ \\
			$\HSPname{$a^3$}       \rightarrow \Char{aaa}$        & ${\tt a}^3$ \\
			$\HSPname{$a^4b$}      \rightarrow \HSPname{$a^3$}\HSPname{$ab$}$         & ${\tt a}^4{\tt b}$ \\
			$\HSPname{$a^5b$}      \rightarrow \HSPname{$a^3$}\HSPname{$a^2b$}$         & ${\tt a}^5{\tt b}$ \\
			\bottomrule
		\end{tabular}
	}%
	\caption{Names of the ESP (\cref{subsecESP}) and HSP (\cref{secHSP}) nodes stored in the global dictionary of our examples. 
	The common dictionary contains all names that are used by both ESP and HSP\@. Each name occurs on the left side only once across all dictionaries.}
	\label{tableNames}
\end{figure}

\subsection{Edit Sensitive Parsing Trees}\label{secDictAndNames}
Applying $\esp$ recursively on its output generates a \ac{CFG} as follows.
Let $\espnode{\textY}_0 := \textY$ be a string on an alphabet $\Sigma_0 := \Sigma$.
The output of $\espnode{\textY}_h := \espExp{h}{\textY} = \esp[\espExp{h-1}{\textY}]$ is a sequence of blocks, which belong to a new alphabet %
$\Sigma_{h}$ with $h \ge 1$. 
We call the elements of~$\Sigma_h$ with $h \ge 1$ \intWort{names}, and use the term \intWort{symbol} for an element that is a name or a character.
A block $\beta \in \Sigma_h$ contains a string of symbols with length two or three ($\in \Sigma_{h-1}^2 \cup \Sigma_{h-1}^3$).
We maintain an injective dictionary $\dic : \Sigma_h \rightarrow \Sigma_{h-1}^2 \cup \Sigma_{h-1}^3$ to map a block to its symbols.
The dictionary entries are of the form $\beta \rightarrow xy$ or $\beta \rightarrow xyz$, where $\beta \in \Sigma_h$ and $x,y,z\in \Sigma_{h-1}$.
We write $\dic(\textX) := \dic(\textX[1]) \cdots \dic(\textX[\abs{\textX}]) \in \Sigma_{h-1}^*$ for $\textX \in \Sigma_h^*$.
Each block on height~$h$ is contained in a meta-block~$\mu$ on height~$h-1$, which is equal to a substring $\espnode{\textY}_{h-1}[i..j] \in \Sigma_{h-1}^*$.
We call $\espnode{\textY}_{h-1}[i..j] \in \Sigma_{h-1}^*$ the \intWort{symbols} of~$\mu$.
Since each application of $\esp$ reduces the string length by at least one half,
there is an integer~$k$ with $k \le \lg \abs{\textY}$ such that $\espnode{\textY}_k = \esp(\espnode{\textY}_{k-1})$ is a single block $\tau \in \Sigma_k$.
We write $\idset := \bigcup_{1\le h \le k}\Sigma_h$ for the set of names in $\espnode{\textY}_1, \espnode{\textY}_2, \ldots, \espnode{\textY}_k$.
The \ac{CFG} for $\textY$ is represented by the non-terminals (i.e., the names) $\idset$, the terminals~$\Sigma_0$, the dictionary~$\dic$, and the start symbol~$\tau$.
This grammar exactly derives $\textY$.

Throughout this article, we comply with the convention to write symbols, i.e., characters ($\in \Sigma_0$) and names ($\in \Sigma_h$, $h \ge 1$), in typewriter font;
characters and names are written in lower and upper cases, respectively.
All examples use the same dictionary such that reappearing names are identical (see \cref{tableNames} for the used dictionary).
Names restricted to a particular figure can be written with Greek letters 
(a necessity due to the limitation of having only 26 letters in the English alphabet).

\begin{figure}[t]
	\centering{%
		\Bild{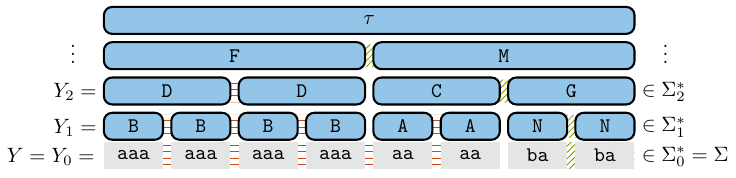}
	}
	\caption{The ESP tree of the string $\textY = {\tt aaaaaaaaaaaaaaaababa}$. 
		Like in \cref{figESP}, nodes belonging to the same meta-block are connected by\C{red} horizontal (repeating meta-block) or\C{green} diagonal (\Type{2} meta-block) lines. 
	}
	\label{figClassicESP}
\end{figure}

\begin{figure}[t]
	\centering{%
		\Bild[width=\textwidth]{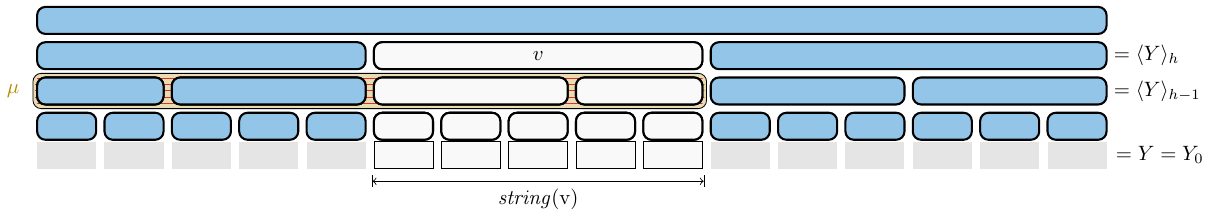}
	}
	\caption{$\espnode{\textY}_{h}$ with a highlighted node~$v$.
		The subtree rooted at $v$ is depicted by the white, rounded boxes.
		The generated substring~$\generated{v}$ of $v$ is the concatenation of the white rectangular blocks on the lowest level in the picture.
	The meta-block $\mu$, on which $v$ is built, is the rounded rectangle covering the children of~$v$ and all nodes connected by a horizontal hatching on height $h-1$.
	}
	\label{figESPtree}
\end{figure}

The \intWort{ESP tree} $\etInst[\textY]$ of a string $\textY$ is the derivation tree of the \ac{CFG} defined above.
	Its root node is the start symbol $\tau$.
	The nodes on height~$h$ are $\espnode{\textY}_h$ for each height~$h \ge 1$.
	In particular, the leaves are $\espnode{\textY}_1$.
	Each leaf refers to a substring in $\Sigma_0^2$ or $\Sigma_0^3$.
	The \intWort{generated substring} of a node $\espnode{\textY}_h[i]$ is the substring of $\textY$ generated by the symbol $\espnode{\textY}_h[i]$ 
	(applying the $h$-th iterate of $\dic$ to $\espnode{\textY}_h[i]$, yields a substring of~$\textY$, i.e., $\dic^{(h)}(\espnode{\textY}_h[i]) \in \Sigma^*$).
We denote the generated substring of~$\espnode{\textY}_h[i]$ by $\generated{\espnode{\textY}_h[i]}$.
For instance, in \cref{figClassicESP}, $\generated{\ESPname{$a^4(ba)^2$}} = {\tt aaaababa}$. 
A node $v$ on height~$h$ is said to be \intWort{built} on $\espnode{\textY}_{h-1}[b..e]$ iff 
$\espnode{\textY}_{h-1}[b..e]$ contains the children of $v$.
Like with blocks, nodes inherit the type of the meta-block on which they are built.
An overview of the definitions is given in \cref{figESPtree}.

\begin{figure}[t]
	\centering{%
		\Bild{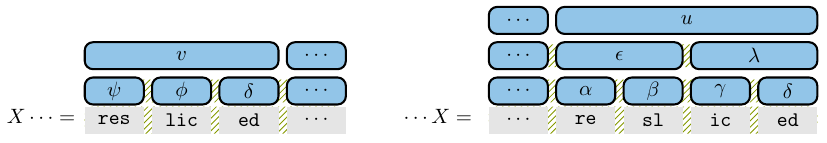}
	}
	\caption{Excerpts of \etInst[\textX\cdots] (\emph{left}) and \etInst[\cdots\textX] (\emph{right}) with
$\textX := \Char{resliced}$.
Under the assumption that  $\lg^*\sigma = 8$,
the common substring~$\textX$ can be blocked differently in both trees (depending on the characters preceding $\textX$ in the right figure).
	}
	\label{figEqualSubtreeDifferentName}
\end{figure}

In what follows, we present two shortcomings of the ESP trees.
The first is that nodes with different names can have the same generated substring, i.e., 
$\dic^{(h)} : \Sigma_h \rightarrow \Sigma^*_0$ is not injective for $h \ge 2$ in general.
The second is that it is not straight-forward to see which nodes of~\etInst[\textY] and~$\etInst[\textZ]$ are equal when $\textY$ is a substring of~$\textZ$.
Both cause problems when comparing subtrees of two nodes, which we later do for answering LCE queries.

Given two nodes~$u$ and~$v$, it holds that $\generated{u} = \generated{v}$ if their names are equal. 
However, the other way around is not true in general.
With $\generated{u} = \generated{v}$, it is not even assured that~$u$ and~$v$ are nodes on the same height.
Suppose that $\Sigma$ is a large alphabet with $\lg^*\sigma = 8$, 
and that $\textX := \Char{resliced}$ occurs in the text that we parse with ESP (see \cref{figEqualSubtreeDifferentName}).
We parse an occurrence of $\textX$ either (a) with the alphabet reduction if it is within a \Type{2} meta-block, 
or (b) greedily if it is at the beginning of a \Type{2} meta-block.
In the former case (a), we apply the alphabet reduction and end at a reduced alphabet with the characters \menge{0,1,2}.
Suppose that this occurrence of $\textX$ is reduced to the string in superscript of
\ruby{\tt r}{$0$}%
\ruby{\tt e}{$2$}%
\ruby{\tt s}{$1$}%
\ruby{\tt l}{$0$}%
\ruby{\tt i}{$1$}%
\ruby{\tt c}{$2$}%
\ruby{\tt e}{$1$}%
\ruby{\tt d}{$0$}.
Then ESP creates the four blocks {\tt re}$\mid${\tt sl}$\mid${\tt ic}$\mid${\tt ed}, whose boundaries are determined by the alphabet reduction.
Further suppose that an application of \esp{} creates two nodes of these blocks, which are put into a node~$u$ by an additional parse such that $\generated{u} = \textX$.
In the latter case (b), ESP creates the three blocks {\tt res}$\mid${\tt lic}$\mid${\tt ed} greedy. 
Suppose that an additional parse puts these blocks in a node~$v$ such that $\generated{v} = \textX$. 
Although $\generated{v} = \generated{u}$, the children of both nodes have different names, and therefore, both nodes cannot have the same name.

\begin{figure}[t]
	\centerline{%
		\Bild[scale=1.0]{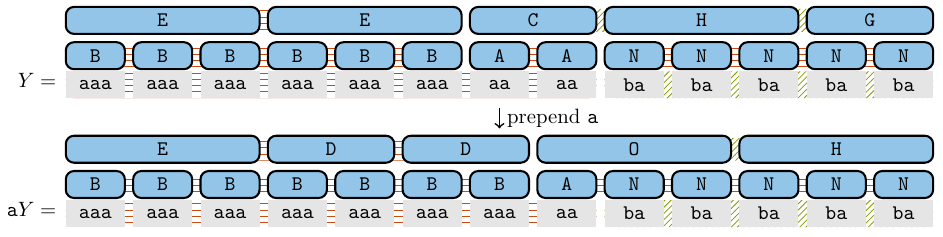}
}

\caption{Excerpt of \etInst[\textY] and \etInst[\Char{a} \textY] (higher nodes omitted), where $\textY = {\tt a}^{9k+4} ({\tt ba})^{3k-1} = {\tt a}^{22} ({\tt ba})^5$ for $k=2$.
	For all $k \ge 2$, there is a unique node in $\espnode{\textY}_2$ with the name \protect\ESPname{$a^4$}.
	This name does not appear in \etInst[\Char{a} \textY].
	} 
\label{figFragileBadSurrounding}
\end{figure}

The second shortcoming is that it is not clear how to transfer the property of the alphabet reduction described in \cref{lemAlphabetReduction} from blocks to nodes.
Given a substring $\textY$ of a string $\textZ$, the task is to analyze 
whether a node $\espnode{\textY}_h[i]$ is also present in the tree~$\etInst[\textZ]$, i.e., we analyze
changes of a node~$\espnode{\textY}_h[i]$ when \ac{preapp} characters to $\textY$.
For the sake of analysis, we distinguish the two terminologies \emph{block} and \emph{node}, although a node is represented by a block:
When we analyze a block in~$\esp[\textX] \in \Sigma_h^*$ for a string $\textX \in \Sigma_{h-1}^*$,
we let $\textX$ to be subject to \ac{preapp} characters of~$\Sigma_{h-1}$, 
whereas when we analyze a node $\espnode{\textY}_h[i]$ on a height~$h$ of \etInst[\textY] of a string $\textY \in \Sigma^*$,
we let $\textY$ to be subject to \ac{preapp} characters of~$\Sigma$.
In this terminology, a block in $\esp[\textX]$ is only determined by $\textX$,
whereas $\espnode{\textY}_h[i]$ is not only determined by 
$\esp^{(h-1)}(\textY) \in \Sigma_{h-1}^*$, but also by $\textY$ itself.
The difference is that a surrounded \Type{2} block of $\esp[\textX]$ cannot be changed by \ac{preapp} characters to $\textX$ due to \cref{lemAlphabetReduction}, 
whereas we fail to find integers~$\lcontext[,h]$ and $\rcontext[,h]$ such that 
a \Type{2} node on height~$h$ built on~$\espnode{\textY}_{h-1}[\lcontext[,h]..\rcontext[,h]]$ cannot be changed by \ac{preapp} characters to $\textY$.
That is because the names inside~$\espnode{\textY}_{h-1}$ and $\espnode{\Char{a} \textY}_{h-1}$ for $h \ge 2$ can differ at arbitrary positions.
This can be seen in the following example:
When parsing the string $\textY := {\tt a}^{9k+4} ({\tt ba})^{3k-1}$ with the names defined in \cref{tableNames}, we obtain
$\esp[{\esp[\textY]} ] = \esp[ {\ESPname{$a^3$}}^{3k} \ESPname{$a^2$}\ESPname{$a^2$}\ESPname{$ba$}^{3k-1} ] = \ESPname{$a^9$}^k \ESPname{$a^4$} \ESPname{$(ba)^3$}^{k-1} \ESPname{$(ba)^2$}$.
Let us focus on the unique occurrence of the name~\ESPname{$a^4$}, which is depicted in \cref{figFragileBadSurrounding} for $k=2$. 
On the one hand, there is a block representing the name~\ESPname{$a^4$} on height~two. 
This block is surrounded for a sufficiently large~$k$. Even for $k \ge 1$, it is easy to see that there is no way to change the name of this block
by \ac{preapp} characters to the string ${\ESPname{$a^3$}}^{3k} \ESPname{$a^2$}\ESPname{$a^2$}\ESPname{$ba$}^{3k-1}$.
On the other hand, there is a unique node in \etInst[\textY] with name~\ESPname{$a^4$} on height~two.
Regardless of the value of~$k$, prepending~${\tt a}$ to $\textY$ changes the name of~$v$:
$\esp[{\esp[{\tt a} \textY]} ] = \esp[ {\ESPname{$a^3$}}^{3k+1} \ESPname{$a^2$}\ESPname{$ba$}^{3k-1} ] = \ESPname{$a^9$}^{k-1} \ESPname{$a^6$} \ESPname{$a^6$} \ESPname{$a^2(ba)^2$} \ESPname{$(ba)^3$}^{k-1}$.
Nevertheless, we introduce the notion of surrounded nodes, since they are helpful to find rules that determine those nodes that cannot be changed by \ac{preapp} characters.

\begin{wrapfigure}{r}{17em}
	\vspace{-1em}
	\centering{%
		\Bild[width=0.35\textwidth]{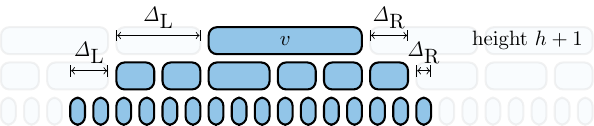}
	}%
\end{wrapfigure}
\block{Surrounded Nodes}
Analogously to blocks we classify nodes as surrounded when they are neighbored by sufficiently many nodes:
A leaf is called surrounded iff its generated substring is surrounded.
The local surrounding of a leaf is the local surrounding of the block represented by the leaf.
Given an internal node~$v$ on height~$h+1$ ($h \ge 1$) whose children are $\espnode{\textY}_h[\beta]$,
the \intWort{local surrounding} of~$v$ is the union of the nodes $\espnode{\textY}_h[\ibeg{\beta}-\lcontext..\iend{\beta}+\rcontext]$ 
and the local surrounding of each node in $\espnode{\textY}_h[\ibeg{\beta}-\lcontext..\iend{\beta}+\rcontext]$.
If all nodes in the local surrounding of~$v$ are surrounded,
we say that $v$ is \intWort{surrounded}.
Otherwise, we say that $v$ is \intWort{non-surrounded}.

\begin{lemma}\label{lemNumberSurroundedNodes}
There are at most $\lcontext + \rcontext$ many non-surrounded nodes on each height, summing up to $\Oh{\lg^*n \lg n}$ non-surrounded nodes in total.
\end{lemma}
\TextFigureBlock{0.8}{%
\begin{proof}
	We show that a node~$v$ on height~$h$ is surrounded if it has~$\lcontext$ preceding and $\rcontext$ succeeding nodes.
	This is clear on height~one by definition.
	Under the assumption that the claim holds for height~$h-1$,
	$v$'s preceding (resp.\ succeeding) nodes have at least~$2\lcontext$ (resp.\ $2\rcontext$) children in total, where at least the $\lcontext$ rightmost nodes (resp.\ $\rcontext$ leftmost nodes) are surrounded by the assumption. Hence, $v$ is surrounded.
\end{proof}
}{%
	\Bild[width=0.9\textwidth]{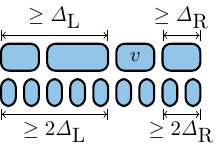}
}%
\vspace{0.5em}

The above example contrasting blocks and nodes reveals that the property for surrounded blocks as shown on the right side of \cref{figSurroundedBlock} 
cannot be transferred to surrounded nodes directly,
since a surrounded node depends not only on its local surrounding, but also on the nodes on which it its built.
Despite this discovery, we show that surrounded nodes can help us to create rules that are similar to \cref{lemAlphabetReduction}.

\subsection{Fragile and Stable Nodes in ESP Trees}\label{secChangingNodes}

\begin{figure}[t]
	\centerline{%
		\Bild[scale=1.0]{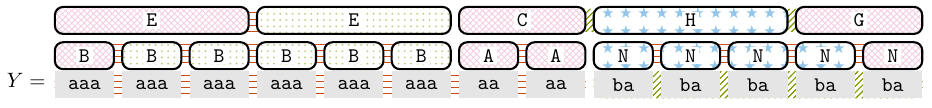}
}

\caption{\etInst[\textY] of \cref{figFragileBadSurrounding} with fragile, semi-stable and stable nodes highlighted.
	The fragile nodes are cross-hatched, the semi-stable nodes are dotted, and the stable nodes have stars attached. 
	The leftmost nodes of the tree change their names when prepending a \Char{b}.
	When prepending \Char{a}'s, we observe that the children of the node with name \protect\ESPname{$a^4$} change.
	Assuming that $\Sigma = \menge{\Char{a}, \Char{b}}$ (and hence $\abs{\Sigma} = 2$), only the rightmost node of the meta-block containing nodes with name \protect\ESPname{$ba$} is fragile.
	} 
\label{figFragileSemistableStable}
\end{figure}

We now analyze which nodes of $\etInst[\textY]$ are still present in $\etInst[\textX\textY\textZ]$ for all strings~$\textX$ and~$\textZ$.
A node $\espnode{\textY}_h[j]$ in $\etInst[\textY]$ at a height~$h$
is said to be \intWort{stable} iff, for \emph{all} strings $\textX$ and $\textZ$, 
there exists a node $\espnode{\textX\textY\textZ}_h[k]$ in $\etInst[\textX\textY\textZ]$
with 
the same name as $\espnode{\textY}_h[j]$
and
$\abs{\textX} + \sum_{i=1}^{j-1} \abs{ \generated{\espnode{\textY}_h[i]}} = \sum_{i=1}^{k-1} \abs{\generated{\espnode{\textX\textY\textZ}_h[i]}}$.
We also consider repeating nodes that are present with slight shifts;
a non-stable repeating node $\espnode{\textY}_h[j]$ in $\etInst[\textY]$
is said to be \intWort{semi-stable} iff, for \emph{all} strings $\textX$ and $\textZ$, 
there exists a node $\espnode{\textX\textY\textZ}_h[k]$ in $\etInst[\textX\textY\textZ]$
with 
the same name as $\espnode{\textY}_h[j]$
and
$
\sum_{i=1}^{k-1} \abs{\generated{\espnode{\textX\textY\textZ}_h[i]}} - \abs{\textS}
<
\abs{\textX} + \sum_{i=1}^{j-1} \abs{ \generated{\espnode{\textY}_h[i]}} <
\sum_{i=1}^{k-1} \abs{\generated{\espnode{\textX\textY\textZ}_h[i]}} + \abs{\textS}$,
where $\textS = \generated{\espnode{\textY}_h[j]} = \generated{\espnode{\textX\textY\textZ}_h[k]}$.

Nodes that are neither stable nor semi-stable are called \intWort{fragile}.
By definition, the children of the \SemiOrStable{} nodes (resp.\ fragile nodes) are also \SemiOrStable{} (resp.\ fragile).
\Cref{figFragileSemistableStable} shows an example, where all three types of nodes are highlighted.
The rest of this section studies how many fragile nodes exist in $\etInst[\textY]$.

As a warm-up, we first restrict the ESP tree construction on strings that are square-free.
A string $\textY$ is \intWort{square-free} iff there is no substring of~$\textY$ occurring consecutively twice.
Since a name of the ESP tree is determined by its generating substring, \etInst[\textY] cannot 
contain two consecutive occurrences of the same name on any height.
We conclude that \etInst[\textY] has no repeating nodes, i.e., it consists only of \Type{2} nodes.
When studying the stability of \Type{2} nodes, the following lemma is especially useful:

\begin{lemma}[{\cite[Lemma 8]{Cormode2007sed}}] \label{lemTypeTwoStable}
A \Type{2} node is stable if (a) it is surrounded and (b) its local surrounding does not contain a fragile node.
\end{lemma}
With \cref{lemTypeTwoStable} we immediately obtain:
\begin{lemma}\label{lemSquareFree}
	Given a square-free string~$\textY$, a fragile node of \etInst[\textY] is a non-surrounded node.
\end{lemma}
\begin{proof}
	According to \cref{lemTypeTwoStable}, we can bound the number of fragile nodes by the number of those nodes that do not satisfy the conditions in \cref{lemTypeTwoStable}.
Since \etInst[\textY] only contains \Type{2} nodes, 
we can show that a fragile node is non-surrounded inductively for all heights of the ESP tree:
Since leaves do not contain any nodes in their subtrees, surrounded leaves are stable due to \cref{lemAlphabetReduction}.
Therefore, the claim holds for $h=1$.
By definition, a node~$v$ on height~$h$ is surrounded
if its local surrounding~$S$ on height~$h-1$ is surrounded.
Given that the claim holds for~$h-1$,
a node in $S$ can only be fragile if it is not surrounded. This concludes that~$v$ can be fragile only if it is not surrounded.
\end{proof}
Combining \cref{lemSquareFree} with \cref{lemNumberSurroundedNodes} yields:
\begin{corollary}\label{corSquareFreeNumberFragileNodes}
	The number of fragile nodes of an ESP tree built on a square-free string of length~$n$ is \Oh{\lg^*n \lg n}.
	On each height, it contains \Oh{\lg^* n} fragile nodes.
\end{corollary}

\JO[%
In \cref{secESPLowerBound}, we show that \cref{corSquareFreeNumberFragileNodes} cannot be generalized for arbitrary strings.
There we show that the ESP technique changes \Om{\lg^2 n} nodes when changing a single character of a specific example string.
]{%
In the following we present a lower and an upper bound on the number of fragile nodes.
First, we show that the ESP technique can change \Om{\lg^2 n} nodes when changing a single character of the input string.
The idea is to give an example that contains a large number of \Type{M} meta-blocks in a specific constellation.
Remembering how the ESP technique parses its input, a remaining single symbol neighbored by two repeating meta-blocks is fused with one of them to form a \Type{M} meta-block.
Regardless of whether we favor to fuse a remaining symbol with either its preceding or succeeding (cf.~\ref{stipulationM}) repeating meta-block to form a \Type{M} meta-block,
for each of the two tie breaking rules, we propose an example string of length at most~$n$ whose ESP tree has \Om{\lg^2 n} fragile nodes.
These examples contradict Lemma~9 in~\cite{Cormode2007sed}, where it is claimed that there are \Oh{\lg^* n \lg n} fragile nodes in the ESP tree of a text of length~$n$.
}
\JJ[\section{A Lower Bound on the Number of Fragile ESP Tree Nodes}\label{secESPLowerBound}
Here, we present two examples reveiling that the ESP technique changes \Om{\lg^2 n} nodes when changing a single character.
Each example contains a large number of \Type{M} meta-blocks in a specific constellation.
Remembering how the ESP technique parses its input, a remaining single symbol neighbored by two repeating meta-blocks is fused with one of them to form a \Type{M} meta-block.
Regardless of whether we favor to fuse a remaining symbol with either its preceding or succeeding (cf.~\ref{stipulationM}) repeating meta-block to form a \Type{M} meta-block,
for each of the two tie breaking rules, we give an example string of length at most~$n$ whose ESP tree has \Om{\lg^2 n} fragile nodes.
These examples contradict Lemma~9 in~\cite{Cormode2007sed}, where it is claimed that there are \Oh{\lg^* n \lg n} fragile nodes in the ESP tree of a text of length~$n$.
]{\JO[\subsection{Fusing with the Preceding Repeating Meta-Block}]{}
\begin{textminipage}{0.45\linewidth}
	\JO{\block{Fusing with the preceding repeating meta-block}}
Consider a \Type{1} meta-block~$\mu$ whose rightmost node is fragile.
If the leftmost node of a repeating meta-block~$\nu$ is built on $\mu$'s rightmost node, then the rightmost node of~$\nu$ can also be fragile.
\end{textminipage}
\begin{minipage}{0.5\linewidth}
	\Bild{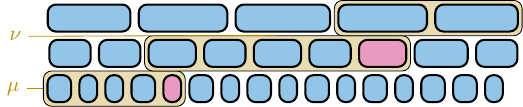}
\end{minipage}

Having this idea in mind, we build an example consisting of a chain of repeating meta-blocks, where the leftmost node of a repeating meta-block is built on the fragile rightmost node of a meta-block of one depth below (\C[shaded]{filled with red color} in the right picture).
The main idea is the following:
Each meta-block of this chain can be of arbitrary (but sufficiently long) length. 
Keeping in mind that changing the name of a node means that the names of its ancestors also have to change, we can create an example string whose ESP tree contains fragile nodes appearing on each height at arbitrary positions:

\begin{example}\label{exCounterA}
	Let $\Char{a}, \Char{b}$ and $\Char{c} \in \Sigma$ be three different characters.
The text $\textY := (X_0)^{3^\V{k}} (X_1)^{3^{\V{k}-1}} (X_2)^{3^{\V{k}-2}} \cdots (X_{\V{k}-1})^3$ 
with
$\V{k} := \gauss{\log_3 (n / \log_3 n)}$
has a length at most~$n$, and
its ESP tree 
has \Om{\lg^2 n} fragile nodes, where

\begin{minipage}[t]{.75\textwidth}
\raggedright
\[
X_0 := {\tt a}, \text{~and~}
X_{i} :=  
\begin{cases}
	X_{i-1}^2 {{\tt b}^3}^{i-1} & \text{if~} i \text{~is odd}, \\
	X_{i-1}^2 {{\tt c}^3}^{i-1} & \text{if~} i \text{~is even},
\end{cases}
\text{~for~} i = 1,\ldots,\V{k}.
\]
\end{minipage}%
\begin{minipage}[t]{.2\textwidth}
	\scriptsize
	$X_0 = {\tt a}$, \\
	$X_1 = {\tt aab}$,\\
	$X_2 = {\tt aabaab} {\tt c}^3$,\\
	$X_3 = X_2^2 {\tt b}^9$,\\
	$\phantom{X_4 }~\vdots$
\end{minipage}
\end{example}

To show that the claim in the example is correct, we insert a lemma showing the associativity of $\esp$ on a special class of strings:

\begin{lemma}\label{lemESPassociativity}
	Contrary to~\ref{stipulationM}, assume that we favor fusing a remaining character with its preceding meta-block to form a \Type{M} meta-block.
	Given a height~$h$ and two strings~$\textX, \textY$ that are either empty or have a length of at least~$2 \cdot 3^{h-1}$, 
	$\espExp{h}{\textX \Char{b}^{3^i} \textY} = \espExp{h}{\textX} \espExp{h}{\Char{b}^{3^i}} \espExp{h}{\textY}$ 
	if $i \ge h$, 
	$\Char{b}$ is neither a suffix of~$\textX$ nor a prefix of~$\textY$,
	and there is no prefix of~$\espExp{j}{\textY}$ of the form ${\tt c}{\tt d}^k$ for some characters~${\tt c, d} \in \Sigma_j$ with ${\tt c} \neq {\tt d}$, and integers $k,j$ with $k \ge 2$ and $0 \le j \le h-1$.
\end{lemma}
\begin{proof}
	The additional requirement for~$\textY$ is to ensure that the leftmost block of~$\espExp{j}{\textY}$ is not a non-repetitive \Type{M} block 
	that has been fused to its succeeding meta-block, only because it has no preceding meta-block.
	Regardless of which characters are prepended to~$\espExp{j-1}{\textY}$, the first character of such a block would form with its preceding 
	characters a new block.

	For~$h=1$, $\esp$ divides the string ${\textX \Char{b}^{3^i} \textY}$ into meta-blocks such that there is one \Type{1} meta-block~$\mu$
	that exactly contains the substring~$\Char{b}^{3^i}$. 
	That is because of the following:
	If $\textX$ (resp.\ $\textY$) is not the empty string, then it contains at least two characters.
	Since we favor fusing with the preceding meta-block, there is no chance that characters of $\textX$ can enter~$\mu$.
	Assume that~$\textY$ is not the empty string.
	Since the first block of $\esp[\textY]$ is neither a non-repetitive \Type{M} block nor a block starting with~\Char{b},
	it is not possible that characters of this block can enter~$\mu$.
	
Under the assumption that the claim holds for a given~$h-1 \ge 0$, we have
	\[ 
		\espExp{h}{\textX \Char{b}^{3^{i+1}} \textY} = \esp[\espExp{h-1}{\textX \Char{b}^{3^{i+1}} \textY}] = \esp[\espExp{h-1}{\textX} \espExp{h-1}{\Char{b}^{3^{i+1}}} \espExp{h-1}{\textY}].
\]
	The strings $\espExp{h}{\textX}$ and $\espExp{h}{\textY}$ are either empty or contain at least two characters. 
Since $i \ge h$, $\espExp{h-1}{\Char{b}^{3^i}}$ is the repetition of the same character. 
This repetition has a length of at least three such that we can apply the shown associativity for $h=1$ to show the claim.
\end{proof}

\begin{proof}[Proof of \cref{exCounterA}]
	We start with determining the length of~$\textY$.
	Since $\abs{X_0} = 3^0$, under the assumption that $\abs{X_i} = 3^i$, we obtain that $\abs{X_{i+1}} = 2 \abs{X_i} + 3^{i} = 3^{i+1}$.
	Therefore, $\abs{X_i^{3^{\V{k}-i}}} = 3^\V{k}$ for all $i=0,\ldots,\V{k}-1$.
	We conclude that the length of~$\textY$ is at most~$n$, since $\abs{\textY} = \V{k} 3^\V{k} \le n \log_3(n/\log_3 n) / \log_3 n \le n$.
	
	We now show that each substring~$X_i$ of~$\textY$ is the generated substring of a node~$x_i$ of~$\etInst[\textY]$ on height~$i$
	whose subtree is equal to the perfect ternary subtree~$T_i := \etInst[X_i]$, for $i=1,\ldots,\V{k}-1$.
	This is true for $i = 1,2,3$, as can be seen in \cref{figHanreiAEx}.
	For the general case, we adapt the associativity shown for $\esp$ in \cref{lemESPassociativity} to the string~$X_i$:
	
	\begin{subclaim}
		For every $i$ with $0 \le i \le \V{k}-2$ it holds that
		\begin{enumerate}[(I)]
			\item $\abs{\espExp{i+1}{X_{i+1}}} = 1$, \label{itHanreiAPropOne}
	\item $\espExp{\V{h}}{X_{i+1}} = \espExp{\V{h}}{X_i X_i {\tt d}_i^{3^i}} = 
		\espExp{\V{h}}{X_i X_i} \espExp{\V{h}}{{\tt d}_i^{3^i}} =
		\espExp{\V{h}}{X_i} \espExp{\V{h}}{X_i} \espExp{\V{h}}{{\tt d}_i^{3^i}}$, and \label{itHanreiAPropAssoc}
	\item $\espExp{\V{h}}{X_{i+1}}$ starts with a repetition of a character, \label{itHanreiAPropRep}
		\end{enumerate}
		for every~$\V{h}$ with $0 \le \V{h} \le i$ , where ${\tt d}_i$ is a character with ${\tt d}_i = {\tt b}$ if $i$ is even, otherwise ${\tt d}_i = {\tt c}$. 
	\end{subclaim}

	\begin{subproof}
	For $i=0$ we have
		\begin{enumerate}[(I)]
			\item $\abs{\espExp{1}{X_{1}}} = \abs{\esp[ {\tt aab} ]} = 1$ ({\tt aab} is put in a \Type{M} meta-block having exactly one block),
			\item $\espExp{0}{X_1} = X_1$, and
			\item $X_1 = \Char{aab}$ starts with a repetition of the character~\Char{a}.
		\end{enumerate}

		Under the assumption that the claim holds for an integer~$i$, we conclude that it holds for~$i+1$ due to
			\begin{align*}
				{\espExp{\V{h}}{X_{i+2}}}
		&= {\espExp{\V{h}}{X_{i+1} X_{i+1} {\tt d}_{i+1}^{3^{i+1}}}} \\
		&= {\espExp{\V{h}}{X_i X_i {\tt d}_i^{3^{i}} X_i X_i {\tt d}_i^{3^{i}} {\tt d}_{i+1}^{3^{i+1}}}} \\
				\phantom{a}^{(\text{\cref{lemESPassociativity}},{\tt d}_i \neq {\tt d}_{i+1} )}  
				&= {\espExp{\V{h}}{X_i X_i {\tt d}_i^{3^{i}} X_i X_i {\tt d}_i^{3^{i}}}} {\espExp{\V{h}}{{\tt d}_{i+1}^{3^{i+1}}}} \\
				\phantom{a}^{\text{(\cref{lemESPassociativity},\ref{itHanreiAPropOne} or~\ref{itHanreiAPropRep})}} 
				&= {\espExp{\V{h}}{X_i X_i}} {\espExp{\V{h}}{{\tt d}_i^{3^{i}}}} {\espExp{\V{h}}{X_i X_i}} {\espExp{\V{h}}{{\tt d}_i^{3^{i}}}} {\espExp{\V{h}}{{\tt d}_{i+1}^{3^{i+1}}}} \\
				\phantom{a}^{\text{(\cref{lemESPassociativity},\ref{itHanreiAPropOne} or~\ref{itHanreiAPropRep})}} 
				&= {\espExp{\V{h}}{X_i X_i {\tt d}_i^{3^{i}}}} {\espExp{\V{h}}{X_i X_i {\tt d}_i^{3^{i}}}} {\espExp{\V{h}}{{\tt d}_{i+1}^{3^{i+1}}}} \\
			 &= {\espExp{\V{h}}{X_{i+1}}} {\espExp{\V{h}}{X_{i+1}}} {\espExp{\V{h}}{{\tt d}_{i+1}^{3^{i+1}}}} \\
			\end{align*}
		for $1 \le \V{h} \le i$.
		The conditions of \cref{lemESPassociativity} hold because 
		${\tt d}_i$ is neither a prefix nor a suffix of $X_i$, ${\tt d}_i \neq {\tt d}_{i+1}$,
			$\abs{X_i X_i} = 2 \cdot 3^\V{i}$, 
			and ${\espExp{\V{h}}{X_i X_i}}$ starts with a repetition of a character due to
			\[
		\begin{cases}
			\text{\ref{itHanreiAPropRep}} & \text{for~} \V{h} < i \text{, or due to} \\
			{\espExp{i}{X_i X_i }} =^{\ref{itHanreiAPropAssoc}} {\espExp{i}{X_i}} {\espExp{i}{X_i}} \text{~and~\ref{itHanreiAPropOne}} & \text{for~} \V{h} = i.
		\end{cases}	
	\]
			For $\V{h} = i+1$ we use that~\ref{itHanreiAPropOne} holds for $X_i$, $\abs{\espExp{i}{{\tt d}_{i}^{3^{i}}}} = 1$, and
			{\espExp{i}{{\tt d}_{i+1}^{3^{i+1}}}} is a repetition of length~$3$ of the same character,
			to obtain
			\begin{align*}
				{\espExp{i+1}{X_{i+2}}} 
				&= {\esp[{\espExp{i}{X_{i+2}} }]} \\
				&= {\esp[{\espExp{i}{X_i X_i}} {\espExp{i}{{\tt d}_i^{3^{i}}}} {\espExp{i}{X_i X_i}} {\espExp{i}{{\tt d}_i^{3^{i}}}} {\espExp{i}{{\tt d}_{i+1}^{3^{i+1}}}} ]} \\
				\phantom{a}^{(\text{\cref{lemESPassociativity}})}
				&= {\esp[{\espExp{i}{X_i X_i}} {\espExp{i}{{\tt d}_i^{3^{i}}}} {\espExp{i}{X_i X_i}} {\espExp{i}{{\tt d}_i^{3^{i}}}} ]} {\esp[ {\espExp{i}{{\tt d}_{i+1}^{3^{i+1}}}} ]} \\
				\phantom{a}^{(\text{evaluate and reformulate})} 
				&= {\esp[{\espExp{i}{X_i X_i}} {\espExp{i}{{\tt d}_i^{3^{i}}}}]} {\esp[{\espExp{i}{X_i X_i}} {\espExp{i}{{\tt d}_i^{3^{i}}}}]} {\esp[{\espExp{i}{{\tt d}_{i+1}^{3^{i+1}}}} ]},
			\end{align*}
			where we used that another application of \esp{} puts ${\espExp{i}{X_i X_i}} {\espExp{i}{{\tt d}_i^{3^{i}}}}$ into a single \Type{M} meta-block of length three, 
			and that $\Char{d}_i$ is neither a prefix nor a suffix of~$X_i$.
			This concludes~\ref{itHanreiAPropAssoc}. A consequence is~\ref{itHanreiAPropRep}:
			For $\V{h} \le i$ we have
				${\espExp{\V{h}}{X_{i+2}}} = {\espExp{\V{h}}{X_{i+1}}} {\espExp{\V{h}}{X_{i+1}}} {\espExp{\V{h}}{{\tt d}_{i+1}^{3^{i+1}}}}$,
				and ${\espExp{\V{h}}{X_{i+1}}}$ starts with a repetition of a character according to our assumption.
				For $\V{h} = i+1$ we have 
				${\espExp{i+1}{X_{i+2}}} = \esp[ \espExp{i}{X_i} \espExp{i}{X_i} \espExp{i}{d_i^{3^i}} \espExp{i}{X_i} \espExp{i}{X_i} \espExp{i}{d_i^{3^i}} \espExp{i}{{\tt d}_{i+1}^{3^{i+1}}} ]$.
				Due to~\ref{itHanreiAPropOne}, $\abs{\espExp{i}{X_i}} = \abs{\espExp{i}{{\tt d}_i^{3^i}}} = 1$; hence the last application of $\esp$ creates three blocks, where each of the first two represents the string~$\espExp{i}{X_i} \espExp{i}{X_i} \espExp{i}{d_i^{3^i}}$ of length three.
				Another application of $\esp$ yields~\ref{itHanreiAPropOne}.
	\end{subproof}

		Let $b_i$ and $c_i$ denote the names of the roots of \etInst[{\tt b}^{3^i}] and of \etInst[{\tt c}^{3^i}], respectively.
		Set $d_i := b_i$ if $i$ is even, otherwise $d_i := c_i$.
		Then $\espnode{X_{i+1}}_{i+1}  = x_{i+1}$ due to Sub-Claim~\ref{itHanreiAPropOne}, and $\espnode{X_{i+1}}_{i} = x_{i} x_{i} d_{i}$ due to Sub-Claim~\ref{itHanreiAPropAssoc}.
		Consequently, 
		\begin{equation}\label{eqESPhanreiXi}
			\esp[ (\espnode{X_{i+1}}_{i})^{3^{\V{k}-i-1}} ] = \esp[ (x_{i} x_{i} d_{i})^{3^{\V{k}-i-1}} ] = (\esp[x_{i} x_{i} d_{i}])^{3^{\V{k}-i-1}} = x_{i+1}^{3^{\V{k}-i-1}}.
		\end{equation}
	\begin{minipage}{0.55\linewidth}
		This means that $\espnode{X_{i}}_{\V{h}}^{3^{\V{k}-i}} =\espnode{X_{i}^{3^{\V{k}-i}}}_{\V{h}}$ is a repetition of length~$3^{\V{k}-\V{h}}$ consisting of the same name, for every height~$\V{h} = i, \ldots, \V{k}$.
		We conclude that $T_i := \etInst[(X_i)^{3^{\V{k}-i}}]$ is a perfect ternary tree.
	\end{minipage}
	\begin{minipage}{0.4\linewidth}
		\Bild[width=\linewidth]{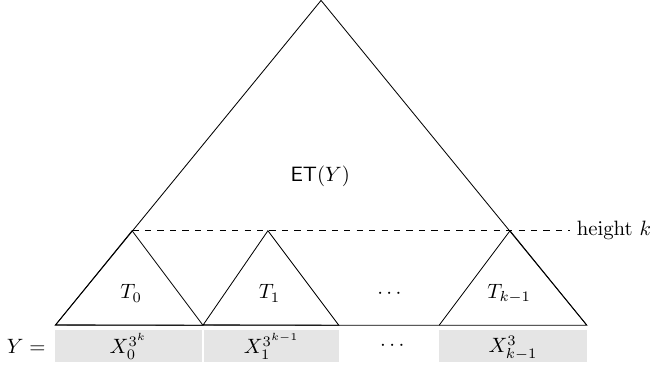}
	\end{minipage}

	Finally, we show that $\espExp{\V{h}}{\textY} = \espExp{\V{h}}{X_1^{3^{\V{k}}}} \cdots \espExp{\V{h}}{X_{\V{k}-1}^{3^1}}$
for each height~\V{h} with $1 \le \V{h} \le \V{k}$.
On the one hand, we have
\begin{align}\label{eqESPhanreiAssoc}
	\begin{split}
	{\espExp{\V{h}}{X_i^{3^{\V{k}-i}} X_{i+1}^{3^{\V{k}-i-1}}}}
	&= {\espExp{\V{h}}{X_i^{3^{\V{k}-i-1}} X_{i-1} X_{i-1} {\tt d}_{i-1}^{3^{i-1}} X_{i+1}^{3^{\V{k}-i-1}}}} \\
	\phantom{a}^{\text{\ref{itHanreiAPropRep}} \text{~with~} 0 \le i \le h-2}
	&= {\espExp{\V{h}}{X_i^{3^{\V{k}-i-1}} X_{i-1} X_{i-1}}} {\espExp{\V{h}}{{\tt d}_{i-1}^{3^{i-1}}}} {\espExp{\V{h}}{X_{i+1}^{3^{\V{k}-i-1}}}} \\
	&= {\espExp{\V{h}}{X_i^{3^{\V{k}-i-1}} X_{i-1} X_{i-1} {\tt d}_{i-1}^{3^{i-1}}}} {\espExp{\V{h}}{X_{i+1}^{3^{\V{k}-i-1}}}} \\
	&= {\espExp{\V{h}}{X_i^{3^{\V{k}-i}}}} {\espExp{\V{h}}{X_{i+1}^{3^{\V{k}-i-1}}}}
\end{split}
\end{align}
for $1 \le \V{h} \le i-1$ due to \cref{lemESPassociativity}.
On the other hand, we have
\begin{align}\label{eqESPhanreiAssocHigh}
	\begin{split}
	{\espExp{\V{h}}{X_i^{3^{\V{k}-i}} X_{i+1}^{3^{\V{k}-i-1}}}}
	&= {\espExp{\V{h-i+1}}{\espExp{\V{i-1}}{X_i^{3^{\V{k}-i}} X_{i+1}^{3^{\V{k}-i-1}}} }} \\
	\phantom{a}^{\text{\cref{eqESPhanreiAssoc}}} 	&= {\espExp{\V{h-i+1}}{\espExp{\V{i-1}}{X_i^{3^{\V{k}-i}}} \espExp{\V{i-1}}{X_{i+1}^{3^{\V{k}-i-1}}} }} \\
	\phantom{a}^{\text{\cref{eqESPhanreiXi}}}    &= {\espExp{\V{h-i}}{{\esp[{(x_{i-1} x_{i-1} d_{i-1})^{3^{\V{k}-i}} (x_{i-1} x_{i-1} d_{i-1} x_{i-1} x_{i-1} d_{i-1} \espnode{{\tt d}_i^{3^i}}_{i-1} )^{3^{\V{k}-i-1}} }]} }} \\
	\phantom{a}^{(\text{apply \esp{}})} &= {\espExp{\V{h-i}}{x_i^{3^{\V{k}-i}} (x_i x_i d_i)^{3^{\V{k}-i-1}} }} \\
	&= {\espExp{\V{h-i-1}}{{\esp[x_i^{3^{\V{k}-i}} x_i x_i d_i]} {\esp[(x_i x_i d_i)^{3^{\V{k}-i-2}}]} }} \\
	\phantom{a}^{(\text{evaluate and reformulate})} 
	&= {\espExp{\V{h-i-1}}{{\esp[x_i^{3^{\V{k}-i}}]} {\esp[(x_i x_i d_i)^{3^{\V{k}-i-1}}]} }} \\
	\phantom{a}^{\text{\cref{eqESPhanreiXi}}} &= {\espExp{\V{h-i-1}}{{\esp[x_{i}^{3^{\V{k}-i}}]} {\esp[x_{i+1}^{3^{\V{k}-i-1}}]} }} \\
	\phantom{a}^{\text{(\cref{lemESPassociativity})}} &= {\espExp{\V{h-i}}{x_{i}^{3^{\V{k}-i}}}} {\espExp{\V{h-i}}{x_{i+1}^{3^{\V{k}-i-1}}}}
\end{split}
\end{align}
for $i \le \V{h} \le \V{k}$.
It is easy to extend the pairwise associativity $X_i^{3^{\V{k}-i}} X_{i+1}^{3^{\V{k}-i-1}}$ for each~$i$ with $0 \le i \le k-2$ to $X_1^{3^{\V{k}}} \cdots X_{\V{k}-1}^{3^1}$.
This concludes that the root of~$T_i$ has the same name as the $i$-th leftmost node of $\etInst[\textY]$ on height~$\V{k}$.
\Cref{figMetablockLeft}(left) shows an excerpt of $T_{i}$ and $T_{i+1}$.
The crucial step in \cref{eqESPhanreiAssocHigh} is the re-formulation of the parsing
\begin{equation}\label{eqESPhanreiMetablockMu}
{\espExp{\V{h-i-1}}{%
	\underbrace{\esp[x_i^{3^{\V{k}-i}} x_i x_i d_i]}_{\text{belongs to~} T_i} 
\underbrace{\esp[(x_i x_i d_i)^{3^{\V{k}-i-2}}]}_{\text{belongs to~} T_{i+1}} }} \\
= {\espExp{\V{h-i-1}}{{\esp[x_i^{3^{\V{k}-i}}]} \underbrace{\esp[(x_i x_i d_i)^{3^{\V{k}-i-1}}]}_{=: \mu_{i+1}} }}
\end{equation}
showing that there is a \Type{1} meta-block~$\mu_{i+1}$ covering all nodes of~$T_{i+1}$ and the rightmost node of~$T_{i}$, on height~$i+1$.
This meta-block is a repetition of the character~$\esp[x_i x_i d_i] = x_{i+1} \in \Sigma_{h+1}$.

Given that $\mu_0$ is the first \Type{1} meta-block of~$\esp[\textY]$ (covering the prefix $X_0^{3^h+2}$),
we now examine what happens with~$\mu_i$ for each~$i$ with $0 \le i \le h-1$ when removing the first {\tt a} from $\textY$. 
Let us call the shortened string~$\textY'$, i.e., $\textY = {\tt a} \textY'$.
On removing the first {\tt a} from $\textY$, we claim that the meta-block~$\mu_i$ contains one character $x_i$ less, 
for every $i$ with $0 \le i \le h-1$ (cf.~\cref{figMetablockLeft}
showing the difference between $\espnode{\textY}_{i}$ and $\espnode{\textY'}_{i}$ on height~$i$ with $0 \le i \le \V{k}-1$):
For $\mu_0$, this is trivial.
For an $i \ge 0$, 
focus on the substring $X_i^{3^{\V{k}-i}} X_{i+1}^{3^{k-i-1}}$ of~$\textY$:
We have
	\begin{align*}
		{\esp[ \espnode{X_i^{3^{\V{k}-i}} X_{i+1}^{3^{k-i-1}}}_i ]} 
		&= {\esp[ x_i^{3^{\V{k}-i}} (x_i x_i d_i)^{3^{\V{k}-i-1}} ]} 
		= {\esp[\underbrace{x_i^{3^{\V{k}-i}}}_{\text{suffix of~} \mu_i}]} \underbrace{\esp[ (x_i x_i d_i)^{3^{\V{k}-i-1}} ]}_{= \mu_{i+1}} \\
		&= {\esp[ x_i^{3^{\V{k}-i}}]} x_{i+1}^{3^{\V{k}-i-1}}
	\end{align*}
due to \cref{eqESPhanreiAssocHigh}.
Under the assumption that removing the first character~\Char{a} from~$\textY$ causes $\mu_i$ to shrink by one character~$x_i \in \Sigma_i$, 
we get
	\begin{align*}
		{\esp[x_i^{3^{\V{k}-i}-1} (x_i x_i d_i)^{3^{\V{k}-i-1}} ]}
		&= {\esp[x_i^{3^{\V{k}-i}} x_i d_i]} {\esp[ (x_i x_i d_i)^{3^{\V{k}-i-1}-1} ]}  \\
		&= {\esp[x_i^{3^{\V{k}-i}} x_i d_i]} x_{i+1}^{3^{\V{k}-i-1}-1} \\
		&\neq {\esp[ x_i^{3^{\V{k}-i-1}}]} x_{i+1}^{3^{\V{k}-i-1}}.
	\end{align*}

We observe that the length of $\mu_i$ is decremented by one, causing
the name of its rightmost block to change, which is the leftmost node of~$T_{i+1}$ on height~$i+1$, and the first character of~$\mu_{i+1}$.
Due to the tie breaking rule, this block gets fused with its preceding meta-block at height~$i+1$, decrementing the length of its succeeding meta-block~$\mu_{i+1}$ by one (and hence, this process repeats for all~$i = 0,\ldots,\V{k}-2$).
This means that the leftmost node on height $i$ of $T_i$ changes, for $1 \le i \le \V{k}-1$. Each of these nodes receives a new name such that 
it is fused with its preceding \Type{1} meta-block to form a \Type{M} meta-block. Since changing a node on height~$i$ changes all its ancestors, 
		at least $\V{k} - i$ nodes are changed in $T_i$.
In total, at least $\V{k} + (\V{k}-1) + (\V{k}-2) + \cdots + 2 = (\V{k}^2 + \V{k})/2 - 1$ nodes are changed.
Hence there is a lower bound of 
$\Om{\V{k}^2} = \Om{\log_3^2(n / \log_3 n)} = \Om{\lg^2 n}$ fragile nodes.
\end{proof}

Note that the later introduced HSP technique (see \cref{secHSP}) with the same tie breaking rule also produces \Om{\lg^2 n} fragile nodes in this example.

\begin{figure}[t]
	\centering{%
		\Bild[scale=1.0]{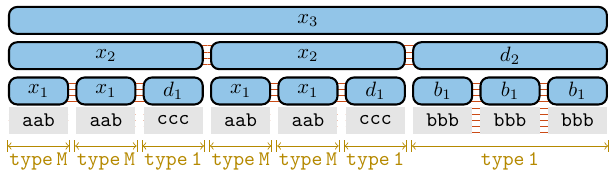}
	}
	\caption{\etInst[X_3] as defined in \cref{exCounterA}. The subtree of each node with name~$x_i$ is equal to \etInst[X_i].}
	\label{figHanreiAEx}
\end{figure}

\begin{figure}[t]
	\centering{%
		\Bild[width=0.45\textwidth,valign=t]{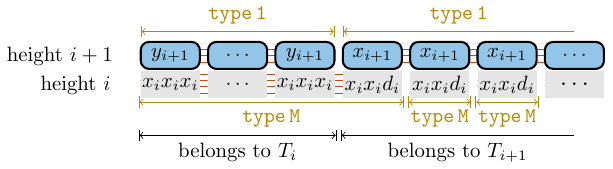}
		\hfill
	\Bild[width=0.45\textwidth,valign=t]{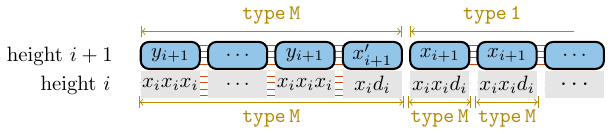}
	}
	\caption{Differences between $\etInst[\textY]$ (\emph{left}) and $\etInst[\textY']$ (\emph{right}) on the heights~$i$ and~$i+1$, where
		$\textY = {\tt a}^{3^\V{k}} \tuple{\tt a^{\text{2}}b}^{3^{\V{k}-1}} \tuple{ ({\tt a}^2 {\tt b})^2 {\tt c}^3}^{3^{\V{k}-2}} \cdots$ and $\textY = {\tt a} \textY'$
	(defined in \cref{exCounterA}). The names~$y_{i+1}$ and $x'_{i+1}$ are only used in this figure. 
	}
	\label{figMetablockLeft}
\end{figure}

\JO[\subsection{Fusing with the Succeeding Repeating Meta-Block}]{\block{Fusing with the succeeding repeating meta-block}}
The idea is similar to the previous example. In particular, we introduce a corollary of \cref{lemESPassociativity}:

\begin{corollary}\label{lemESPassociativityLeft}
	Given a height~$h$ and a string~$\textY$ that is either empty or has a length of at least~$2 \cdot 3^{h-1}$,
	$\espExp{h}{\textX \textY} = \espExp{h}{\textX} \espExp{h}{\textY}$
	if $\Char{a}$ is not a prefix of~$\textY$,
	where $\textX = {\tt b}^{3^i} {\tt a}^{3^{j}}$
	with $i+j \ge h$, and $\Char{a}, \Char{b} \in \Sigma$ with $\Char{a} \neq \Char{b}$.
\end{corollary}

In the following example, we build a text whose ESP tree has a specific \Type{M} meta-block on each height that we want to change.
Given a \Type{M} meta-block~$\mu$ that emerged from prepending a character to a \Type{1} meta-block,
we can create a new meta-block by prepending another character such that it precedes $\mu$ and absorbs $\mu$'s first character ($\mu$ then returns to be a \Type{1} meta-block).
We can arrange the \Type{M} meta-blocks such that prepending a character to the text changes a \Type{M} meta-block on each height:
\begin{example}\label{exCounterB}
	Let $\V{k} = \gauss{\log_3(n / \log_3 n)}$ be a natural number, and ${\tt a, b} \in \Sigma$.
Define $\textY := X_0 X_1 \cdots X_\V{k}$ with $X_i := {\tt b}^{3^i} {\tt a}^{3^\V{k} - 3^i}$, for $0 \le i \le \V{k}-1$.
Then $\abs{\textY} \le n$, and $\etInst[\textY]$ has \Om{\lg^2 n} fragile nodes.
\end{example}
\begin{proof}
	Given an integer~$i$ with $0 \le i \le \V{k}-1$, we have $\abs{X_i} = 3^\V{k}$ and $\abs{Y} = \V{k} 3^\V{k} \le n$.
	\Cref{lemESPassociativityLeft} yields
	$\espExp{h}{X_i} = \espExp{h}{ {\tt b}^{3^i}} \espExp{h}{{\tt a}^{3^k-3^i}}$
	for all heights~$h$ with~$0 \le h \le i$, since
	$3^k - 3^i \ge 3^k - 3^{k-1} = 2 \cdot 3^{k-1}$.
	Let $a_i := \espnode{{\tt a}^{3^\V{k}}}_i[1]$ and $b_i := \espnode{{\tt b}^{3^\V{k}}}_i[1]$ be the nodes on height~$i$ with $0 \le i \le \V{k}$ and $\generated{a_i} = {\tt a}^{3^i}$ or respectively $\generated{b_i} = {\tt b}^{3^i}$ ($a_0 := {\tt a}$, $b_0 := {\tt b}$),
	$\esp[ \espExp{h-1}{X_i} ]$ partitions its input~$\espExp{h-1}{X_i}$ into two meta-blocks: 
	a \Type{1} meta-block containing all $b_i$'s, and a subsequent \Type{1} meta-block containing all $a_i$'s.
	All blocks of these two meta-blocks contain three characters, since each meta-block has a length that is equal to a power of three.
For the upper heights we get
\begin{equation}\label{eqESPhanreiBHigherLevel}
		\espExp{h+i}{X_i} = \espExp{h}{\overbrace{\underbrace{\espExp{i}{{\tt b}^{3^i}}}_{= b_i} \underbrace{\espExp{i}{{\tt a}^{3^k-3^i}}}_{= a_i^{3^{k-i}-1}}}^{\abs{\cdot} = 3^{k-i}} } \text{~for~} 0 \le h+i \le \V{k-1}.
\end{equation}
	Hence, $\espExp{h+i}{X_i}$ consists of exactly one \Type{M} meta-block, which has the length~$3^{k-h-i}$,
	and each block contains three characters.
	We conclude that the tree $T_i := \etInst[X_i]$ is a perfect ternary tree, for $0 \le i \le \V{k}-1$.
	Since $\abs{\espExp{h}{X_i}} = 3^{k-h}$ for all $i,h$ with $0 \le i \le \V{k}-1$ and $0 \le h \le \V{k}$, with \cref{lemESPassociativityLeft} it is easy to see that 
	$\espExp{h}{\textY} = \espExp{h}{X_1 \cdots X_{\V{k}-1}} = \espExp{h}{X_1} \cdots \espExp{h}{X_{\V{k}-1}}$ for all $0 \le h \le \V{k}$.
	A conclusion is that $X_i$ is the generated substring of the $i$-th leftmost node \etInst[\textY] on height~$\V{k}$ whose name 
	is equal to the name of the root of~$T_i$, for $0 \le i \le \V{k}-1$.

	For the proof, we prepend an {\tt a} to $\textY$ and call the new string~$\textY'$, i.e., $\textY' = {\tt a} \textY$.
Our analysis of the difference between \etInst[\textY] and \etInst[\textY'] focuses on the unique meta-block at height~$i$ of~$T_i$:
From \cref{eqESPhanreiBHigherLevel} with $h=0$, we observe that there is a single meta-block~$\mu_i$ at height~$i$ of~$T_i$, and this meta-block is a \Type{M} meta-block (cf. \cref{figHanreiBEx}(right)).
Our claim is that prepending \Char{a} to~$\textY$ changes the blocks of the borders of every~$\mu_i$ ($0 \le i \le \V{k}-1$):
The prepended \Char{a} forms a \Type{2} meta-block with the first character of~$X_0$ 
by ``stealing'' the first character from $\mu_0$, and this character is a $b_0 = {\tt b}$.
Assume that~$\mu_i$ ($0 \le i \le \V{k}-1$) looses its first character (i.e., $b_i$).
By relinquishing this character, $\mu_i$ becomes a \Type{1} meta-block, consisting only of $a_i$'s.
The last two $a_i$'s contained in~$\mu$ are grouped into a block~$a'_{i+1}$ of length \emph{two},
where $a'_{i+1} := \espnode{a_i a_i}_{1}[1]$ is the name of the root node of \etInst[a_i a_i].
Every newly appearing node~$a'_{i+1}$ gets combined with its right-adjacent node~$b_{i+1}$ to form a new \Type{2} meta-block.
The used $b_{i+1}$ is stolen from $\mu_{i+1}$, and hence we observe an iterative process of stealing the first character~$b_{i+1}$ from~$\mu_{i+1}$
for each height~$i=0,\ldots,\V{k}-2$.
\Cref{figMetablockRight} visualizes this observation on the lowest two heights.
 
This can be inductively proven for each even integer~$i$ with $0 \le i \le h-2$.
By \cref{eqESPhanreiBHigherLevel}, we know that $\espnode{X_i}_i = b_i a_i^{3^{k-i}-1}$ and $\espnode{X_{i+1}}_i = b_i^3 a_i^{3^{k-i}-3}$.
Then
	\begin{align*}
			  {\esp[{\esp[{\espnode{X_i}_i \espnode{X_{i+1}}_i} ]} ]}
			  &= {\esp[{\esp[ b_i a_i^{3^{k-i}-1} b_i^3 a_i^{3^{k-i}-3}]} ]} \\
				\phantom{a}^{\text{(\cref{lemESPassociativityLeft})}} 
			  &= {\esp[{\esp[b_i a_i a_i] \esp[a_i^{3^{k-i}-3}] b_{i+1} a_{i+1}^{3^{k-i-1}-1}} ]} \\
			  &= {\esp[{\esp[b_i a_i a_i] a_{i+1}^{3^{k-i-1}-1} b_{i+1} a_{i+1}^{3^{k-i-1}-1}} ]} \\
				\phantom{a}^{\text{(\cref{lemESPassociativityLeft})}} 
			  &= {\esp[{\esp[b_i a_i a_i] a_{i+1}^{3^{k-i-1}-1}} ] \esp[ b_{i+1} a_{i+1}^{3^{k-i-1}-1} ]} \\
			  &= {\esp[{\esp[b_i a_i a_i] a_{i+1}^{3^{k-i-1}-1}} ] \esp[ b_{i+1} a_{i+1} a_{i+1} ] \esp[a_{i+1}^{3^{k-i-1}-3}] },
	\end{align*}
	and $\esp[a_{i+1}^{3^{k-i-1}-3}] = a_{i+2}^{3^{k-i-2}-1}$.
	Adding $a'_i$ (set $a_0' := {\tt a}$) to the string $\espnode{X_i}_i \espnode{X_{i+1}}_i$ yields
	\begin{align*}
			  {\esp[{\esp[{a'_i \espnode{X_i}_i \espnode{X_{i+1}}_i} ]} ]}
			  &= {\esp[{\esp[a'_i b_i a_i^{3^{k-i}-1} b_i^3 a_i^{3^{k-i}-3}]} ]} \\
				\phantom{a}^{\text{(\cref{lemESPassociativityLeft})}} 
			  &= {\esp[{\esp[a'_i b_i] \esp[a_i^{3^{k-i}-1}] b_{i+1} a_{i+1}^{3^{k-i-1}-1}} ]} \\
			  &= {\esp[{\esp[a'_i b_i] \esp[a_i^{3^{k-i}-3} a_i a_i] b_{i+1} a_{i+1}^{3^{k-i-1}-1}} ]} \\
			  &= {\esp[{\esp[a'_i b_i] a_{i+1}^{3^{k-i-1}-1} a'_{i+1} b_{i+1} a_{i+1}^{3^{k-i-1}-1}} ]} \\
				\phantom{a}^{\text{(\cref{lemESPassociativityLeft})}} 
				&= {\esp[{\esp[a'_i b_i] a_{i+1}^{3^{k-i-1}-1}} ] \esp[ a'_{i+1} b_{i+1} ] \esp[ a_{i+1}^{3^{k-i-1}-3} a_{i+1} a_{i+1} ]} \\
		  &= {\esp[{\esp[a'_i b_i] a_{i+1}^{3^{k-i-1}-1}} ] \esp[ a'_{i+1} b_{i+1} ] a_{i+2}^{3^{k-i-2}-1} a'_{i+2} }, 
	\end{align*}
	and $a'_{i+2}$ carries on to the nodes ${\espnode{X_{i+2}}_{i+2} \espnode{X_{i+3}}_{i+2}}$ on height~$i+2$ due to \cref{lemESPassociativityLeft}.

Overall, the leftmost and rightmost node on height~$i+1$ of $T_i$ changes, for $i = 0,\ldots,\V{k}-1$.
In total, \Om{\lg^2 \V{k}} nodes are changed. 
\end{proof}

\begin{figure}[t]
	\centering{%
		$\vtop{\vspace{0pt}\hbox{\Bild[scale=1.0]{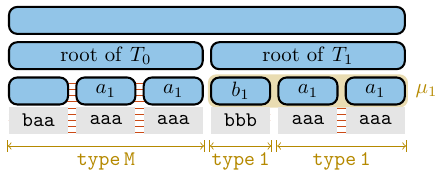}}}
		\hfill
		\vtop{\vspace{0pt}\hbox{\Bild[scale=1.0]{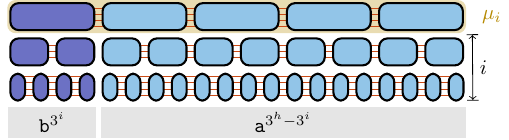}}}$
	}
	\caption{\etInst[\textY] of the example string~$\textY$ defined in \cref{exCounterB} with~$\V{k}=2$ (\emph{left}) and as a schematic illustration (\emph{right}) with the meta-block~$\mu_i$ on height~$i$ (due to space issues the number of nodes/children is incorrect).
	}
	\label{figHanreiBEx}
\end{figure}

\begin{figure}[t]
	\centering{%
		\Bild[width=0.9\textwidth,valign=t]{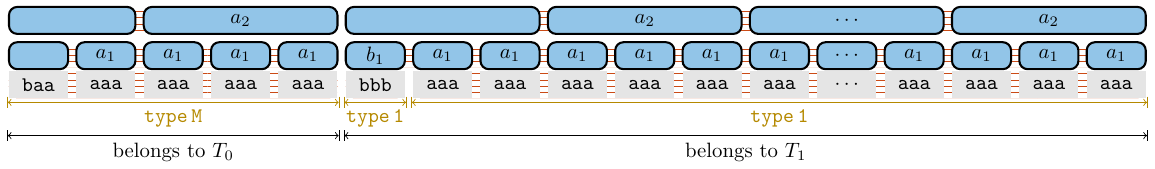}
	}\\
	    \vspace{2em}
	\centering{%
	\Bild[width=0.9\textwidth,valign=t]{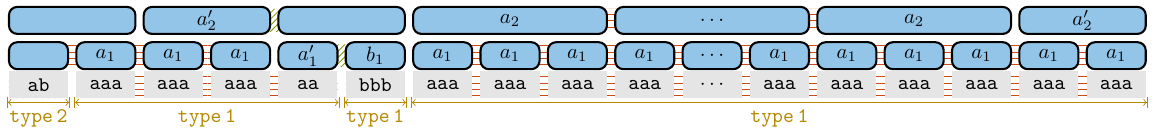}
	}
	\caption{Excerpt of the ESP trees $\etInst[\textY]$ (\emph{top}) and $\etInst[\textY']$ (\emph{bottom}), where
	$\textY = {\tt b}{\tt a}^{3^\V{k} -1} {\tt b}^3 {\tt a}^{3^\V{k} -3} {\tt b}^9 {\tt a}^{3^\V{k} -9} \cdots$ and $\textY = {\tt a} \textY'$
	(defined in \cref{exCounterB}). 
	Due to space issues we contracted $T_0$ to $\etInst[{\tt b a}^{14}]$.
	Note that right of the rightmost $a'_2$ (bottom figure, top right node) is the node $b_2$ (not shown in the figure due to space issues), and both nodes are combined into a \Type{2} meta-block.
}
	\label{figMetablockRight}
\end{figure}

}

\begin{figure}[t]
	\centerline{%
		\Bild[scale=1.0]{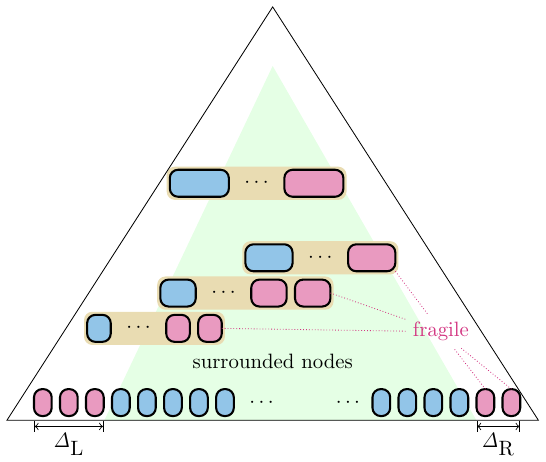}
}
\caption{Division of an ESP tree in surrounded and fragile nodes. The surrounded nodes form an inner cone. 
Neighboring fragile blocks can appear in the non-surrounded areas. 
On each height, the ESP tree can have a constant number of fragile surrounded nodes that do not have fragile nodes in their subtrees.
}
\label{figSurroundedCone}
\end{figure}

\begin{figure}[t]
	\centerline{%
			\Bild[scale=1.0]{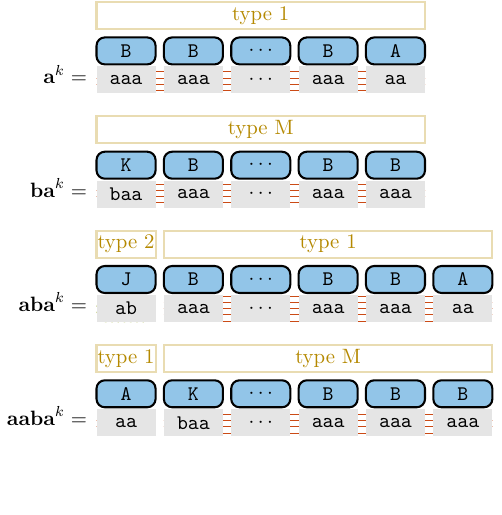}
	}%
	\caption{%
		Prepending the string \Char{aab} to the text $\Char{a}^k$ character by character.
		Each step is given as a row, in which we additionally computed the ESP of the current text.
		The last row shows an example, where a former \Type{1} meta-block changes to \Type{M}, 
		although it is right of a \Type{2} meta-block.
		Here, $k \mod 3 = 2$.
}
	\label{tableBadType3Beginning}
\end{figure}

\begin{figure}[t]
	\centerline{%
		\Bild[scale=1.0]{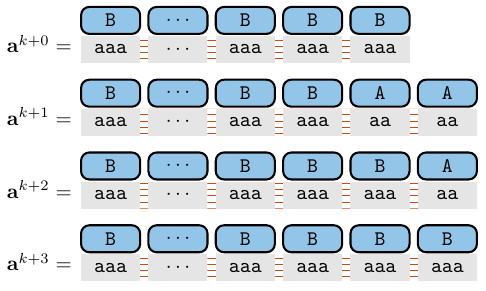}
	}%
	\caption{%
		Greedy blocking of a \Type{1} meta-block.
The greedy blocking is related to the Euclidean division by three.
The remainder $k \mod 3$ is determined by the number of symbols in the last two blocks (here, $k \mod 3 = 0$).
In this example, the ESP technique creates a single, repeating meta-block on each input.
	}
	\label{tableRepeatingAs}
\end{figure}

\block{A new upper bound}
With \JO[the examples in the appendix]{\cref{exCounterA,exCounterB}}, we conclude that the \Oh{\lg^* n \lg n}-bound on the number of fragile nodes for square-free strings (\cref{lemSquareFree}) does not hold for general strings.
To obtain a general upper bound (we stick again to~\ref{stipulationM}), %
we include the repeating meta-blocks in our study of fragile nodes. 
Fragile nodes can now be surrounded (trees of square-free strings do not have fragile surrounded nodes according to \cref{lemSquareFree}).
Remembering that a node is fragile if it has a fragile child, a consequence is that a fragile \Type{2} node is not necessarily non-surrounded (e.g., one of its children can be a fragile surrounded repeating node).
\Cref{figSurroundedCone} sketches the possible occurrences of fragile surrounded nodes.
A first result on a special case is given in the following lemma:

\begin{lemma}\label{lemTypeTwoFragile} %
	A surrounded node~$v$ is contained in the local surroundings of \Oh{\lg^* n \lg n} nodes.
	Given that all those nodes are of \Type{2}, a change of~$v$ causes \Oh{\lg^* n \lg n} name changes.
\end{lemma}
\begin{proof}
\def\windowpagestuff{\flushright\Bild[width=0.95\textwidth]{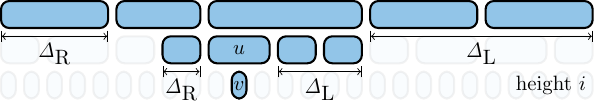}}
\opencutright
\vspace{1em}
\begin{cutout}{1}{0.6\textwidth}{10pt}{3}
	We follow~\cite[Proof of Lemma 9]{Cormode2007sed}:
	We count the number of nodes that contain~$v$ in its local surrounding.
	Given that~$v$ is a node on height~$i$ and $u$ is $v$'s parent, 
	then there are at most $\rcontext/2 \le \rcontext$ nodes preceding~$u$ and $\lcontext/2 \le \lcontext$ nodes succeeding~$u$ that have~$v$ in its local surrounding.
	We count one on height~$i$, and $(\lcontext+\rcontext+1)/2$ on height~$i+1$.
	Since the counted nodes on height~$i+1$ are consecutive, 
	there are at most~$(\lcontext+\rcontext+1)/2$ nodes that are all parents of the counted nodes on height~$i+1$.
	Consequently, there are at most~$(\lcontext+\rcontext+1)/2+\lcontext+\rcontext$ nodes on height~$i+2$ that have~$v$ in their local surroundings.
	Iterating over all heights gives an upper bound of $(\lcontext+\rcontext+1) \sum_{h=0}^{\lg n - i} 1/2^h \le 2(\lcontext+\rcontext+1)$ nodes on each height.
\end{cutout}
\vspace{-1.5em}
\end{proof}

\begin{figure}[t]
	\centerline{%
		\Bild[scale=1.0]{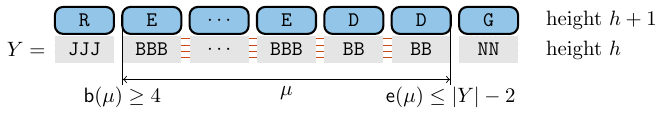}
}
\caption{Setting of \cref{lemRepeatingFragileLeft}. 
According to \cref{lemRepeatingFragileLeft}, a meta-block~$\mu$ in $\esp[\textY]$ of a string~$\textY$ cannot contain a surrounded fragile block if $\ibeg{\mu} \ge 4$ and $\iend{\mu} \le \abs{\textY} - 2$.}
\label{figSurroundedFragile}
\end{figure}

Second, we narrow down the fragile blocks in repeating meta-blocks.
The first block (cf.~\cref{tableBadType3Beginning}) and the two rightmost blocks (cf.~\cref{tableRepeatingAs}) of a repeating meta-block
can be fragile.
Due to the greedy parsing, all other blocks of a repeating meta-block are \SemiOrStable{}.
A repeating meta-block containing fragile \emph{surrounded} blocks needs to start very early or end within the last symbol, as can be seen by the following \lcnamecref{lemRepeatingFragileLeft}:
\begin{lemma}\label{lemRepeatingFragileLeft}
	A repeating meta-block~$\mu$ of $\esp(\textY)$ with $\ibeg{\mu} \ge 4$ and $\iend{\mu} \le \abs{\textY}-2$ cannot contain a fragile block.
\end{lemma}
\begin{proof}
	Since $\ibeg{\mu} \ge 4$, there are at least three symbols before~$\mu$ that are assigned to one or more other meta-blocks.
	When prepending symbols, those meta-blocks can change, absorbing the new symbols or giving the leftmost symbol away to form a \Type{2} meta-block.
	In neither case, they can affect the parsing of~$\mu$, since $\mu$ is parsed greedily.
	Similarly, the succeeding meta-blocks of~$\mu$ keep $\mu$'s blocks from changing when appending symbols.
See \cref{figSurroundedFragile} for a sketch.
\end{proof}

\begin{corollary}\label{lemRepeatingFragileSurrounded}
The edit sensitive parsing introduces at most two fragile surrounded blocks. 
These blocks are the two rightmost blocks of a repeating meta-block whose leftmost block is not surrounded.
\end{corollary}

\begin{lemma}\label{lemRepeatingFragileChild}
	Changing the symbol in a substring of $\espnode{\textY}_{h-1}$ on which a repeating node on height~$h$ is built 
	changes \Oh{1} names on height~$h$.
\end{lemma}
\begin{proof}
	Let $u$ be a repeating node on height~$h$.
	Since it is repeating, it is built on a substring~$\textX := \espnode{\textY}_{h-1}[\ibeg{\textX}..\iend{\textX}]$ 
	of a repeating meta-block~$\mu =  \espnode{\textY}_{h-1}[\ibeg{\mu}..\iend{\mu}]$ with $\dic(u) = \textX$.
	Now change a symbol in $\textX$, say $\espnode{\textY}_{h-1}[i_u]$ with $\ibeg{\textX} \le i_u \le \iend{\textX}$.
	This causes the name of $u$ to change.
	Additionally, it causes the meta-block $\mu$ to split into 
	a repeating meta-block~$\espnode{\textY}_{h-1}[\ibeg{\mu}..i_u-1]$ and a \Type{M} meta-block~$\espnode{\textY}_{h-1}[i_u..\iend{\mu}]$, 
	causing the names of the two rightmost nodes built on the new meta-blocks to change.
	Altogether, there are \Oh{1} name changes on height~$h$.
\end{proof}
An easy generalization of \cref{lemRepeatingFragileChild} is that
changing $k$ consecutive nodes on height~$h-1$ that are children of repeating nodes on height~$h$ changes \Oh{k} names on height~$h$.
With \cref{lemRepeatingFragileChild}, the following lemma translates the result of \cref{lemRepeatingFragileSurrounded} for blocks to nodes:

\begin{lemma}\label{lemBoundFragileESP}
	The ESP tree \etInst[\textY] of a string~$\textY$ of length~$n$ has \Oh{\lg^2 n \lg^*n} fragile nodes, and \Oh{h \lg^*n} fragile nodes on height~$h$.
\end{lemma}
\begin{proof}
	While computing $\espnode{\textY}_{h+1}$ from $\espnode{\textY}_h$, 
	the ESP technique introduces \Oh{1} fragile surrounded blocks according to \cref{lemRepeatingFragileSurrounded}.
	Each fragile surrounded block corresponds to a fragile surrounded node.

	Similar to the proof of \cref{lemSquareFree}, we count all surrounded nodes as fragile whose local surrounding contains a fragile node.
	\Cref{lemTypeTwoFragile} shows that each introduced fragile surrounded block makes \Oh{\lg^* n \lg n} nodes fragile.
	Although we considered only \Type{2} nodes in \cref{lemTypeTwoFragile}, we can generalize this result for all fragile nodes with \cref{lemRepeatingFragileChild}.
	
	To sum up, there are \Oh{h \lg^*n} fragile nodes on height~$h$.
	Because \etInst[\textX] has a height of at most $\lg n$,
	there are $\Oh{\lg^* n \sum_{h=1}^{\lg n} h} = \Oh{\lg^*n \lg^2 n}$
	fragile nodes in total.
\end{proof}

Showing that the number of fragile nodes is indeed larger than assumed makes ESP trees a more unfavorable data structure,
since fragile nodes are cumbersome when comparing strings with ESP trees as done in~\cite{Cormode2007sed}.
Fortunately, we can restore the claimed number of \Oh{\lg n \lg^* n} fragile nodes for a string of length~$n$ with a slight modification of the parsing,
as shown in the following section.

\begin{figure}[t]
	\centering{%
		\Bild{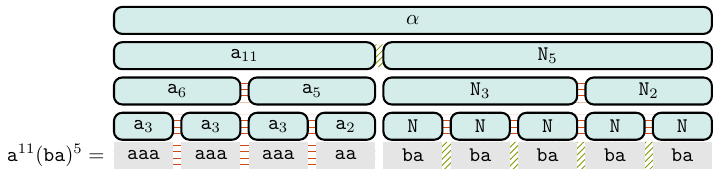}
	}
	\caption{Hierarchical stable parsing. The repeating meta-blocks are determined by the surnames.
	}
	\label{figDiffSurname}
\end{figure}

\begin{figure}[t]
	\centerline{%
		\Bild[scale=1.0]{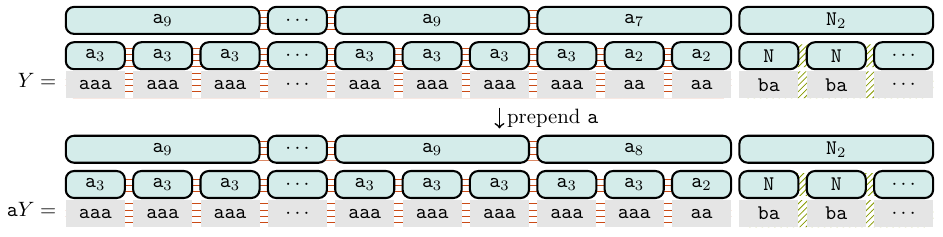}
}

\caption{Excerpt of \hInst[\textY] and \hInst[a \textY] (higher nodes omitted), where $\textY = {\tt a}^{k} {\tt (ba)}^{k'}$
	with $k = 18 + 9^i + 7$ for an integer~$i \ge 0$ and $k' \ge 2$ (cf. \cref{figFragileBadSurrounding}).
	The parsing of~$\textY$ creates a repeating meta-block consisting of ${\tt a}^k$, and a \Type{2} meta-block consisting of ${\tt (ba)}^2$.
		For $k \ge 2$ it is impossible 
		to modify the latter meta-block by prepending characters (bottom figure),
		since the parsing always groups adjacent nodes with the same surname into one repeating meta-block.
	} 
\label{figFragileStabled}
\end{figure}

\begin{figure}[t]
	\centering{%
 		\Bild[width=\textwidth]{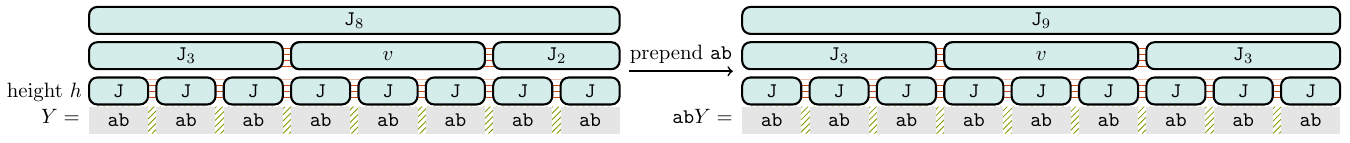}
	}
	\caption{Comparison of \hInst[\textY] and $\hInst[\Char{ab}\textY] = \hInst[\textY\Char{ab}]$, where~$\textY = (\Char{ab})^8$. 
		The node~$v$ with name~\protect\HSPname{$(ab)^3$} is semi-stable. Its generated substring shifts with a length of~$\abs{\generated{\protect\HSPname{$ab$}}} = 2$.
 }
	\label{figHSPsurname}
\end{figure}

\begin{figure}[t]
	\centerline{%
		\Bild[scale=0.6]{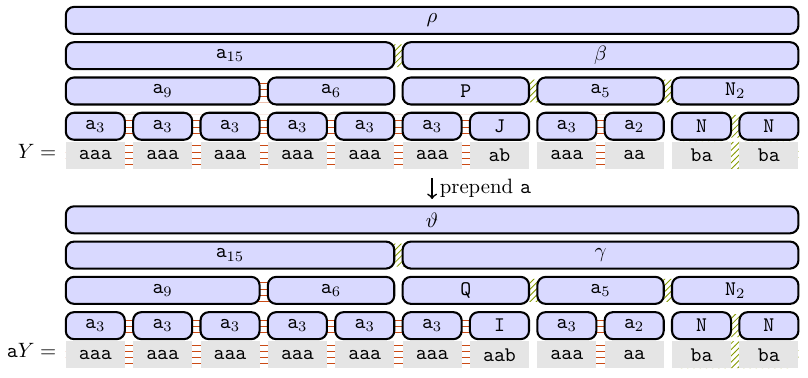}
		\Bild[scale=0.6]{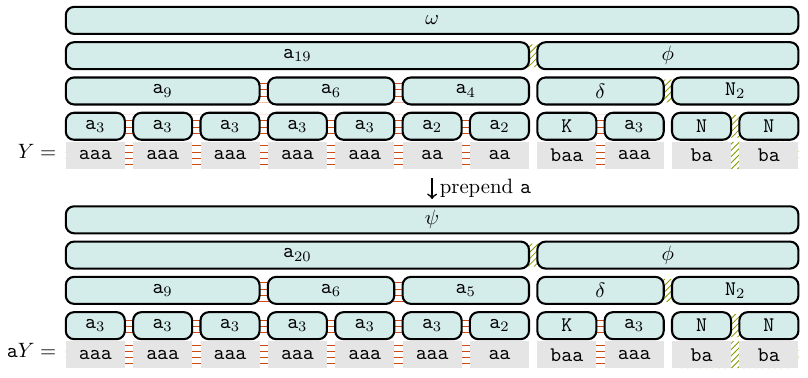}
}
\caption{Impact of the tie breaking rule (\ref{stipulationM}) on emerging \Type{M} nodes.
	A \Type{M} node is created by fusing a single symbol with its sibling meta-block.
	Remember that~\ref{stipulationM} prescribes to fuse the symbol with its \emph{right} meta-block.
	To see why this rule is advantageous, the HSP trees on the \emph{left} (resp.\ \emph{right}) use the tie breaking rule choosing the \emph{left} (resp.\ \emph{right}) meta-block.
	While on the \emph{right} side only the fragile nodes of the leftmost meta-blocks on each height differ after prepending {\tt a} 
	(e.g., the unique occurrence of \protect\HSPname{$a^4$} changes to \protect\HSPname{$a^5$}),
	the change is more dramatical on the \emph{left} side.
	In the top left tree, which is built on $\textY = {\tt a}^{19}{\tt ba}^5{({\tt ba})}^2$,
	the two rightmost nodes~\protect\HSPname{$a^3$} and~\protect\HSPname{$ab$} of the \Type{M} meta-block on the bottom level are children of the leftmost node~\protect\HSPname{$a^4b$}
		of the right meta-block on the next level.
Prepending the character~\Char{a} to $\textY$ (\emph{bottom left}) changes the names of the nodes with names~\protect\HSPname{$ab$} and \protect\HSPname{$a^4b$} to~\protect\HSPname{$a^2b$} and~\protect\HSPname{$a^5b$}, respectively.
	} 
\label{figMixed}
\end{figure}

\begin{figure}[t]
	\Bild[width=\textwidth]{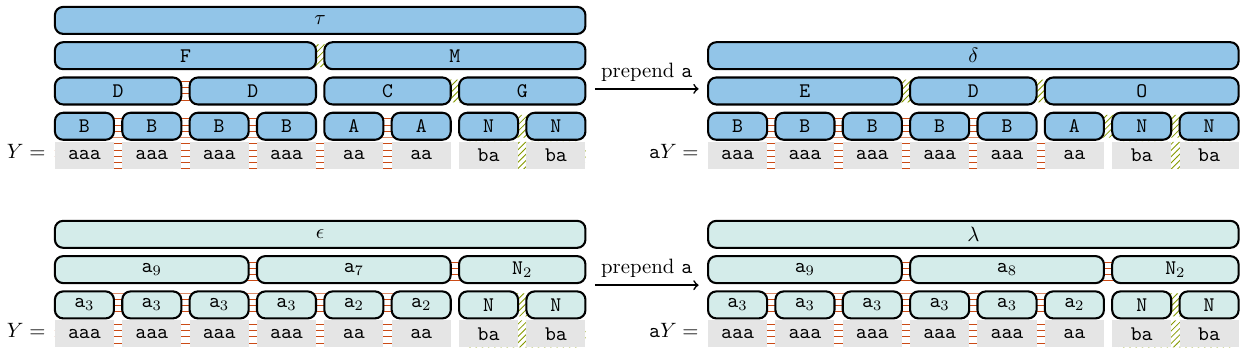}
	\caption{%
		\emph{Left:} \etInst[\textY] (\emph{top}) and \hInst[\textY] (\emph{bottom}) of the string~$\textY$ defined in \cref{figClassicESP}.
		\emph{Right:} \etInst[{\tt a} \textY] (\emph{top}) and  \hInst[{\tt a} \textY]. 
		Unlike the two ESP trees on the top, the two HSP trees below share the same tree topology.
	}
	\label{figDifferenceToESPT}
\end{figure}

\begin{figure}[t]
	\centerline{%
		\Bild[scale=1.0]{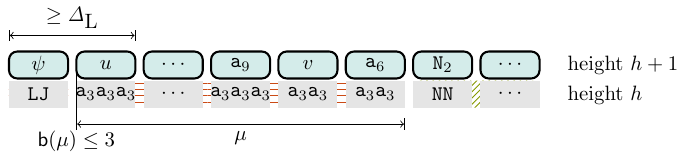}
}
\caption{Setting of \cref{corRepeatingPrefixBlock}. 
	According to \cref{lemRepeatingFragileLeft}, a meta-block~$\mu$ can contain a surrounded fragile block if $\ibeg{\mu} \le 3$ (cf. \cref{figSurroundedFragile}).
	In the figure, the node~$v$ is fragile, since prepending \protect\HSPname{$bab$} changes its name.
	According to \cref{corRepeatingPrefixBlock}, there is a non-surrounded node~$u$ whose generated substring has the generated substring of~$v$ as a prefix.
}
\label{figSurroundedFragileHSP}
\end{figure}

\section{Hierarchical Stable Parsing Trees}\label{secHSP}
Our modification, which we call \intWort{hierarchical stable parsing} or \intWort{HSP},
augments each name with a \intWort{surname} and a \intWort{surname-length}, whose definitions follow:
Given a name~$\Name{Z} \in \Sigma_h$,
let $h'$ with $0 \le h' \le h$ be the largest integer such that
$\dic^{(h')}(\Name{Z})$ consists of the same symbol, say $\dic^{(h')}(\Name{Z}) = \Name{Y}^\ell \in \Sigma^*_{h-h'}$ for a symbol~$\Name{Y} \in \Sigma_{h-h'}$ and an integer~$\ell \ge 1$. 
Then the surname and surname-length of~$\Name{Z}$ are the symbol~$\Name{Y}$ and the integer~$\ell$, respectively.\footnote{By definition, the surname of~\Name{Z} is \Name{Z} itself if $\ell = 1$.}
For convenience, we define the surname of a character to be the character itself.
Then all symbols in $\dic^{(j)}(\Name{Z})$ for every~$j$ with $1 \le j \le h'$ share the same surname with~$\Name{Z}$.

Having the surnames of the nodes at hand, we present the hierarchical stable parsing.
It differs to ESP in how a string of names is partitioned into meta-blocks, whose boundaries now depend on the surnames:
When factorizing a string of names into meta-blocks, we relax the check whether two names are equal;
instead of comparing names we compare by surnames.\footnote{The check is relaxed since names with different surnames cannot have the same name.}
This means that we allow meta-blocks of \Type{1} to contain different symbols as long as all symbols share the same surname.
The other parts of the edit sensitive parsing defined in \cref{subsecESP} are left untouched; in particular, the alphabet reduction uses the symbols as before.
We write $\hInst[\textY]$ for the resulting parse tree, called \intWort{HSP tree}, when the HSP technique is applied to a string~\textY{}. %
\Cref{figDiffSurname} shows $\hInst[\Char{a}^{11} (\Char{ba})^5]$. 
In the rest of this article (and as shown in \cref{figDiffSurname}), we give a repetitive node with surname~\Name{Z} and surname-length~$\ell$ the name~$\MakeHSPName{\Name{Z}}{\ell}$.
We omit the surname-length if it is one (and thus, the label of a non-repetitive node is equal to its name).
For the other nodes, we use the names of \cref{tableNames}.
We can do that because the name of a node can be identified by its surname and surname-length, as can be seen by the following \lcnamecref{lemSurnameToNameIsom}:
\begin{lemma}\label{lemSurnameToNameIsom}
	The name of a node is uniquely determined by its surname and surname-length.
\end{lemma}
\begin{proof}
A node with surname-length one is not repetitive, and therefore, its name is equal to its surname.
Given a repetitive node~$v$ with surname~$\Name{Z}$ and surname-length~$\ell$, there is a height~$h$ such that $\dic^{(h)}(v) = \Name{Z}^\ell$.
For every height~$h'$ with $1 \le h' \le h$, $\dic^{(h')}(v)$ consists of the same symbol, and hence $\dic^{(h')}(v)$ is parsed greedily by HSP\@.
This means that the iterated greedy parsing of the string~$\Name{Z}^\ell$ determines the name of~$v$.
\end{proof}

\begin{figure}[t]
	\centering{%
		\Bild[scale=1.0]{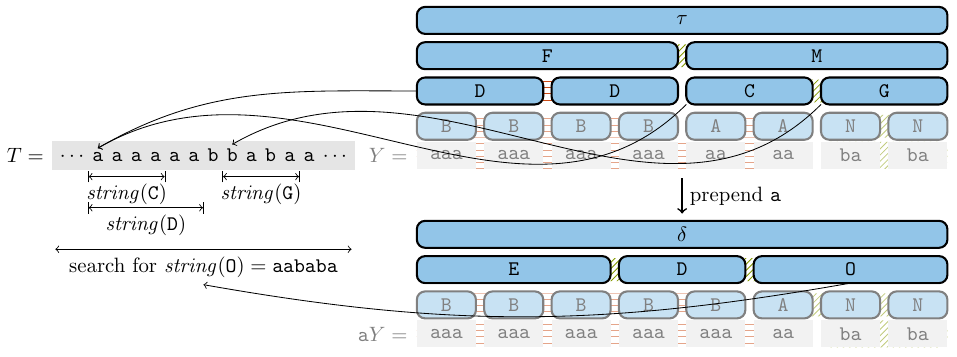}
	}
	\caption{Problem with dynamic updates of ESP trees stored in text space.
		Truncated nodes are grayed out.
		Each leaf of the truncated trees is assigned a pointer to its generated substring, which are 
		substrings of the text~$T$ (\emph{left}).
		Suppose that we have built \etInst[\textY] (\emph{top right}) on a substring~$\textY$ of $\textT$ ($\textY$ defined as in \cref{figDifferenceToESPT}), and that the names \protect\ESPname{$a^6$}, \protect\ESPname{$a^4$} and \protect\ESPname{$(ba)^2$} are already present in the dictionary (hence, they have different generated substrings).
		Further suppose that the space of $\textY$ in $T$ has been overwritten.
		When prepending an {\tt a} to \etInst[\textY] to form \etInst[{\tt a} \textY] (\emph{bottom right}), the node \protect\ESPname{$(ba)^2$} changes to \protect\ESPname{$a^2(ba)^2$},
		for which we need to search its generated substring (assuming that \protect\ESPname{$a^2(ba)^2$} is not yet stored in the dictionary).
	The example can be elaborated such that \protect\ESPname{$(ba)^2$} and \protect\ESPname{$a^2(ba)^2$} become surrounded nodes (prepend ${\tt a}^{9k}$ and append ${\tt b}^{9k}$ for a sufficiently large $k \ge 1$).
}
	\label{figNoInPlaceForESP}
\end{figure}

\subsection{Upper Bound on the Number of Fragile Nodes}\label{subsecHSPFragileNodes}
The motivation of introducing the HSP technique becomes apparent with the three following facts:
\begin{Facts}
\item \label{itHSPPropsRep} %
	Given that the surnames of the repetitive nodes in a repeating meta-block~$\mu$ are~$w$,
	the generated substring of each such repetitive node is a repetition of the form~$\textX^k$ with the same~$\textX = \generated{w} \in \Sigma^*$
	(or $\textX = w$ in case $w \in \Sigma$),
	but with possibly different surname-lengths~$k$ (e.g., $\generated{\HSPname{$(ba)^3$}} = ({\tt ba})^3$ and~$\generated{\HSPname{$(ba)^2$}} = ({\tt ba})^2$ in \cref{figDiffSurname}).
	Due to the greedy parsing of the repeating meta-blocks, 
	the surname-lengths of the last two nodes in $\mu$ cannot be larger than the surname-lengths of the generated substrings of the other nodes (with the same surname) contained in~$\mu$. See \cref{figFragileStabled} for an example when prepending a character to the input.

	\item \label{itHSPPropSurname}
		The shift of a semi-stable node is always a multiple of the length of its surname
		(recall that semi-stable nodes are defined like stable nodes, but with slight shifts, cf.\ \cref{secChangingNodes}):
		Let~\HSPname{$ab$} be the surname of a semi-stable node~$v \in \espnode{\textY}_h$ on height~$h$.
				Given $\HSPname{$ab$} \in \Sigma_{h'}$ for a height~$h'$ with $h' \ge 0$, 
				$\dic^{(h-h')}(v)$ is a repetition of the symbol~\HSPname{$ab$} on height~$h'$.
				A shift of $v$ can only be caused by adding one or more \HSPname{$ab$}s to $\espnode{\textY}_{h'}$.
				In other words, the shift is always a multiple of~$\dic^{(h')}(\HSPname{$ab$})$.
\Cref{figHSPsurname} shows an example of a semi-stable node~$v$.

	\item \label{itHSPPropsMixed} 
		A non-repetitive \Type{M} block can be fragile only if it is non-surrounded.
		By definition, a repeating meta-block~$\mu$ contains a non-repetitive block~$\beta$ iff $\mu$ is \Type{M}.
	The block~$\beta$ can only be located at the beginning or ending of~$\mu$.
	Remembering~\ref{stipulationM}, 
	$\beta$'s none-repetitiveness is caused by
	\begin{itemize}
		\item fusing a symbol with its \emph{succeeding} meta-block, or
		\item fusing the \emph{last} symbol with its \emph{preceding} meta-block.
	\end{itemize}
	In both cases, it is impossible that $\beta$ is a surrounded block if $\ibeg{\mu} \le \lcontext$.
	If~$\beta$ is surrounded, it is \SemiOrStable{} due to \cref{lemRepeatingFragileLeft}.
	Note that with sticking to the choice made in~\ref{stipulationM}, we also experience a more stable behavior like in \cref{figMixed}.
\end{Facts}

\begin{figure}[t]
	\centering{%
		\Bild[scale=1.0]{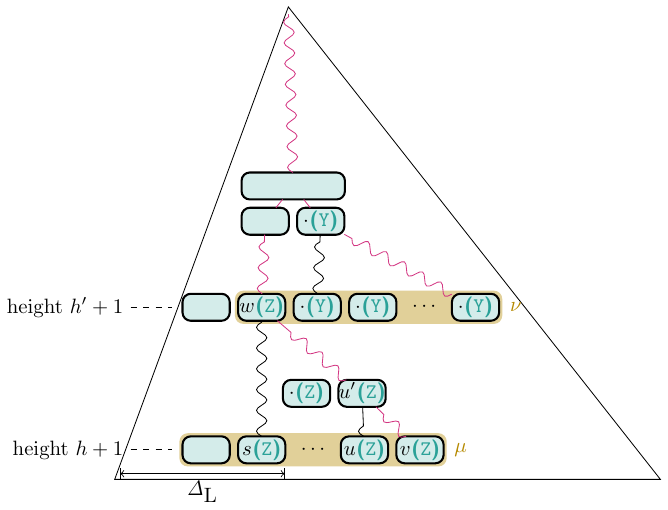}
	}
\caption{Sketch of the HSP tree used to show \cref{lemmaSurroundedFragileHSP}. 
	In the sketch, we give the repetitive nodes of the meta-block~$\nu$ the surname~\Name{Y}.
	Repetitive nodes are labeled with their surnames, which are put into parentheses.
	}
	\label{figHspProof}
\end{figure}

These facts make the HSP technique more stable than the ESP technique, as can be seen in \cref{figDifferenceToESPT}, for instance.
In the following, we study the number of fragile surrounded nodes (like in \cref{secChangingNodes} for the ESP trees), and show 
the invariant (\ref{itSurroundedFragileProofGenerated} in \cref{lemmaSurroundedFragileHSP}) that the generated substring of a fragile surrounded node is always the prefix of the generated substring of a name that is already stored in~$\dic$. 
On block level, this is an easy conclusion of \cref{lemRepeatingFragileLeft} and \cref{itHSPPropsMixed,itHSPPropsRep}.

\begin{corollary}\label{corRepeatingPrefixBlock}
Given $n > 4$ and a repeating meta-block~$\mu$ having a fragile surrounded block~$\beta$,
$\mu$ has at least one block preceding~$\beta$ that contains three symbols with the same surname.
In particular, the leftmost of these preceding blocks is non-surrounded.
\end{corollary}
\begin{proof}
Since $\beta$ is surrounded, the condition $\abs{\mu} \ge \lcontext-2$ holds.
By the definition of~$\lcontext$ in \cref{lemAlphabetReduction}, $\lcontext - 2 \ge 5$ for $n > 4$.
Assuming that the repetitive blocks in~$\mu$ have the surname~\Name{Z},
this means that there is at least one repetitive block~$\gamma$ with surname~\Name{Z} preceding~$\beta$ that contains three symbols of $\mu$.
But the fragile surrounded block $\beta$ is also a repetitive block according to \cref{itHSPPropsMixed}.
This means that the surname-length of~$\beta$ is at most as long as the surname-length of~$\gamma$ due to \cref{itHSPPropsRep}, 
i.e., the generated substring of the node corresponding to~$\beta$ is a prefix of the generated substring of the node corresponding to~$\gamma$.
Let~$\gamma$ be the leftmost such block.
Remembering that~$\mu$ can start with a non-repetitive node in case that $\mu$ is of \Type{M}, it is not obvious that~$\gamma$ is non-surrounded.
However, according to \cref{lemRepeatingFragileLeft}, $\ibeg{\mu} \le 2$ must hold. 
This means that $\ibeg{\gamma} \le 5 \le \lcontext$, so $\gamma$ is non-surrounded.
See \cref{figSurroundedFragileHSP} for a sketch (with $\Name{Z} = \Char{a}$).
\end{proof}

In general, the aforementioned invariant does not hold for ESP trees, but is essential for the sparse suffix sorting in text space.
There, our idea is to create an HSP or ESP tree on a newly found re-occurring substring. 
We would like to store the ESP tree in the space of one of those substrings, which we can do by
truncating the tree at a certain height (removing the lower nodes), and 
changing the pointer of each (new) leaf such that the name of a leaf refers to its generated substring that is found 
in the remaining text. Unfortunately, there is a problem when \ac{preapp} characters to enlarge the ESP tree, since
a leaf could change its name such that its generated substring needs to be updated - which can be non-trivial 
if its generated substring refers to an already overwritten part of the text that is not present in the remaining text as a (complete) substring.
\Cref{figNoInPlaceForESP} demonstrates the problem when truncating ESP trees at height~$2$.
Fortunately, the following \lcnamecrefs{lemmaSurroundedFragileHSP} restrict the problem of updating the generated substring when an HSP node is surrounded and fragile.
We start with appending characters:
\begin{lemma}\label{lemmaSurroundedFragileHSPAppending}
	There is no surrounded HSP node~$v$ whose name changes when appending characters.
\end{lemma}
\begin{proof}
	Assume that $v$'s name changes on appending characters.
	Moreover, assume that $v$'s local surrounding does not contain a fragile node (otherwise swap~$v$ with this node).
	First, since there is no fragile node in~$v$'s local surrounding, it has to be a repeating node according to \cref{lemTypeTwoStable}.
	Second, according to \cref{lemRepeatingFragileSurrounded}, it has to be one of the last two nodes built on a repeating meta-block~$\mu$.
	But there is no way to change the names of the last two blocks of~$\mu$ by appending characters unless these blocks are non-surrounded.
	So a surrounded node cannot have a node in its surrounding whose name changes when appending characters.
\end{proof}

\begin{lemma}\label{lemmaSurroundedFragileHSP}
Let~$v$ be a fragile surrounded node of an HSP tree.
Then
\begin{enumerate}[leftmargin=*,label={Claim~\arabic*:},ref={Claim~\arabic*}]
	\item $v$ is a repetitive node,  \label{itSurroundedFragileProofRepeating}
	\item \ac{preapp} characters cannot change $v$'s surname, and \label{itSurroundedFragileProofReparsing}
	\item the generated substring of~$v$ is always a prefix of the generated substring of an already existing node 
		belonging to the same meta-block as~$v$. \label{itSurroundedFragileProofGenerated}
\end{enumerate}
\end{lemma}
\begin{proof}
To show the lemma, let $n > \lcontext + \rcontext$, otherwise there are no surrounded nodes.
There are two (non-exclusive) possibilities for a node to be fragile and surrounded:
\begin{itemize}
	\item it belongs to the last two nodes built on a repeating meta-block (due to \cref{lemRepeatingFragileSurrounded}), or 
	\item its subtree contains a fragile surrounded node, since by definition, 
		\begin{itemize}
			\item a node is fragile if it contains a fragile node in its subtree, and
			\item all nodes in the subtree of a surrounded node are surrounded.
		\end{itemize}
\end{itemize}
We iteratively show the claim for all heights, starting at the bottom:
Let $v$ be one of the \emph{lowest} fragile surrounded nodes in \hInst[\textY] (\emph{lowest} meaning that there is no fragile node in $v$'s subtree).
Suppose that~$v$ is a node on height~$h+1$ with $h \ge 0$.
Since there is no fragile surrounded node in $v$'s subtree, 
$v$ is one of the last two nodes built on a repeating meta-block $\espnode{\textY}_h[\mu]$ (i.e., $\textY[\mu]$ for $h = 0$).
Due to \cref{itHSPPropsMixed}, \ref{itSurroundedFragileProofRepeating} holds for~$v$; let~\Name{Z} be its surname.
Since~$v$ is fragile, $\ibeg{\mu} \le 3$ must hold (otherwise we get a contradiction to \cref{lemRepeatingFragileLeft}). %
But since $v$ is surrounded, there is a repetitive node~$u$ with surname~\Name{Z} preceding~$v$ that is built on three symbols ($\dic(u) \in \Sigma_h^3$) of $\mu$ due to
\cref{corRepeatingPrefixBlock}.
In particular, the leftmost repetitive node~$s$ of~$\mu$ is not surrounded.

We only consider prepending a character (appending is already considered in \cref{lemmaSurroundedFragileHSPAppending}).
Assume that $v$'s name changes when prepending a specific character.
By \cref{itHSPPropsRep}, the HSP technique assigns a new name to~$v$, but it does not change its surname (so~\ref{itSurroundedFragileProofReparsing} holds for~$v$).
Additionally, $\generated{v}$ is a substring of \generated{u}, where~$u$ is one of $v$'s preceding nodes having the surname~\Name{Z}, and therefore~\ref{itSurroundedFragileProofGenerated} holds for~$v$.
For example, let~$v$ be the node with name~\HSPname{$a^7$} in \hInst[\textY] of \cref{figFragileStabled}, then $\generated{v} = \Char{a}^7$,
which is a prefix of $\generated{\HSPname{$a^9$}} = \Char{a}^9$.
After prepending the character \Char{a}, $v$'s name becomes \HSPname{$a^8$} with $\generated{v} = \Char{a}^8$. Still,
\generated{v} is a prefix of $\generated{\HSPname{$a^9$}}$.

Due to this behavior, the node~$v$ is always assigned to~$\mu$, regardless of what character is prepended.
This means that it is only possible to extend or shorten $\mu$ on its left side, or equivalently,
$\mu$'s right end is \emph{fixed}; the parsing of a meta-block succeeding~$\mu$ cannot change.
This means that the parsing assures that every surrounded node located to the right of $\espnode{\textY}_h[\mu]$ is (semi-)stable.
We conclude that the claim holds for the heights $1,\ldots,h+1$.

Next, we show that the claim holds for all height~$h+2,\ldots,h'$, where $h'+1$ is the height of the lowest common ancestor~$w$ of~$s$ and~$v$.
\Cref{figHspProof} gives a visual representation of the following observations:
When following the nodes from~$v$ up to~$w$, there is a path of ancestor nodes with surname~\Name{Z}.
Except for~$w$, each such ancestor node~$u'$ has a neighbor with surname~\Name{Z}.
On changing the name of~$v$, all nodes on the height of~$u'$ are unaffected, except~$u'$.
That is because the ancestor of~$s$ on the same height as~$u'$ is put with~$u'$ in the same repeating meta-block,
which comprises all neighboring nodes with surname~\Name{Z}.
By the analysis above, changing the name of~$u'$ cannot change the parsing of the other nodes on the same height.
We conclude that the claim holds for the heights $h+2,\ldots,h'$.

Let us focus on the nodes on height~$h'+1$:
The node $w$ is not surrounded, because it contains the non-surrounded node~$s$ in its subtree.
Having neighbors with different surnames, $w$ is either blocked in a \Type{2} or \Type{M} meta-block. 
\begin{itemize}
	\item In the former case (\Type{2}), the analysis of~\cref{lemTypeTwoFragile} shows that~$w$ only affects the parsing of the non-surrounded nodes.
	There can be a non-surrounded meta-block on a height~$h'' > h'+1$ having a fragile surrounded node~$v'$.
	But then~$v'$ cannot contain a fragile node (the descendants of~$w$ are the last fragile \emph{surrounded} nodes, and $w$ is non-surrounded).
	This means that we can apply the same analysis to~$v'$ as for~$v$.

\item In the latter case (\Type{M}), $w$ is fused with a repeating meta-block to form a \Type{M} meta-block~$\nu$, 
changing the names of the leftmost and two rightmost nodes of~$\nu$, where the leftmost node is~$w$. 
Assume that the two rightmost nodes of~$\nu$ are fragile and surrounded
(otherwise we conclude with the previous case that there are no fragile surrounded nodes on height~$h'+1$).
Under this assumption, the rightmost nodes of~$\nu$ are repeating nodes due to \cref{itHSPPropsMixed}.
Hence, we can apply the same analysis as for~$v$, and conclude the claim for all heights above~$h'$.
\qedhere{}
\end{itemize}
\end{proof}
A direct consequence is that there are \Oh{1} fragile surrounded nodes on each height.
With \cref{lemBoundFragileESP} we get the following \lcnamecref{thmFragileHSP}:

\begin{theorem}\label{thmFragileHSP}
	The HSP tree \hInst[\textY] of a string~$\textY$ of length~$n$ contains at most \Oh{\lg^*n} fragile nodes on each height.
\end{theorem}

Having a bound on the number of fragile nodes, we start to study the algorithmic operations of an HSP tree.
The first operation is how to actually build an HSP tree.
For that, we have to think about its representation:

\subsection{Tree Representation}\label{secTreeRepresentation}
Unlike \citeauthor{Cormode2007sed}, who use hash tables to represent the dictionary~$\dic$, we follow a deterministic approach.
In our approach, we represent~$\dic$ by storing the HSP tree as a \ac{CFG}\@.
A name (i.e., a non-terminal of the \ac{CFG}) is represented
by a pointer to a data field (an allocated memory area), which is composed differently for leaves and internal nodes:
\begin{description}
	\item[Leaves.] A leaf stores a position $i$ and a length $\ell \in \{2,3\}$ such that $\textY[i..i+\ell-1]$ is the generated substring.
	\item[Internal nodes.] An internal node stores the length of its generated substring, and the names of its children.
		If it has only two children, we use a special, invalid name $\invalid$ for the non-existing third child 
		such that all data fields are of the same length.
\end{description}
This information helps us to navigate from a node to its children or its generated substring in constant time, and
to navigate top-down in the HSP tree by traversing the tree from the root in time linear in the height of the tree.

To accelerate substring comparisons, we want to give nodes with the same children (with respect to their order and names) the same name, 
such that the dictionary~$\dic$ is injective.
To keep the dictionary injective, we do the following: Before creating a new name for the rule $b \rightarrow xyz$ (we set~$z = \invalid$ if the rule is $b \rightarrow xy$), we check whether there already exists a name for $xyz$.
To perform this lookup efficiently, we need also the \emph{reverse} dictionary of~$\dic$, with the right hand side of the rules as search keys.
We want the reverse dictionary to be of size $\Oh{\abs{\textY}}$, supporting lookup and insert in $\Oh{\lookupTime}$ (deterministic) time for a $\lookupTime = \lookupTime(n)$ depending on~$n$.
For instance, a balanced binary search tree has $\lookupTime = \Oh{\lg n}$.

With this tree representation, we can build HSP trees within the following time and space bounds:

\begin{lemma}\label{lemmaBuildHSP}
The HSP tree~\hInst[\textY] of a string~$\textY$ of length~$n$ can be built in \Oh{n \tuple{\lg^*n + \lookupTime}} time.
It takes $\Oh{n}$ words of space.
\end{lemma}
\begin{proof}
	A name is inserted or looked-up in $\lookupUpper$ time.
	Due to the alphabet reduction technique (see \cref{lemAlphabetReductionSpeed}), applying $\esp$ on a substring of length $\ell$ takes $\Oh{\ell \lg^* n}$ time, 
	returning a sequence of blocks of length at most $\ell/2$.
\end{proof}

\subsection{LCE Queries in HSP Trees}\label{secLCE}
Like the trees~\cite{Alstrup2000Pmi,NishimotoIIBT16} based on signature encoding, we show that HSP trees are good at answering LCE queries.
The idea is to compare the names of two nodes to test whether the generated substrings of both nodes are the same.
Remembering that two nodes with the same generated substring can have different names (cf.\ the end of \cref{secDictAndNames}), 
we want to have a rule at hand saying when two nodes with different names must have different generated substrings.
It is easy to provide such a rule when the input string is square-free:
In this case, all fragile nodes are non-surrounded according to \cref{lemSquareFree}, and thus
we know that the surrounded nodes are stable. 
Since each height consists of exactly one \Type{2} meta-block, the equality of two substrings~$\textX$ and $\textY$ can be checked by comparing the names of two surrounded nodes whose generated substrings are~$\textX$ and $\textY$, respectively.
For general strings, we need additional information about the generated substring of each repeating node.
That is because the names of two repeating nodes at \emph{the same height} already differ when the generated substring of one node is a proper prefix of the  generated substring of the other node.
Fortunately, this additional information is given by the surnames and surname-lengths (see \cref{itHSPPropSurname} in \cref{subsecHSPFragileNodes}):

Having a common dictionary~$\dic$ for all HSP trees that stores the length of the string $\dic^{(h)}(\Name{Z})$ for each name $\Name{Z} \in \Sigma_h$,
we explain how HSP trees can answer LCE queries efficiently.

\begin{figure}[t]
	\centering{%
		\Bild[scale=1.0]{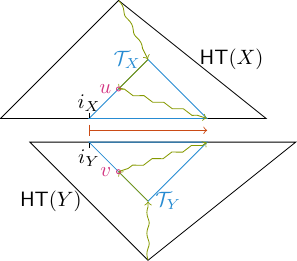}
	}
	\caption{Conception of the proof of \cref{lemmaLCE}. 
 		To compute the longest common prefix of $\textX[i_{\textX}..]$ and $\textY[i_{\textY}..]$ (arrow in the center), 
 	we walk down the trees $\hInst[\textX]$ and $\hInst[\textY]$ (depicted by the upper and the lower triangle, respectively) on the paths towards the
	leaves containing $\textX[i_\textX]$ and $\textY[i_\textY]$, respectively, by simultaneously visiting two nodes on the same height of both trees.
	The nodes~$u$ and~$v$ in the figure are on these paths. 
	Suppose that they are on the same height and have the same surname. 
	On visiting both nodes, we know that the longest common prefix is at least 
	$\min\tuple{\abs{\generated{u}},\abs{\generated{v}}}$ long. 
	We update the destination of our traversal accordingly, such that we follow the paths from~$u$ and~$v$ to the leaves covering
	the not-yet checked parts of the longest common prefix that we want to compute.
	}
	\label{figLCE}
\end{figure}

\begin{lemma}\label{lemmaLCE}
	Given $\hInst[\textX]$ and $\hInst[\textY]$ built on two strings~$\textX$ and $\textY$ with $\abs{\textX} \le \abs{\textY} \le n$
	and two text positions $1 \le i_\textX \le \abs{\textX}, 1 \le i_\textY \le \abs{\textY}$,
	we can compute $\lcp(\textX[i_{\textX}..], \textY[i_{\textY}..])$ in $\Oh{\lg n \lg^* n}$ time.
\end{lemma}
\begin{proof}
	We use the following property:
	If two nodes have the same surname~\Name{Z}, then the generated substrings of both nodes are $\textZ^i$ and $\textZ^j$, respectively, 
	with the respective surname-lengths~$i$ and~$j$,
	where $\textZ = \generated{\Name{Z}}$. 
	This means that the generated substring of one node is a prefix of the generated substring of the other.
	In the particular case $i = j$, both nodes share the same subtree and consequently have the same name according to \cref{lemSurnameToNameIsom}.
	In summary, this property allows us to omit the comparison of the subtrees of two nodes with the same surname, and thus speeds up the LCE computation,
	which is done in the following way (cf. \cref{figLCE}):

	\begin{enumerate}[(1)]
		\item 
			We start with traversing the two paths from the roots of $\hInst[\textX]$ and $\hInst[\textY]$ to the leaves~$\leaf_{\textX}$ and~$\leaf_{\textY}$
			whose generated substrings contain $\espnode{\textX}_0[i_\textX]$ and $\espnode{\textY}_0[i_\textY]$, respectively:
	\item	We traverse the two paths leading to the leaves~$\leaf_{\textX}$ and~$\leaf_{\textY}$, respectively, in a simultaneous manner such that we always visit a pair~$(u,v)$ of nodes on the same height belonging to $\hInst[\textX]$ and $\hInst[\textY]$, respectively. \label{itlemLCE2}
	\item Given that $u$ and $v$ share the same surname $\Name{Z} \in \Sigma_h$, 
we know the lengths of their generated substrings ($\abs{\dic^{(h)}(\Name{Z})}^{\ell_u}$ and $\abs{\dic^{(h)}(\Name{Z})}^{\ell_v}$) by having their surname-lengths~$\ell_u$ and~$\ell_v$ at hand.
As a consequence, we know that $\textX[i_{\textX}..]$ and $\textY[i_{\textY}..]$ have a common prefix of at least $\min\tuple{\abs{\dic^{(h)}(\Name{Z})}^{\ell_u}, \abs{\dic^{(h)}(\Name{Z})}^{\ell_v}}$.
We update the variables~$\leaf_{\textX}$ and $\leaf_{\textY}$ to be the leaves whose generated substrings contain $\espnode{\textX}_0[i_\textX+\abs{\dic^{(h)}(\Name{Z})}^{\ell_u}]$ and $\espnode{\textY}_0[i_\textY+\abs{\dic^{(h)}(\Name{Z})}^{\ell_v}]$, respectively.
Subsequently, we continue our tree traversals from~$u$ and~$v$ to the updated destinations~$\ell_{\textX}$ and~$\ell_{\textY}$, respectively.
Since $\leaf_{\textX}$ and $\leaf_{\textY}$ are not in the respective subtrees of~$u$ and~$v$, 
we climb up the tree to the lowest common ancestor of $u$ (resp.\ $v$) and~$\leaf_{\textX}$ (resp.\ $\leaf_{\textY}$), and
recurse on~\ref{itlemLCE2}.
\item If we end up at a pair of leaves (i.e., $u = \leaf_{\textX}$ and $v = \leaf_{\textY}$), we compare their generated substrings naively. 
	If we find a mismatching character in both generated substrings, we can determine the value of~$\ell$ and terminate.
	We also terminate if there is no mismatch, but $\leaf_{\textX}$ or $\leaf_{\textY}$ is the rightmost leaf of $\hInst[\textX]$ or $\hInst[\textY]$, respectively.
	In all other cases, 
	we set $\leaf_{\textX}$ and $\leaf_{\textY}$ to their respectively succeeding leaves, climb up to the 
	parents of $u$ and $v$,
	and recurse on~\ref{itlemLCE2}.
	\end{enumerate}
During the traversals of both trees, we spend constant time for each navigation operation, i.e., (a) selecting a child, and (b) climbing up to the parent of a node:
On the one hand, we select a child of a node~$v$ in constant time by following the pointer of the name of~$v$ (defined in \cref{secTreeRepresentation}).
On the other hand, we maintain, for each tree, a stack storing all ancestors of the currently visited node during the traversal of the respective tree:
Each stack uses $\Oh{\lg n}$ words, and can return the parent of the currently visited node in constant time.

To upper bound the running time of the traversals, we examine the nodes visited during the traversals.
Starting at both root nodes, we follow the path from the root of \hInst[\textX] (resp.\ \hInst[\textY]) down to
the roots of the minimal subtree~\SubTree{\textX} of $\hInst[\textX]$ (resp.\ \SubTree{\textY} of \hInst[\textY]) covering $\textX[i_{\textX}..i_{\textX}+\ell]$ (resp.\ $\textY[i_{\textY}..i_{\textY}+\ell]$).%
\footnote{We assume that $i_{\textX}+\ell \le \abs{\textX}$ and $i_{\textY}+\ell \le \abs{\textY}$. Otherwise, let \SubTree{\textX} and \SubTree{\textY} cover $\textX[i_{\textX}..i_{\textX}+\ell-1]$ and $\textY[i_{\textY}..i_{\textY}+\ell-1]$, respectively.}
After entering the subtrees~$\SubTree{\textX}$ and \SubTree{\textY},
we will never visit nodes outside of $\SubTree{\textX}$ and \SubTree{\textY}.
The question is how many nodes of $\SubTree{\textX}$ and $\SubTree\textY$ differ.
This can be answered by studying the tree $\hInst[\textZ]$ built with the same dictionary~$\dic$, 
where $\textZ := \textX[i_{\textX}..i_{\textX}+\ell-1] = \textY[i_{\textY}..i_{\textY}+\ell-1]$:
On the one hand, $\hInst[\textZ]$ has \Oh{\lg^* n} fragile nodes on each height according to \cref{thmFragileHSP}.
On the other hand, each (semi-)stable node in $\hInst[\textZ]$ is found in both $\SubTree\textX$ and $\SubTree\textY$ with the same name and surname.
This means that when traversing \hInst[\textX] and \hInst[\textY] within their respective subtrees $\SubTree\textX$ and $\SubTree\textY$, 
we only visit \Oh{\lg^* n} pairs of nodes per height 
(remember that we follow the two paths to the leaves~$\leaf_{\textX}$ and $\leaf_{\textY}$, respectively, up to the point where the surnames of the visited pair of nodes match). 

To sum up, we 
(a) compute paths from the roots to  $\espnode{\textX}_0[i_\textX]$ and $\espnode{\textY}_0[i_\textY]$, respectively, in \Oh{\lg\abs{\textY}} time, and
(b) compare the children of at most $\Oh{\lg^*n}$ nodes per height.
Since both trees have a height of \Oh{\lg \abs{\textY}}, we obtain our claimed running time.
\end{proof}

The following \lcnamecref{lemmaLCEPointer} is a small refinement of \cref{lemmaLCE} that already shows the result of \cref{thmLCEtradeoff} for $\tau = 1$:
\begin{corollary}\label{lemmaLCEPointer}
	Given $\hInst[\textX]$ and $\hInst[\textY]$ built on two strings~$\textX$ and $\textY$ with $\abs{\textX} \le \abs{\textY} \le n$
	and two text positions $1 \le i_\textX \le \abs{\textX}, 1 \le i_\textY \le \abs{\textY}$,
	we can compute $\ell := \lcp(\textX[i_{\textX}..], \textY[i_{\textY}..])$ in $\Oh{\lg \ell \lg^* n}$ time.
\end{corollary}
\begin{proof}
	Our idea is to enhance an HSP tree with a data structure such that
		climbing up from a child to its parent can be performed in constant time.
	This can be achieved when we represent the tree topology of an HSP tree with a pointer based tree, 
	in which each node stores its name and the pointer to its parent.
	The leaves are stored sequentially in a list.
	A bit vector with the same length as the input string is used to mark the borders of the generated substrings of the leaves.
	Given a text position~$i$, we can access the leaf whose generated substring contains~$i$ in constant time with a rank-support on the bit vector. 
	The bit vector with rank-support takes $n + \oh{n}$ bits.
	The pointer based tree can be built with the HSP tree without an additional time overhead, and takes \Oh{n} words of space.
\end{proof}

In the next section, we describe a preliminary version of our sparse suffix sorting algorithm that does not exploit the text space yet.

\section{Sparse Suffix Sorting}\label{secSSS}
The sparse suffix sorting problem asks for the order of suffixes starting at certain positions in a text~$T$.
In our case, these positions need only be given online, i.e., sequentially and in an arbitrary order.
We collect them conceptually in a dynamic set~$\Pos$ with $m := \abs{\Pos}$.
The online sparse suffix sorting problem is to keep the suffixes starting at the positions stored in the incrementally growing set~$\Pos$ in sorted order.
Due to the online setting, we represent the order of $\Suf(\Pos)$ by a dynamic, self-balancing binary search tree (e.g., an AVL tree).
Each node of the tree is associated with a distinct suffix in $\Suf(\Pos)$; the lexicographic order is used as the sorting criterion.

The technique of \citet{Irving2003sbs} augments an AVL tree on a set of strings~$\SetStr$ with LCP values so that
$\ell_\textY := \max \{ \lcp(\textX, \textY) \mid \textX \in \SetStr \}$ 
can be computed
in $\Oh{\ell_\textY \wordpack + \lg \abs{\SetStr}}$ time
for a string $\textY$.
Inserting a new string~$\textY$ into the tree is supported in the same time complexity~($\ell_\textY$ is defined as before).
\citeauthor{Irving2003sbs} called this data structure the \intWort{\SAVLT} on $\SetStr$; we denote it by $\SAVL(\SetStr)$.

Remembering \cref{secSparseSuffAlgoOutline}, our goal is to build $\SAVL(\Suf(\Pos))$ efficiently.
However, inserting $m$~suffixes naively takes $\Om{\abs{\Comlcp} m \wordpack + m\lg m}$ time.
How to speed up the comparisons by exploiting a data structure for LCE queries is the topic of this section.

\subsection{Abstract Algorithm}\label{secAbstAlgo}

Starting with an empty set of positions~$\Pos = \emptyset$,
our algorithm updates $\SAVL(\Suf(\Pos))$ on the input of every new text position,
involving LCE computations between the new suffix and suffixes already stored in $\SAVL(\Suf(\Pos))$.
A crucial part of the algorithm is performed by these LCE computations, for which an LCE data structure is advantageous to have.
In particular, we are interested in a \emph{mergeable} LCE data structure that 
	is mergeable in such a way that the merged instance answers queries faster than 
		performing a query on both former instances separately.
		We call this a \intWort{\ac{DynLCE}}; it supports the following operations:

\sitemize{%
\item $\DynLCEInst[\intervalI]$ constructs a \DynLCE{} data structure $M$ on the substring $T[\intervalI]$. Let $M.\memtext$ denote the interval~$\intervalI$.
\item {\tt LCE}$(M_1,M_2, p_1,p_2)$ computes 
	$\lce(p_1, p_2)$, where $p_i \in M_i.\memtext$ for $i=1,2$.
	\item {\tt merge}$(M_1, M_2)$ merges two \DynLCE{}s $M_1$ and $M_2$ such that the output is a \DynLCE{} built
		on the string concatenation of $T[M_1.\memtext]$ and $T[M_2.\memtext]$.
	}
	We use the expression $\timeDynConstruct{\abs{\intervalI}}$ to denote the construction time of such a data structure on the substring $\textT[\intervalI]$.
We assume that the construction of \DynLCEInst[\intervalI] takes at least as long as scanning all characters on \textY{}, i.e.,

\ConstraintBox{%
	\begin{LCEproperty}[series=lceprops]
	\item $\timeDynConstruct{\abs{\intervalI}} = \Om{\abs{\intervalI} \wordpack}$.
		\label{itemDynLCEConstLowerBound}
	\end{LCEproperty}
}%

We use the expressions 
$\timeDynLCE{\abs{\textX}+\abs{\textY}}$ and $\timeDynMerge{\abs{\textX}+\abs{\textY}}$ to denote the time for querying and the time for merging two such data structures built on two given strings $\textX$ and $\textY$, respectively.
Querying two \acp{DynLCE} for a length~$\ell$ 
is faster than the word-packed character comparison
iff $\ell = \Om{ \timeDynLCE{\ell} \lg n / \lg \sigma}$.
Hence, we obtain the following property:

\ConstraintBox{%
	\begin{LCEproperty}[resume*=lceprops]
	\item A \DynLCE{} on a text smaller than $\gapLCE := \Ot{\timeDynLCE{\gapLCE} \lg n / \lg \sigma}$ is always slower than
		the word-packed character comparison.
		\label{itemInterallGap}
	\end{LCEproperty}
}%

In the following, we build \acp{DynLCE} on substrings of the text.
Each interval of the text that is covered by a \DynLCE{} is called an \intWort{\LCEI{}}.
The \LCEI{}s are maintained in a self-balancing binary search tree~\LCEITree{} of size $\Oh{m}$.
The tree~\LCEITree{} stores the starting and the ending positions of each \LCEI{}, and uses the starting positions as keys to answer the queries
\begin{itemize}
	\item whether a position is covered  by a \DynLCE{}, and
	\item where the next text position starts that is covered by a \DynLCE{},
\end{itemize}
in \Oh{\lg m} time.
Additionally, each \LCEI{} is assigned to one \DynLCE{} data structure (a \DynLCE{} can be assigned to multiple \LCEI{}s) such that~\LCEITree{}
 can not only retrieve the next position covered by a \DynLCE{}, but actually return a \DynLCE{} that covers that position.
The \DynLCE{} is retrieved by augmenting an \LCEI{}~$\intervalI$ with
a pointer to
its \DynLCE{} data structure~$M$, and with
an integer~$i$ such that 
$T[M.\memtext \cap [i..i+\abs{\intervalI}-1]] = T[\intervalI]$ (since $M$ could be built on a text interval~$M.\memtext \not= \intervalI$ that contains an occurrence of~$T[\intervalI]$).

\ifthenelse{\boolean{withAlgo}}{%
\begin{algorithm}[t]
\DontPrintSemicolon{} %
Let $\intervalI.\memref$ return the resp. \DynLCE{} of a \LCEI{} $\intervalI$. \;
\Function(\Comment*[f]{called by \cref{algoSparseSuf} to find the insertion position for $\hat{p}$}){\lce}{%
	\Input{text positions $p,p' \in \Pos$}
	$\intervalI \gets \LCEITree.\predecessor[p]$ and
	$\intervalJ \gets \LCEITree.\predecessor[p']$ 
	\Comment*[r]{$\intervalI = \argmax \menge{\ibeg{\intervalI} \mid \intervalI \in \LCEITree \wedge \ibeg{\intervalI} < p}$} 
	\lIf{$\iend{\intervalI} > p$ and $\iend{\intervalJ} > p'$}{%
		$\ell \gets \lce[\intervalI.\memref,\intervalJ.\memref,p,p']$
}
\Else{%
	$\intervalI \gets \LCEITree.\successor[p]$ and
	$\intervalJ \gets \LCEITree.\successor[p']$ 
	\Comment*[r]{$\intervalI = \argmin \menge{\ibeg{\intervalI} \mid \intervalI \in \LCEITree \wedge \ibeg{\intervalI} > p}$} 
	$\ell \gets {\tt naivecmp}(T,p,p',\max\tuple{\ibeg{\intervalI}-p,\ibeg{\intervalJ}-p'})$ \Comment*[f]{compare naively until reaching positions covered by \LCEI{}s}
}
	$p \gets p+\ell+1$ and
	$p' \gets p'+\ell+1$ \;
	\lIf(\Comment*[f]{Found mismatch}){$T[p] \not= T[p']$}{%
		\Return{$\ell$}
	}
	\Return{$\ell + \lce[p,p']$} \Comment*{recurse since we reached either the beginning or the end of an \LCEI{} covered by $\LCEITree$}
}
\caption{Hybrid LCE Algorithm -- Locating Process}
\label{algoHybridLCE}
\end{algorithm}
}{}%

\ifthenelse{\boolean{withAlgo}}{%
\begin{algorithm}[t]
\DontPrintSemicolon %
$S \gets \SAVL(\Suf(\Pos))$ \Comment*{\SAVLT{}}
\Function(\Comment*[f]{online query with a text position $\hat{p}$}){add}{%
	\Input{text position $\hat{p} \in \Int{1}{n} \setminus \Pos$}
	$v \gets S.{\tt root}$ \Comment*{locate insertion position of $T[\hat{p}..]$ in $S$}
	\lWhile(\Comment*[f]{perform on each depth an LCE query with a suffix, use \cref{algoHybridLCE}}){$v$ is not a leaf}{%
		$v \gets S.{\tt locateChild}(v, \hat{p}, \lce)$ 
	}
	$\bar{p} \gets \mlcparg{\hat{p}}$ and
	$\ell \gets {\tt lce}(\bar{p}, \hat{p})$ \Comment*{already computed above}
	$S$.{\tt insert}$(v,\hat{p})$ \Comment*{Add the suffix $T[\hat{p}..]$ to~$S$} 
	\lIf(\Comment*[f]{\Cref{itemIntervallength}}){$\ell < 2\gapLCE$}{%
		\Return{}
	}
	$A \gets \menge{\intervalI \in \LCEITree \mid \intervalI \cap ([\hat{p}-\gapLCE..\hat{p}+\gapLCE+\ell-1] \cup [\bar{p}-\gapLCE..\bar{p}+\ell+\gapLCE-1]) \neq \emptyset}$ \;
	$U \gets \bigcup_{\intervalI \in A} \intervalI$ \Comment*{union over the intervals that are interesting for merging}
	\For{each interval $\intervalI \in ([\hat{p}..\hat{p}+\ell-1] \cup [\bar{p}..\bar{p}+\ell-1]) \setminus U$}{%
		either assign $\intervalI.\memref$ to a \DynLCEInst{} (\Cref{ruleReference} or~\ref{ruleIntersect}) or set $\intervalI.\memref \gets {\DynLCEInst[\substrI{T}{\intervalI}]}$  (\Cref{ruleCreate}) \;
		$\LCEITree$.{\tt insert}($\intervalI$) and 
		$A$.{\tt insert}($\intervalI$) \;
	}
	sort $A$ ascendingly by the starting positions of its intervals \;
$\intervalI := A.{\tt pop}()$ \Comment*{use $A$ like a queue} 
	\While(\Comment*[f]{enforce~\cref{itemIntervalgap}}){$A$ is not empty}{%
		$\intervalJ \gets A.{\tt pop}()$ \;
		\lIf(\Comment*[f]{only if $\intervalI \subset [\hat{p}..\hat{p}+\ell-1], \intervalJ \subset [\bar{p}..\bar{p}+\ell-1]$ or vice versa}){$\iend{\intervalI}+1 \not= \ibeg{\intervalJ}$}{%
			{\bf continue} 
		}
			$\intervalI' \gets \intervalI \cup \intervalJ$ \Comment*[r]{merge both intervals}
			$\intervalI'.\memref \gets {\tt merge}(\intervalI.\memref, \intervalJ.\memref)$ \Comment*[r]{apply \cref{ruleMerge}} 
			$\LCEITree.{\tt delete}(\intervalI)$ and $\LCEITree.{\tt delete}(\intervalJ)$ \;
			$\LCEITree.{\tt insert}(\intervalI')$ \;
			$\intervalI \gets \intervalI'$ \;
	}
	$\Pos \gets \Pos \cup \menge{\hat{p}}$ \;
}
\caption{Sparse Suffix Sorting}
\label{algoSparseSuf}
\end{algorithm}
}{}%

Given a new position~$\hat{p} \not\in \Pos$ with $1 \le \hat{p} \le \abs{T}$,
updating $\SAVL(\Suf(\Pos))$ to $\SAVL(\Suf(\Pos \cup \{ \hat{p} \}))$
involves two parts: first \emph{locating} the insertion node for $\hat{p}$ in $\SAVL(\Suf(\Pos))$, and then \emph{updating} the set of \LCEI{}s.

\block{Locating} 
The insertion operation performs an LCE computation for each node encountered in~$\SAVL(\Suf(\Pos))$ while locating the insertion point of $\hat{p}$. 
Suppose that the task is to compare the suffixes $T[i..]$ and $T[j..]$ for two text positions~$i$ and~$j$ with~$1 \le i,j \le \abs{T}$. 
We perform the following steps to compute $\lce[i,j]$:
\begin{enumerate}[(1)]
	\item Check whether the positions $i$ and $j$ are contained in an \LCEI{}, in \Oh{\lg m} time with the search tree~\LCEITree{}.
		\begin{itemize}
			\item If both positions are covered by \LCEI{}s, 
				then query the respective \DynLCE{}s for the length~$\ell$ of the LCE starting at~$i$ and~$j$.
			Increment~$i$ and~$j$ by~$\ell$.
			Return the number of compared characters on finding a mismatch while computing the LCE.

			\begin{minipage}{0.7\linewidth}
				\item Otherwise (if~$i$ or~$j$ are not contained in an \LCEI{}), 
					find the smallest length~$\ell$ 
					such that~$i+\ell$ and~$j+\ell$ are covered by \LCEI{}s.
					Increment $i$ and $j$ by $\ell$, and naively compare $\ell$ characters.
					Return the number of compared characters on a mismatch. 
			\end{minipage}
			\begin{minipage}{0.3\linewidth}
	\centering{%
		\Bild{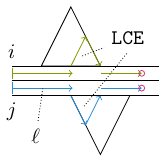}
	}%
			\end{minipage}
		\end{itemize}
	\item Return the total number of matched positions if a mismatch is found in (1). 
		Otherwise, repeat the above check again (with the incremented values of~$i$ and~$j$).
\end{enumerate}
\ifthenelse{\boolean{withAlgo}}{The steps are additionally listed in \cref{algoHybridLCE}.}{}
After locating the insertion point of $\hat{p}$ in $\SAVL(\Suf(\Pos))$, 
we obtain 
$\bar{p} := \mlcparg{\hat{p}}$ and
$\ell := \mlcp{\hat{p}}$
as a byproduct, where
$\mlcparg{p} := \argmax_{p' \in \Pos, p \neq p'} \lcp(T[p..], T[p'..])$ 
and $\mlcp{p} := \lcp(T[p..], T[\mlcparg{p}..])$
for each text position~$p$ with $1 \le p \le \abs{T}$.
We insert~$\hat{p}$ into $\SAVL(\Suf(\Pos))$, and use the position~$\bar{p}$ and the length~$\ell$ to update the \LCEI{}s.

\block{Updating}
The \LCEI{}s are updated dynamically, subject to the following properties (see \cref{figConstraints}):

\ConstraintBox{%
	\begin{LCEproperty}[resume*=lceprops]
	\item The length of each \LCEI{} is at least $\gapLCE$ (defined in \cref{itemInterallGap}). \label{itemIntervallength}
	\item For every $p \in \Pos$, the interval $[p..p+\mlcp{p}-1]$ is covered by an \LCEI{},
	  \emph{except at most} $\gapLCE$ positions at its left and right ends. \label{itemIntervalcover} 
  \item There is a gap of at least $\gapLCE$ positions between every pair of \LCEI{}s. \label{itemIntervalgap}
\end{LCEproperty}
}%

\begin{figure}[t]
	\centering{%
		\Bild{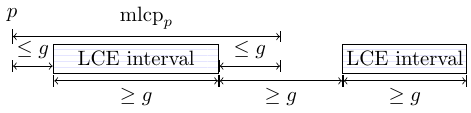}
	}
	\caption{Sketch of two \LCEI{}s with \cref{itemIntervallength,itemIntervalcover,itemIntervalgap}.
	}
	\label{figConstraints}
\end{figure}

\begin{figure}[t]
	\centerline{%
		\Bild{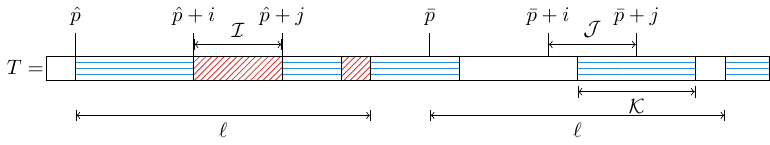}
	}%
	\caption{%
		Application of \namecrefs{ruleReference}~\ref{ruleReference}~to~\ref{ruleMerge} for preserving the properties.
		The interval $\intervalI := [\hat{p}+i..\hat{p}+j]$ is not yet covered by an \LCEI{}, but is contained in 
		$[\hat{p}..\hat{p}+\ell-1]$ --- a conflict with \cref{itemIntervalcover}. 
		The conflict is resolved based on the LCE intervals covering the positions of $\intervalJ := [\bar{p}+i..\bar{p}+j]$.
		The intervals with the\C{~blue} horizontal lines are the LCE intervals, and the intervals with the diagonal\C{~red} lines
		are the intervals of $[\hat{p}..\hat{p}+\ell-1] \setminus U$.
		Here, $\intervalJ$ intersects with an \LCEI{} $\intervalK$. This case is treated in \cref{ruleIntersect}.
	}
	\label{figAlgoRules}
\end{figure}

After adding~$\hat{p}$ to $\Pos$, we perform the following instructions to satisfy the properties.
If $\ell \le 2\gapLCE$, we do nothing, because all properties are still valid (in particular, \cref{itemIntervalcover} still holds).
Otherwise, we need to restore \cref{itemIntervalcover}.  
There are at most two positions in $\Pos$ that possibly invalidate \cref{itemIntervalcover} after adding $\hat{p}$, and these are~$\hat{p}$ and~$\bar{p}$
(otherwise, by transitivity, we would have created a longer \LCEI{} previously).

We introduce an algorithm that does not restore \cref{itemIntervalcover} directly, but first ensures that

\ConstraintBox{%
	\begin{LCEproperty}
\item[\Cref*{itemIntervalcover}':] the intervals~$[\hat{p}..\hat{p}+\ell-1]$ and $[\bar{p}..\bar{p}+\ell-1]$ are covered by one or multiple \LCEI{}s.
	\end{LCEproperty}
}%

In a later step, we restore \cref{itemIntervalcover} by merging \LCEI{}s that are in conflict with \cref{itemIntervalgap}, and thus restore all properties:
Let $U \subset [1..n]$ be the set of all positions that belong to an \LCEI{}.
The set $[\hat{p}..\hat{p}+\ell-1] \setminus U$ can be represented as a set of disjoint intervals of maximal length.
For each interval $\intervalI := [\hat{p}+i..\hat{p}+j] \subset [\hat{p}..\hat{p}+\ell-1]$ of that set,
apply the following rules with $\intervalJ := [\bar{p}+i..\bar{p}+j]$ 
(for integers~$i,j$ with $0 \le i \le j \le \ell-1$, see \cref{figAlgoRules})
sequentially:

\RuleBox{%
\begin{Rules}[series=lcerules]
	\item \label{ruleReference}	If $\intervalJ$ is a sub-interval of an \LCEI{} $\intervalK$, 
		then declare $\intervalI$ as an \LCEI{} and let it refer to the \DynLCE{} of $\intervalK$.
	\item \label{ruleIntersect}
		If $\intervalJ$ intersects with an \LCEI{} $\intervalK$, 
		enlarge the \DynLCE{} on $T[\intervalK]$ to cover $T[\intervalK \cup \intervalJ]$
		(create a \DynLCE{} on $T[\intervalJ \setminus \intervalK]$ and merge it with the \DynLCE{} on $T[\intervalK]$).
		Apply \cref{ruleReference}.
	\item \label{ruleCreate} 
		Otherwise (there is no \LCEI{} \intervalK{} with $\intervalJ \cap \intervalK \neq \emptyset$), create \DynLCEInst[\intervalJ], and make $\intervalI$ and $\intervalJ$ to \LCEI{}s referring to \DynLCEInst[\intervalJ].
\end{Rules}
}%

We satisfy \Cref*{itemIntervalcover}' on $[\bar{p}..\bar{p}+\ell-1]$ by 
updating $U$, 
computing the set of disjoint intervals $[\bar{p}..\bar{p}+\ell-1] \setminus U$, and applying the same rules on it.
However, \cref{ruleReference} or \cref{ruleCreate} can create \LCEI{}s shorter than~$\gapLCE$, violating \cref{itemIntervallength}.
By construction, such a short \LCEI{} is adjacent to another \LCEI{} (the rules compute a cover of $[\hat{p}..\hat{p}+\ell-1]$ and $[\bar{p}..\bar{p}+\ell-1]$ with \LCEI{}s).
This means that we can restore \cref{itemIntervallength} by restoring \cref{itemIntervalgap}.
We do that by applying the following rule subsequently to \cref{ruleCreate}:

\RuleBox{%
\begin{Rules}[resume*=lcerules]
	\item  \label{ruleMerge} 
		Merge a newly created or extended \LCEI{} violating \cref{itemIntervalgap} with its nearest \LCEI{} (ties can be broken arbitrarily).
		Merge those \LCEI{}s and their \DynLCE{}s.
\end{Rules}
}%

\Cref{ruleMerge} also restores \cref{itemIntervalcover} (since \cref*{itemIntervalcover}' and \cref{itemIntervalgap} hold).
After applying all rules, we have introduced at most two%
\footnote{The number of new \LCEI{}s could be indeed two: Although~$\bar{p} \in \Pos$, we would not have created an \LCEI{} covering $[\bar{p}+\gapLCE..\bar{p}+\bar{\ell}-1-\gapLCE]$ if $\mlcp{\bar{p}}$ was smaller than $\gapLCE$ at the time when we inserted~$\bar{p}$ in $\Pos$ with $\bar{\ell} := \mlcp{\bar{p}}$.}
new \LCEI{}s that cover the intervals~$[\hat{p}+\gapLCE..\hat{p}+\ell-1-\gapLCE]$
and~$[\bar{p}+\gapLCE..\bar{p}+\ell-1-\gapLCE]$, respectively, to satisfy \cref{itemIntervallength,itemIntervalcover,itemIntervalgap}. 
\ifthenelse{\boolean{withAlgo}}{\Cref{algoSparseSuf} gives a pseudo code of the algorithm.}{}
The running time of this algorithm is analyzed in the following \lcnamecref{thmAVLAbstractExtraSpace}:

\begin{lemma}\label{thmAVLAbstractExtraSpace}
	Given a text $T$ of length $n$ and a set of $m$ arbitrary positions~$\Pos$ in~$\textT$,
	the \SAVLT{} $\SAVL(\Suf(\Pos))$ with the suffixes of~$T$ starting at the positions~$\Pos$ can be computed deterministically in
  $\Oh{\timeDynConstruct{\abs{\Comlcp}} + \timeDynLCE{\abs{\Comlcp}} m \lg m + \timeDynMerge{\abs{\Comlcp}} m}$ time.
\end{lemma}
\begin{proof}
	The analysis is split into managing the \DynLCE{}s, and the LCE queries:
	\begin{itemize}
		\item 
			We build \DynLCE{}s on substrings covering at most $\abs{\Comlcp}$ characters of the text,
			taking at most $\timeDynConstruct{\abs{\Comlcp}}$~time for constructing all \DynLCE{}s.
			During the construction of the \DynLCE{}s we spend $\Oh{\abs{\Comlcp}\wordpack} = \Oh{\timeDynConstruct{\abs{\Comlcp}}}$ time on naive searches due to \cref{itemDynLCEConstLowerBound}.
		\item 
			The number of merge operations on the \LCEI{}s 
			is upper bounded by $2m$ in total, since we create at most two new \LCEI{}s for every position in $\Pos$.
			In total, we spend at most $2 \timeDynMerge{\abs{\Comlcp}} m$ time for the merging in total.
		\item 
			The algorithm performs $\Oh{m \lg m}$ LCE queries.
			LCE queries involve either (a) naive character comparisons or (b) querying a \DynLCE{}.
			Given that we have $\delta < 2m$ \LCEI{}s, 
			we switch between both techniques at most $4\delta+1$ times for an LCE query.
			\begin{enumerate}[(a)]
				\item	On the one hand, the overall time for the naive character comparisons is bounded by $\Oh{\timeDynConstruct{\abs{\Comlcp}} + \timeDynLCE{\abs{\Comlcp}} m \lg m}$:
					\begin{itemize}
						\item By \cref{itemIntervallength}, all substrings $T[p..p+\mlcp{p}-1]$ are covered by an \LCEI{}, except at most at $2\gapLCE$ positions.
							This means that all substrings that are not covered by an \LCEI{}, but have been subject to a naive character comparison, are shorter than $2\gapLCE$.
							For a naive character comparison with one of those substrings,
							we spend at most $\Oh{\gapLCE m \lg m \wordpack} = \Oh{\timeDynLCE{\gapLCE} m \lg m} = \Oh{\timeDynLCE{\abs{\Comlcp}} m \lg m}$ time.
							In the case that $\gapLCE > \abs{\Comlcp}$, we do not create any \LCEI{}, and spend 
							$\Oh{\abs{\Comlcp}\wordpack + m \lg m} = \Oh{ \timeDynConstruct{\abs{\Comlcp}} + m \lg m}$ overall time due to \cref{itemDynLCEConstLowerBound}.
						\item If we compare more than $\gapLCE$ characters for an LCE query, we create at most two \LCEI{}s, 
							possibly involving the construction of \DynLCE{}s on the compared substrings.
		The construction of a \DynLCE{} on an interval~\intervalI{} takes $\timeDynConstruct{\abs{\intervalI}} = \Om{\abs{\intervalI} \wordpack}$ time
		due to \cref{itemDynLCEConstLowerBound}.
					\end{itemize}
				\item	On the other hand, querying the \DynLCE{}s take at most $\Oh{\timeDynLCE{\abs{\Comlcp}} m \lg m}$ overall time.
					Suppose that we look up $d < \delta$ \LCEI{}s for an LCE query.
					Since we look up an \LCEI{} in $\Oh{\lg m}$ time with \LCEITree{},
					we spend $\Oh{d \lg m}$ time on the lookups during this LCE query.
					However, we subsequently merge all $d$~looked-up \LCEI{}s, reducing the number of \LCEI{}s~$\delta$ by $d-1$.
					Consequently, we perform a look-up of an \LCEI{} at most $2m$ times in total. \qedhere{}
			\end{enumerate}
	\end{itemize}
\end{proof}
The last step is to compute $\SSA[] := \SSA[T,\Pos]$ and $\SLCP[] := \SLCP[T,\Pos]$ from $\SAVL(\Suf(\Pos))$ by traversing $\SAVL(\Suf(\Pos))$ and performing LCE queries on the already computed \DynLCE{}s:
The $\SAVL(\Suf(\Pos))$ is a binary search tree storing all elements of~$\Suf(\Pos)$ in lexicographically sorted order.
This means that we can compute $\SSA[]$ with an in-order traversal of~$\SAVL(\Suf(\Pos))$.
Afterwards, we compute $\SLCP[][i] = \lce[ {\SSA[]}{[i]}, {\SSA[]}{[i-1]} ]$.
Given that the text positions $\Int{\SSA[][i]}{\SSA[][i]+\SLCP[][i]-1]}$ and
\Int{\SSA[][i-1]}{\SSA[][i-1]+\SLCP[][i]-1]} are not covered by an LCE interval, 
$\SLCP[][i] = \Oh{\gapLCE}$ due to \cref{itemIntervallength}, and we spend at most $\Oh{\gapLCE \wordpack}$ time on computing
$\SLCP[][i]$ by naive character comparisons.
Otherwise, we spend $\Oh{\gapLCE\wordpack + \timeDynLCE{{\SLCP[][i]}}} = \Oh{\timeDynLCE{{\SLCP[][i]}}}$ time by querying a \emph{single} \DynLCE{} due to \cref{itemIntervalcover}.
Querying whether both text intervals are covered by an \DynLCE{} costs \Oh{\lg m} time with \LCEITree{}.
In total, we can compute $\SLCP[][i]$ for each integer~$i$ with $2 \le i \le m$ 
in $\Oh{\timeDynLCE{\abs{\Comlcp}} + m \lg m}$ time, since $\Oh{\gapLCE \wordpack} = \Oh{\timeDynLCE{g}}$ due to \cref{itemInterallGap}. 
The following \lcnamecref{thmAbstractExtraSpace} of \cref{thmAVLAbstractExtraSpace} summarizes the achievements of this section:

\begin{corollary}\label{thmAbstractExtraSpace}
	Given a text $T$ of length $n$ that is loaded into RAM,
  the SSA and SLCP of $T$ for a set of $m$ arbitrary positions can be computed deterministically in 
  $\Oh{\timeDynConstruct{\abs{\Comlcp}} + \timeDynLCE{\abs{\Comlcp}} m \lg m + \timeDynMerge{\abs{\Comlcp}} m}$ time.
  We need $\Oh{m}$ words of space, and space to store $\DynLCE{}$ on $\abs{\Comlcp}$ positions.
\end{corollary}

\subsection{Sparse Suffix Sorting with HSP Trees}
\begin{wrapfigure}{r}{20em}
	\centering{%
		\Bild[scale=0.7]{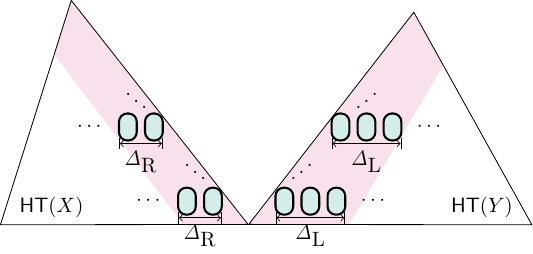}
	}%
\end{wrapfigure}
We show that the HSP tree is a \DynLCE{} data structure.
Remembering that the algorithm from~\cref{secAbstAlgo} depends on the merging operation of \DynLCE{}, 
we now introduce the merging of HSP trees. %
	A naive way to merge two HSP trees \hInst[\textX] and \hInst[\textY] is to
	build \hInst[\textX \textY] completely from scratch.
	Since only the fragile nodes of \hInst[\textX] and \hInst[\textY] can change when merging both trees,
	a more sophisticated approach would reparse only the fragile nodes of both trees.
	Remembering the properties studied in \cref{secChangingNodes}, we show such an approach in the following lemma:

\begin{lemma}\label{lemmaTreeCombine}
	Merging $\hInst[\textX]$ and $\hInst[\textY]$ of two strings $\textX,\textY \in \Sigma^*$ into $\hInst[\textX\textY]$
	takes \Oh{\lookupTime \tuple{\rcontext \lg \abs{\textX} + \lcontext \lg \abs{\textY}}} time.
\end{lemma}
\begin{proof}
	First assume that $\hInst[\textX]$ and $\hInst[\textY]$ only contain \Type{2} nodes.
	In this case, we examine the rightmost nodes of $\hInst[\textX]$ and the leftmost nodes of $\hInst[\textY]$
	from the bottom up to the root:
	At each height $h$, we merge the nodes $\espnode{\textX}_{h}$ and $\espnode{\textY}_{h}$ to $\espnode{\textX\textY}_{h}$
	by reparsing the $\rcontext$ rightmost nodes of $\espnode{\textX}_h$, 
	and the $\lcontext$ leftmost nodes of $\espnode{\textY}_h$. 
	By doing so, we reparse all nodes of \hInst[\textX] (resp.\ \hInst[\textY]) whose local surrounding on the right (resp.\ left) side does not exist.
	Nodes of \hInst[\textX] (resp.\ \hInst[\textY]) that have a local surrounding on the right (resp.\ left) side are not changed by the parsing. 
	In total, we spend \Oh{\lookupTime\tuple{\rcontext \lg \abs{\textX} + \lcontext \lg \abs{\textY}}} time on merging two trees consisting of \Type{2} nodes.

\def\windowpagestuff{\flushright\Bild[width=0.9\textwidth]{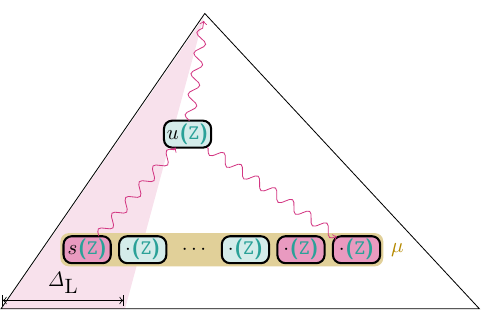}}
\opencutright
\vspace{1em}
\begin{cutout}{1}{0.7\textwidth}{10pt}{7}
	Next, we allow repeating nodes.
	\cref{lemmaSurroundedFragileHSPAppending} shows that there are no fragile surrounded nodes in $\hInst[\textX]$ that need to be fixed.
	The remaining problem is to find and recompute the surrounded nodes in $\hInst[\textY]$ whose names change on merging both trees.
	The lowest of these nodes belong to a repeating meta-block due to \cref{lemTypeTwoStable} and \cref{lemRepeatingFragileSurrounded}.
	To find this meta-block, 
	we adapt the strategy of the first paragraph considering only \Type{2} meta-blocks.
	On each height~$h$, we reparse the $\lcontext$ leftmost nodes of $\espnode{\textY}_h$.
	If the rightmost of these $\lcontext$ nodes are contained in a repeating meta-block~$\mu$ that does not end within those $\lcontext$ leftmost nodes,
	chances are that the names of some nodes in~$\mu$ change. 
	Due to \cref{lemRepeatingFragileSurrounded}, it is sufficient to reparse the two rightmost nodes of~$\mu$.
	This is done as follows:
  \end{cutout}
	\begin{enumerate}
		\item Take the leftmost repetitive node~$s$ of~$\mu$ (which exists due to \cref{corRepeatingPrefixBlock}, and is one of the $\lcontext+1$ leftmost nodes on height~$h$).
		\item Given that $s$ has the surname~\Name{Z}, climb up the tree to find the highest ancestor~$u$ with surname~\Name{Z}.
			The ancestor~$u$ is the lowest common ancestor of~$s$ and the rightmost repetitive node of~$\mu$.
			\item Walk down from~$u$ to the rightmost nodes of~$\mu$.
			\item Reparse $\mu$'s two rightmost nodes.
			\item Reparse all ancestors of these two nodes that are surrounded.
			\item Check whether the reparsed ancestors invalidate the parsing of their meta-blocks; 
				fix the parsing for those meta-blocks recursively.
	\end{enumerate}
	Climbing up to find~$u$ and walking down to the rightmost nodes of~$\mu$ takes $\Oh{\lookupTime \lg\abs{\mu}} = \Oh{\lookupTime \lg (n/2^h)}$ time,
	reparsing the surrounded ancestor nodes of the two rightmost nodes of~$\mu$ takes $\Oh{\lookupTime \lg (n/2^h)}$ time.
	Given that the highest nodes of this reparsing are on a height~$h' > h$, 
	\cref{lemmaSurroundedFragileHSP} states that up to the height~$h'+1$, there is no need to reparse a fragile surrounded node
	(we follow the paths of fragile nodes as depicted in \cref{figHspProof}).
	Given that there are $\mu_1,\ldots,\mu_k$ such meta-blocks (for which we apply Steps~1 to~6),
	we have $\Oh{\lookupTime \sum_{i=1}^k \lg\abs{\mu_i}} = \Oh{\lookupTime \lg n}$ due to $\sum_{i=1}^k \lg \mu_i \le \lg n$.
	Hence, we spend \Oh{(\lcontext + \rcontext) \lookupTime \lg \abs{\textY}} time overall.
\end{proof}

The following theorem combines the results of \cref{thmAbstractExtraSpace} and \cref{lemmaTreeCombine}.

\begin{theorem}\label{thmExtraSpace}
	Given a text $T$ of length $n$ and a set of $m$ text positions $\Pos$,
$\SSA[T,\Pos]$ and $\SLCP[T,\Pos]$ can be computed in 
  $\Oh{\abs{\Comlcp} (\lg^* n + \lookupTime) +  m \lg m \lg n \lg^* n}$ time.
  We need $\Oh{n+m}$ words of space.
\end{theorem}
\begin{proof}
We have 
\begin{itemize}
	\item $\timeDynConstruct{\abs{\Comlcp}} = \Oh{\abs{\Comlcp} \tuple{\lg^*n + \lookupTime }}$ due to \cref{lemmaBuildHSP}, 
	\item $\timeDynLCE{\abs{\Comlcp}} = \Oh{\lg^*n \lg n}$ due to \cref{lemmaLCE}, and
	\item $\timeDynMerge{\abs{\Comlcp}} = \Oh{\lookupTime \lg n \lg^* n}$ due to \cref{lemmaTreeCombine}.
\end{itemize}
Actually, the time cost for merging is already upper bounded by the cost for the tree creation.
To see this, let $\delta \le m$ be the number of \LCEI{}s.
Since each HSP tree covers at least $\gapLCE$ characters, $\delta \gapLCE $ is at most $\abs{\Comlcp}$, and we obtain
$\delta \timeDynMerge{\abs{\Comlcp}} = \Oh{\abs{\Comlcp} \timeDynMerge{\abs{\Comlcp}} /\gapLCE } = \Oh{\abs{\Comlcp} \lookupTime}$ overall time for merging,
where $\gapLCE = \Ot{\timeDynLCE{\abs{\Comlcp}} \lg n / \lg \sigma } = \Ot{\lg^*n \lg^2 n / \lg \sigma}$.
Plugging the times $\timeDynConstruct{\abs{\Comlcp}}$, $\timeDynLCE{\abs{\Comlcp}}$, and the refined analysis of the merging time cost in \cref{thmAbstractExtraSpace} yields the claimed time bounds.
\end{proof}

\section{Sparse Suffix Sorting in Text Space}\label{secSparseSuffixSortingTextSpace}
Remembering the outline in the introduction, the key idea to solve the limited space problem is storing \DynLCE{}s in text space.
Taking two \LCEI{}s of the text containing the same substring, 
we free up the space of \emph{one} part while marking the \emph{other} part as a reference.
The freed space could be used to store an HSP tree whose leaves refer to substrings of the other \LCEI{}.
By doing so, we would use the text space for storing the HSP trees, while using only \Oh{m} additional words for storing $\SAVL(\Suf(\Pos))$ and the search tree~\LCEITree{} of the \LCEI{}s.
However, an HSP tree built on a string of length~$n$ takes \Oh{n \lg n} bits, while the string itself provides only $n\lg \sigma$ bits.
Our solution is to truncate the HSP tree at a fixed height~$\eta$, discarding the nodes in the lower part.
The truncated version~\thInst[\textY] stores just the upper part, while its new leaves refer to (possibly long) substrings of $\textY$.
The resulting tree is called the \intWort{\thTreeF{}~($\text{\thTree{}}_\eta$)}, whose definition follows:

\begin{figure}[t]
	\centering{%
		\Bild{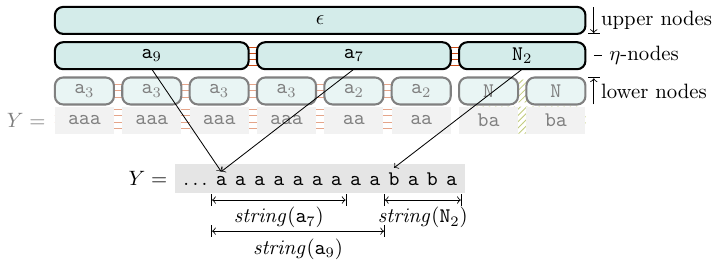}
	}
	\caption{The \thTreeF{} \thInst[\textY] of the substring~$\textY$ defined in \cref{figClassicESP} with $\eta = 2$.
	Like in \cref{figNoInPlaceForESP}, the lower nodes are grayed out.
An $\eta$-node is a leaf in \thInst[\textY], and has a generated substring with a length between four and nine.
}
	\label{figEtatruncatedtree}
\end{figure}

\begin{figure}[t]
	\centering{%
		\Bild[scale=0.8]{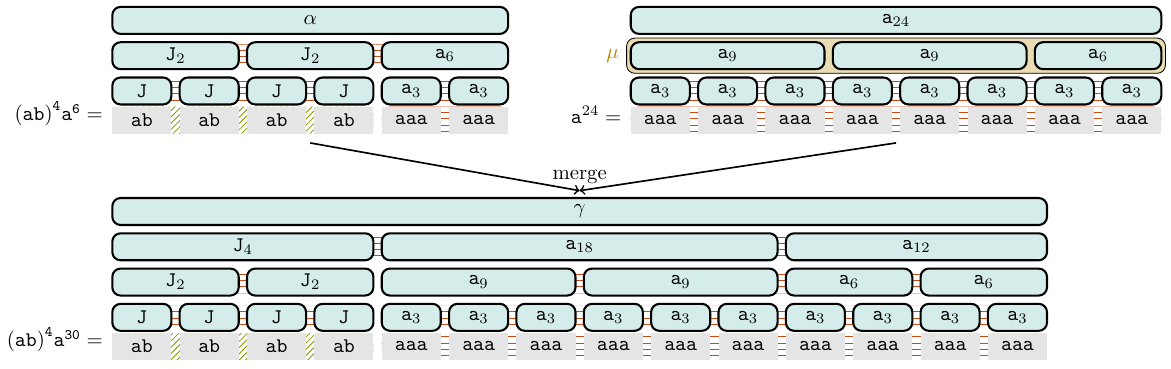}
	}
	\caption{Merging $\hInst[{\tt{(ab)}^4a^6}]$ with $\hInst[{\tt a^{24}}]$ (both at the top) to $\hInst[{\tt {(ab)}^4a^{30}}]$ (bottom tree).
		Reparsing the repeating meta-block~$\mu$ on height one of the right tree is done by rearranging $\mu$'s two rightmost nodes.
			}
\label{figHspmerge}
\end{figure}

\subsection{Truncated HSP Trees}\label{secTSESPImpl}
We define a height $\eta$ and delete all nodes at height less than $\eta$, which we call \intWort{lower nodes}.
A node higher than $\eta$ is called an \intWort{upper node}.
The nodes at height $\eta$ form the new leaves and are called \intWort{$\eta$-nodes}.
Similar to the former leaves, their names are pointers to their generated substrings appearing in $\textY$.
Remembering that each internal node has two or three children, 
an $\eta$-node generates a string of length at least $2^\eta$ and at most $3^\eta$.
The maximum number of nodes in an \thTreeF{} of a string of length $n$ is $\numNodes{n}$. %
\Cref{figEtatruncatedtree} shows an example with $\eta = 2$.

Similar to leaves in untruncated HSP trees, we use the generated substring $\textX$ of an $\eta$-node~$v$ for storing and looking up $v$:
While the leaves of the HSP tree have a generated substring of constant size (two or three characters), 
the generated substring of an~$\eta$-node can be as long as $3^\eta$.
Storing such long strings in a binary search tree representing the reverse dictionary of~$\dic$ is inefficient; 
it would need $\Oh{\ell \lg \sigma}$ time for a lookup or insertion of a key of length~$\ell$.
Instead, we want a dictionary data structure storing $\Oh{\abs{\textY}}$ elements in $\Oh{\abs{\textY}}$ words of space%
\footnote{The data structure is not necessarily stored in consecutive space like an array.},
supporting lookup and insert in $\Oh{\lookupTime + \ell\wordpack}$ time for a key of length $\ell$.
For instance, \citeauthor{betterSearching}'s data structure~\cite{betterSearching} with word-packing supports the desired time and space bounds with $\lookupTime = \Oh{\lookupTimeBest}$.

\begin{lemma}\label{lemmaBuildTSESP}
	We can build an \thTreeF{}~$\thInst[\textY]$ of a string $\textY$ of length $n$ in 
	$\Oh{n \tuple{\lg^*n + \eta\wordpack + \lookupTime/2^\eta }}$
time, using
$\Oh{3^\eta \lg^* n}$ words of working space.
The tree takes $\Oh{\numNodes{n}}$ words of space.
\end{lemma}
\begin{proof}
Instead of building the HSP tree level by level, we compute the $\eta$-nodes one after another, from left to right.
We can split the parsing of the whole string into several parts.
Each part computes one $\eta$-node.

First assume that $\thInst[\textY]$ only contains \Type{2} nodes.
Then the name of an $\eta$-node~$v$ is determined by $v$'s local surrounding (as far as it exists) due to \cref{lemTypeTwoStable}.
This means that it is sufficient to keep $v$'s local surrounding at height~$\eta-1$, which we denote by $\textX_v$, in memory.
$\textX_v$ is a string of lower nodes.
To parse a string of lower nodes by HSP, we have to give each lower node a name.
Unfortunately, storing the names of all lower nodes in a dictionary would take too much space.
Instead, we create the name of a lower node temporarily by setting the name of a lower node to its generated substring.
This means that we cannot retrieve their names later.
Luckily, we only need the names of the lower nodes for constructing~$\textX_v$.
We construct~$\textX_v$ as follows:
Given that we parsed the local surrounding of~$v$ at height~$h$ ($0 \le h \le \eta-3$) with HSP,
we store the borders of the blocks on height~$h+1$ in an integer array 
such that we can access the name (i.e., the generated substring) of the $i$-th block on height~$h+1$.
With this integer array, we can parse the blocks on height~$h+1$ to obtain the blocks on height~$h+2$, 
whose borders are again stored in an integer array.
Having the borders of the blocks on height~$h+2$, we can remove the integer array on height~$h+1$.
The blocks on height~$\eta-1$ are the nodes of~$\textX_v$.

In the general case (when $\thInst[\textY]$ contains repeating nodes), 
it can happen that the name of a greedily parsed node (i.e., a repeating node or one of the $\lcontext$ leftmost nodes of a \Type{2} meta-block)
depends not necessarily on its local surrounding, 
but on the length of its repeating meta-block, its surname and its children (in case of a \Type{M} node).
This means that when computing~$\textX_v$ of an $\eta$-node~$v$, we additionally have to consider the case when nodes
in the local surrounding of~$v$ are contained in a meta-block~$\mu$ on height~$h < \eta$ that extends over the nodes in $v$'s surrounding at height~$h$.
It is sufficient to use a counting variable that tracks the position of the last block of~$\mu$ belonging to the subtree
of the preceding $\eta$-node of~$v$ (remember that the greedy parsing determines the blocks by an arithmetic progression).
Another necessity is to maintain the surnames of the lower nodes. 
In our approach, each array storing the borders of the blocks on the heights below~$\eta$ is
accompanied with two arrays. 
The first array stores the length of the prefix of the generated substring of each block~$\beta$ that is equal to $\beta$'s surname;
the second array stores the surname-length of each block.

\block{Working Space}
We compute~$v$ after computing~$\textX_v$.
To compute~$\textX_v$, we apply the HSP technique $(\eta-1)$-times on the generated substring of the nodes in~$\textX_v$.
Since the nodes of $\textX_v$ cover at most $3^\eta (\lcontext+\rcontext)$ characters,
we need $\Oh{3^\eta (\lcontext+\rcontext)}$ words of working space to maintain the integer arrays 
storing the borders of the blocks at two consecutive heights. 
To cope with the meta-blocks extending over the border of the subtrees of two $\eta$-nodes, 
we store the last position of each such meta-block belonging to the local surrounding of the previous $\eta$-node.
These positions take \Oh{\eta} words, since such a meta-block can exist on every height below~$\eta$.

\block{Time} 
The time bound $\Oh{n \lg^* n}$ for the repeated application of the alphabet reduction is the same as in \cref{lemmaBuildHSP}.
The new part is the construction of an~$\eta$-node by constructing~$\textX_v$:
To construct the lower nodes~$X_v$, we apply the HSP technique $(\eta-1)$-times on $\generated{v}$.
The HSP technique compares lower nodes by their generated substrings (instead of comparing by a name stored in $\dic$).
It always compares two adjacent lower nodes during the construction of~$X_v$.
To bound the number of comparisons of the lower nodes, we focus on all lower nodes on a fixed height~$h$ with~$1 \le h \le \eta-1$:
Since the sum of the lengths of the generated substrings of the lower nodes on height~$h$ is always~$n$,
the comparisons of the lower nodes on height $h$ take \Oh{n \wordpack} time, independent of the number of nodes on height~$h$.
Summing over all heights, these comparisons take $\Oh{n \eta \wordpack}$ time in total.
By the same argument,
maintaining the names of all $\eta$-nodes takes $\Oh{n\wordpack + \lookupEta \numNodes{n} }$ time.

A name is looked-up in $\Oh{\lookupUpper}$ time for an upper node.
Since the number of upper nodes is at most $\numNodes{n}$,
maintaining the names of the upper nodes takes $\Oh{\lookupUpper \numNodes{n}}$ time.
This time is subsumed by the lookup time for the $\eta$-nodes.

\block{Surnames}
Augmenting the (remaining) nodes of the \thTreeF{} with surnames cannot be done as trivially as in the standard HSP tree construction,
since a repetitive node can have a surname equal to the name of a lower node (remember that lower nodes are generated only temporarily, 
and hence are not maintained in the reverse dictionary).
To maintain the surnames pointing to lower nodes, we need to save the names of certain lower nodes in a supplementary reverse dictionary~$\dic'$ of~$\dic$.
This is only necessary when one of the remaining nodes (i.e., the upper nodes and the $\eta$-nodes) in the \thTreeF{} has a surname that is the name of a lower node.
If such a remaining node~$v$ is an upper node having a surname equal to the name of a lower node, the $\eta$-nodes in the subtree rooted at~$v$ have also the same surname.
Hence, the number of entries in~$\dic'$ is upper bounded by the number of $\eta$-nodes.
The dictionary~$\dic'$ is filled with the surnames of the children of all $\eta$-nodes, whose number is at most $3n/2^\eta$.
Filling or querying~$\dic'$ takes the same time as maintaining the $\eta$-nodes.
\end{proof}

Similar to the standard HSP trees, we can conduct LCE queries on two \thTreeF{}s in the following way:

\begin{lemma}\label{lemmaTruncatedLCE}
	Let $\textX$ and $\textY$ be two strings with $\abs{\textX}, \abs{\textY} \le n$.
	Given that $\thInst[\textX]$ and $\thInst[\textY]$ are built with the \emph{same} dictionary,
	and given two text positions $1 \le i_\textX \le \abs{\textX}, 1 \le i_\textY \le \abs{\textY}$,
	we can compute $\lcp(\textX[i_{\textX}..], \textY[i_{\textY}..])$ in 
	$\Oh{\lg^*n (\lg (\numNodes{n}) + 3^\eta \wordpack )}$ time
	using \Oh{\lg (\numNodes{n})} words of working space.
\end{lemma}
\begin{proof}
\Cref{lemmaLCE} gives the time bounds for computing the longest common prefix with two HSP trees.
The \lcnamecref{lemmaLCE} describes an LCE algorithm that uses the surnames to compare the generated substring of two nodes.
By doing so, it accelerates the search for the first pair of mismatching characters in $\textX[i_{\textX}..]$ and $\textY[i_{\textY}..]$.
To find this mismatching pair, it examines the subtrees of the two nodes if both nodes mismatch.
Since we cannot access a child of an $\eta$-node in our \thTreeF{}s without rebuilding its subtree
(as we do not store the lower nodes in~$\dic$),
we treat the $\eta$-nodes as the leaves of the tree.
This means that we compare two $\eta$-nodes (given their surnames are different) with a naive comparison of their generated substrings in \Oh{3^\eta \wordpack} time,
remembering that the length of the generated substring of an $\eta$-node is at most $3^\eta$.
For the upper nodes, the algorithm works identically to the original version such that it takes \Oh{\lg^*n (\lg (\numNodes{\ell})} time for  
traversing those.
\end{proof}
Applying the idea of \cref{lemmaLCEPointer} to \cref{lemmaTruncatedLCE} gives the following \lcnamecref{lemmaTruncatedLCEPointer}:

\begin{corollary}\label{lemmaTruncatedLCEPointer}
	Let $\textX$ and $\textY$ be two strings with $\abs{\textX}, \abs{\textY} \le n$.
	Given that $\thInst[\textX]$ and $\thInst[\textY]$ are built with the \emph{same} dictionary,
	we can augment both trees with a data structures such that
	given two text positions $1 \le i_\textX \le \abs{\textX}, 1 \le i_\textY \le \abs{\textY}$,
	we can compute $\ell := \lcp(\textX[i_{\textX}..], \textY[i_{\textY}..])$ in 
	$\Oh{\lg^*n (\lg (\numNodes{\ell}) + 3^\eta \wordpack )}$ time
	using \Oh{\lg (\numNodes{n})} words of working space.
	The additional data structures can be constructed in \Oh{n} time with \Oh{n/\lg n} words of space.
	Their space bounds are within the space bounds of the HSP trees.
\end{corollary}
\begin{proof}
	To support accessing the parent of a node in constant time, we construct a pointer based tree structure of the truncated tree
during its construction.
Since $\thInst[\textY]$ contains at most $n/2^\eta$ nodes, the pointer based tree structure takes \Oh{n/2^\eta} words.

Given that $\eta \le \lg \lg n$,
we augment the tree structure with a bit vector to jump from a text position to an $\eta$-node like in \cref{lemmaLCEPointer}:
We create a bit vector of length~$n$ marking the borders of the generated substrings of the $\eta$-nodes such that a rank-support data structure
on this bit vector allows us to jump from a position~$\textY[i]$ to the $\eta$-node~$\espnode{\textY}_\eta[j]$ 
with 
$1+\sum_{k=1}^{j-1} \generated{\espnode{\textY}_\eta[k]} \le i \le \sum_{k=1}^j \generated{\espnode{\textY}_\eta[k]}$ in constant time.
The bit vector with rank-support takes $\Oh{n/\lg n}$ words, which is too much to obtain the space bounds of \Oh{\numNodes{n}} words
when $\eta = \Om{\lg \lg n}$.

Instead, we compute a sorted list of pairs if $\eta \ge \lg_3 (\lg^2 n)$.
During the construction of a truncated tree, we collect pairs of constructed $\eta$-nodes and their starting positions in a list.
This list is automatically sorted by the starting positions as we construct the tree from left to right.
The list takes \Oh{n/2^\eta} words, and we can find the $\eta$-node whose generated substring covers a given position in $\Oh{\lg (n/2^\eta)} = \Oh{\lg n}$ time by binary searching the starting positions.
This time is bounded by the time $\Oh{\lg^* n\; 3^\eta / \lg_\sigma n}$ for scanning the generated substrings of all $\eta$-nodes during an LCE query, which is 
$\Oh{\lg^* n \lg n \lg \sigma}$ time when $\eta \ge \lg_3 (\lg^2 n)$.

It is left to consider the case that $\lg \lg n < \eta < \lg_3 \lg^2 n$.
\newcommand{\kvar}{k}
Let $\kvar$ be the number of $\eta$-nodes such that $n/3^\eta \le \kvar \le n/2^\eta$.
We build the above bit vector in the representation of~\citet{pagh01low}.
In this representation, the rank-support answers rank queries in constant time.
The bit vector together with its rank-support takes 
\Oh{\kvar \lg (n/\kvar) + \kvar^2/n + \kvar (\lg \lg \kvar)^2/\lg \kvar} = \Oh{\kvar \eta} bits
(which are \Oh{n/2^\eta} words)
when $\kvar = n / \lg^c n$ for a constant~$c > 0$~\cite[Theorem 4(b)]{Rahman08bv}.
The constant~$c$ exists, because $n/\lg^2 n < n/3^\eta \le \kvar \le n/2^\eta < n/\lg n$.
However, the construction needs $\Oh{n/\lg n}$ words of space.
\end{proof}

With $\tau := 2^\eta$ we obtain the claim of \cref{thmLCEtradeoff}.

\begin{remark}
	In the following, we stick to the result obtained in \cref{lemmaTruncatedLCE} instead of \cref{lemmaTruncatedLCEPointer}.
	Although \cref{lemmaTruncatedLCE} has a slower running time for longest common prefixes that are short,
	the additional rank-support data structures of \cref{lemmaTruncatedLCEPointer} makes it difficult to achieve our aimed running time for merging two trees (and therefore would restrain us from achieving our final goal stated in \cref{thmSparseSuffixSorting}).
	To merge two trees, where each tree is augmented with the bit vector and its rank-support data structure, 
	the task would be to build a rank-support data structure on the concatenation of the bit vectors (preferably in logarithmic time).
	Unfortunately, we are not aware of a rank-support data structure that is efficiently mergeable (a naive way would be to build the rank-support data structure of the large bit vector from scratch in linear time).
\end{remark}

\begin{figure}[t]
	\centering{%
		\Bild[scale=1.0]{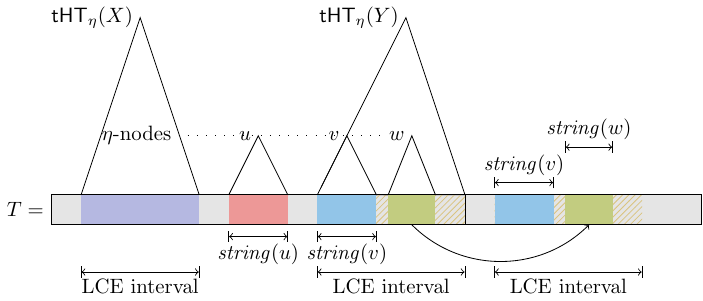}
	}
	\caption{Problem with generated substrings when merging \thInst[\textX] and \thInst[\textY]. 
	Assume that we want to merge \thInst[\textX] and \thInst[\textY], and thus compute the bridging $\eta$-nodes (like~$u$) between both trees.
	On the one hand, the generated substrings of the non-surrounded $\eta$-nodes (like~$v$) and of the bridging nodes are marked protected, because we cannot find a surrogate substring in general. 
Although there is a second occurrence of~$\generated{v}$ to the right, $\generated{v}$ can be extended or shortened when prepending characters 
(e.g., suppose that $\generated{v} = {\tt a}^k$, and that there is an {\tt a} to the left of the left occurrence of $\generated{v}$, but not to the left of the right occurrence).
		On the other hand, the space of the recyclable interval can be used for storing the \thTreeF{}s, because here we find suitable surrogate substrings for the generated substrings of the $\eta$-nodes (like for~$w$).
	}
	\label{figTextSpaceMerge}
\end{figure}

\subsection{Sparse Suffix Sorting with \thTree{}s}
To use the \thTreeF{}s as \DynLCE{}s stored \emph{in text space}, we have to think about how to merge them.
Like with HSP trees, merging two \thTreeF{}s involves a reparsing of the nodes at the facing borders~(cf.~\cref{figHspmerge}).
However, the reparsing of the $\eta$-nodes on that borders is especially problematic, as can be seen in \cref{figTextSpaceMerge}:
Suppose that we rename an $\eta$-node~$v$ from~\HSPname{$(ba)^2$} to~\HSPname{$(ba)^3$} with $\abs{\generated{\HSPname{$(ba)^2$}}} < \abs{\generated{\HSPname{$(ba)^3$}}}$.
If the name~\HSPname{$(ba)^3$} is not yet maintained in the dictionary, we have to create~\HSPname{$(ba)^3$}, i.e., a pointer to a substring~$\textX$ of the text with $\textX = \generated{\HSPname{$(ba)^3$}}$.
The critical part is to find~$\textX$ in the not-yet overwritten parts of the text:
Although we can create a suitably long string containing~$\textX$ 
by concatenating the generated substrings of $v$'s preceding and succeeding siblings,
these $\eta$-nodes may point to text intervals that are not consecutive.
Since the name of an $\eta$-node is the representation of a \emph{single} substring,
we would have to search $\textX$ in the \emph{entire} remaining text.
In the case that $v$ is surrounded, \cref{lemmaSurroundedFragileHSP} shows that \textX{} is a prefix of the generated substring of a sibling $\eta$-node
(unlike in \cref{figNoInPlaceForESP}, where the generated substring of the ESP node with name~\ESPname{$a^2(ba)^2$} cannot be easily determined).
With this insight, we finally show an approach that proves \cref{thmSparseSuffixSorting}.
For that, it remains to implement \cref{ruleCreate} and \cref{ruleMerge} from \cref{secAbstAlgo}
in the context that we maintain \thTreeF{}s \emph{in text space}:
We explain
\begin{Goals}
\item how the parameter~$\eta$ has to be chosen such that $\thInst[\textY]$ fits into $\abs{\textY} \lg \sigma$ bits (needed for \cref{ruleCreate}), and \label{goalInPlaceFixEta}
\item how to merge two \thTreeF{}s without the need for extra working space (needed for \cref{ruleMerge}). \label{goalInPlaceExtraSpace}
\end{Goals}

Our first goal is to store $\thInst[\substrI{T}{\intervalI}]$ in a text interval $\intervalI$.
Since \thInst[\substrI{T}{\intervalI}] can contain nodes with $\abs{\intervalI}/2^\eta$~distinct names, 
it requires $\Oh{\abs{\intervalI}/2^\eta}$ words, i.e., $\Oh{\abs{\intervalI} \lg n/2^\eta}$ bits of space
that might not fit in the $\abs{\intervalI} \lg \sigma$~bits of $T[\intervalI]$.
Declaring a constant $\alpha$ (independent of $n$ and $\sigma$, but dependent on the size of a single node),
we can solve this space issue by setting $\eta := \log_3 (\alpha \lg^2 n/\lg \sigma)$:

\begin{lemma}\label{lemmaEta}
	The number of nodes of an \thTreeF{} on a substring of length~$\ell$ is bounded by $\Oh{\ell (\lg\sigma)^{0.7} / (\lg n)^{1.2}}$ with $\eta = \log_3 (\alpha \lg^2 n/\lg \sigma)$.
\end{lemma}
\begin{proof}
To obtain the upper bound on the number of nodes, we first compute a lower bound on the number of bits taken by the generated substring of an $\eta$-node, which is already lower bounded by $2^\eta \lg \sigma$ bits.
We begin with changing the base of the logarithm from~$3$ to $2/3$, and reformulate
$\eta = \log_3(\alpha \lg^2 n / \lg\sigma) = (\log_3 2 - 1) \log_{2/3}(\alpha \lg^2 n / \lg\sigma) = \log_{2/3}(\alpha \lg^2 n / \lg\sigma)^{\log_3 2 - 1}$.
This gives 
\[ 2^{\eta} \lg\sigma = 3^{\eta} (2/3)^{\eta} \lg\sigma = \alpha (\alpha \lg^2 n / \lg\sigma)^{\log_3 2 - 1} \lg^2 n = (\alpha^{\log_3 2}) (\lg n)^{2 \log_3 2} (\lg \sigma)^{1-\log_3 2}.  \]
With the estimate $0.6 < \log_3 2 < 0.7$ we simplify this to
\[ (\alpha^{\log_3 2}) (\lg n)^{2 \log_3 2} (\lg \sigma)^{1-\log_3 2} > \alpha^{0.6} (\lg n)^{1.2} (\lg \sigma)^{0.3}. \]
Hence, the generated substring of an $\eta$-node takes at least $2^{\eta} \lg\sigma \ge \alpha^{0.6} (\lg n)^{1.2} (\lg \sigma)^{0.3}$ bits.

Finally, the number of nodes is bounded by 
\[
	\numNodes{\ell} \le \ell \lg\sigma / (\alpha^{0.6} (\lg n)^{1.2} (\lg\sigma)^{0.3}) = \ell (\lg\sigma)^{0.7} / \tuple{\alpha^{0.6} (\lg n)^{1.2}}. \qedhere{}
\]
\end{proof}
A consequence is that an $\eta$-node with $\eta = \log_3 (\alpha \lg^2 n/\lg \sigma)$ generates a substring containing at most $3^\eta = \alpha (\lg n)^2 / (\lg\sigma)$ characters.

Plugging this value of $\eta$ in \cref{lemmaBuildTSESP} and \cref{lemmaTruncatedLCE} yields two corollaries for the \thTreeF{}s:
\begin{corollary}\label{corTSESPTextSpace}
	We can compute an \thTreeF{} on a substring of length~$\ell$ 
	in $\Oh{\ell \lg^*n + \lookupTime \numNodes{\ell} + \ell \lg\lg n }$~time.
	The tree takes $\Oh{\numNodes{\ell}}$ words of space. We need a working space of $\Oh{\lg^2 n \lg^*n / \lg \sigma}$ characters.
\end{corollary}
\begin{proof}
	The tree has at most $\numNodes{\ell}$ nodes, and thus takes 
	$\Oh{\numNodes{\ell}}$ words of space.
	According to \cref{lemmaBuildTSESP},
	constructing an $\eta$-node uses $\Oh{3^\eta \lg^*n} = \Oh{\lg^2 n \lg^*n / \lg \sigma}$ characters as working space.
\end{proof}

\begin{corollary}\label{corTSESPLCE}
	An LCE query on two \thTreeF{}s can be answered in 
	$\Oh{\lg^*n \lg n}$ time.
\end{corollary}
\begin{proof}
	LCE queries are answered as in \cref{lemmaTruncatedLCE}, where the time bound depends on~$\eta$.
Since an $\eta$-node generates a substring of at most $3^{\eta} = \alpha \lg^2 n/\lg \sigma$ characters, 
we can compare the generated substrings of two $\eta$-nodes in \Oh{\alpha \lg n} time. 
Overall, we compare $\Oh{\lg^* n}$ many times two $\eta$-nodes, such that these additional costs are bounded by $\Oh{\lg^* n \lg n}$ time overall,
and do not slow down the running time 
$\Oh{ \lg^* n \lg(n/2^\eta) + \lg^* n \lg n } = \Oh{\lg^*n \lg n}$.
\end{proof}

Our second and final goal is to adapt the merging used in the sparse suffix sorting algorithm (\cref{secAbstAlgo}). 
Suppose that our algorithm finds two intervals $[i..i+\ell-1]$ and $[j..j+\ell-1]$ with $T[i..i+\ell-1] = T[j..j+\ell-1]$.
Ideally, we want to construct \thInst[\substr{T}{i}{i+\ell-1}] in the text space $[j..j+\ell-1]$,
leaving $T[i..i+\ell-1]$ untouched so that parts of this substring can be referenced by the $\eta$-nodes.
Unfortunately, \namecrefs{ruleReference}~\ref{ruleReference}~to~\ref{ruleMerge} cannot be applied directly due to our working space limitation.
Since we additionally use the text space as working space, we have to be careful about what to overwrite.
In particular, we focus on how to
\begin{enumerate}[(a)]
	\item partition the \LCEI{}s such that the generated substrings of the fragile non-surrounded $\eta$-nodes are protected from becoming overwritten,\label{itTextSpaceMergeFragileNodes}
	\item keep enough working space in text space available for merging two trees,\label{itTextSpaceMergeMemory}
	\item construct \thInst[\substr{T}{i}{i+\ell-1}] in the text space $[j..j+\ell-1]$ when the intervals $[i..i+\ell-1]$ and $[j..j+\ell-1]$ overlap, and how to \label{itTextSpaceMergeOverlap}
	\item bridge the gap~$T[\iend{\intervalI}+1..\ibeg{\intervalJ}-1]$ when 
		merging \thInst[ {\textT[\intervalI]} ] and \thInst[ {\textT[\intervalJ]} ] to \thInst[ {\textT[\ibeg{\intervalI}..\iend{\intervalJ}]} ]
		for two intervals~$\intervalI$ and~$\intervalJ$ with $\ibeg{\intervalI} < \ibeg{\intervalJ}$ and $\abs{\Int{\iend{\intervalI}+1}{\ibeg{\intervalJ}-1}} < \gapLCE$, as performed in \cref{ruleMerge}. \label{itTextSpaceMergeGap}
	\end{enumerate}

	\block{\ref{itTextSpaceMergeFragileNodes} Partitioning of \LCEI{}s}
	To merge two \thTreeF{}s, we have to take special care of those $\eta$-nodes that are fragile, 
because their names can change due to a merge.
If the parsing changes the name of an $\eta$-node~$v$, we first check whether $v$'s new name is present in the dictionary.
If it is not, we have to create $v$'s new name consisting of a text position~$i$ and a length~$\ell$ such that $T[i..i+\ell-1] = \generated{v}$.
The new name of a fragile \emph{surrounded} $\eta$-node $v$ can be created easily:
According to \cref{lemmaSurroundedFragileHSP}, the generated substring of $v$ is always
a prefix of the generated substring of an already existing $\eta$-node $w$, which is found in the reverse dictionary of the $\eta$-nodes.
Hence, we can create a new name of $v$ with \generated{w}.

Unfortunately, the same approach does not work with the non-surrounded $\eta$-nodes, because
those nodes have generated substrings that are found at the borders of $T[j..j+\ell-1]$ (remember \cref{figTextSpaceMerge}).
If the characters around the borders are left untouched (meaning that we prohibit overwriting these characters), 
they can be used for creating the names of the fragile non-surrounded $\eta$-nodes during a reparsing.
To prevent overwriting these characters, we mark both borders of the interval $[j..j+\ell-1]$ as protected.
Conceptually, we partition an \LCEI{} into (1) \intWort{recyclable} and (2) \intWort{protected} intervals (see~\cref{figMargin});
we free the text of a recyclable interval for overwriting, while prohibiting write access on a protected interval.
The recyclable intervals are managed in a dynamic, global list.
We keep the property that

\ConstraintBox{%
	\begin{LCEproperty}[resume*=lceprops]
	\item $\marginLCE := \upgauss{2 \alpha \lg^2 n \lcontext / \lg \sigma} = \Ot{g}$ text positions
		of the left and right ends of each \LCEI{} are \emph{protected}. \label{invRedInterval}
\end{LCEproperty}
}%

This property solves the problem for the non-surrounded nodes, 
because a non-surrounded $\eta$-node has a generated substring that is found in $T[j..j+\marginLCE-1]$ or $T[j+\ell-1-\marginLCE..j+\ell-1]$.

\begin{figure}[t]
\centering{%
	\Bild[scale=1.0]{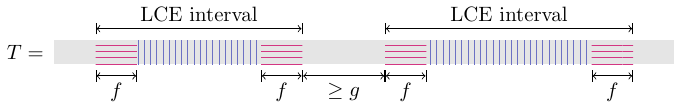}
}
\caption{Division of \LCEI{}s in protected (horizontal\C{~magenta} lines) and recyclable (vertical\C{~violet} lines) parts.} 
\label{figMargin}
\end{figure}

\block{\ref{itTextSpaceMergeMemory} Reserving Text Space}
We can store the upper part of the \thTreeF{} in a recyclable interval, because it needs
$\numNodes{\ell} \lg n \le \ell \alpha^{0.6} (\lg\sigma)^{0.7} / (\lg n)^{0.2} = \oh{\ell \lg \sigma}$ bits.
Since $\marginLCE$ depends on $\alpha$ and $\gapLCE$, we can choose $\gapLCE$ (the minimum length of a substring on which an \thTreeF{} is built) and $\alpha$ 
(relative to the number of words taken by a single \thTreeF{} node) appropriately 
to always leave $\marginLCE \lg \sigma /\lg n = \Oh{\lg^*n \lg n}$ words on a recyclable interval untouched, sufficiently large for the working space needed by \cref{corTSESPTextSpace}.
Therefore, we precompute $\alpha$ and $\gapLCE$ based on the input text~$T$, and set both as \emph{global} constants dependent on~$\textT$.
Since the same amount of free space is needed during a subsequent merging when reparsing an $\eta$-node, we add the following property:

\ConstraintBox{%
	\begin{LCEproperty}[resume*=lceprops]
		\setcounter{enumi}{5}
	\item Each \LCEI{} has $\marginLCE \lg \sigma /\lg n$ words of free space left on a recyclable interval. \label{invWorkingSpace}
\end{LCEproperty}
}%

  \begin{figure}[t]
  	\centering{%
		\Bild[valign=t]{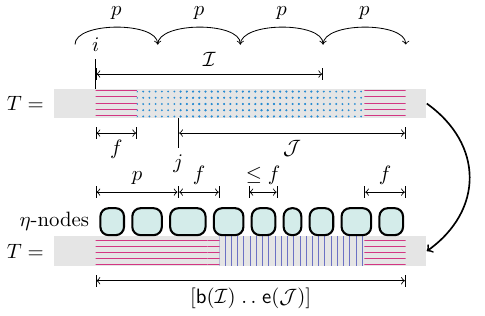}
		\hfill
		\Bild[valign=t]{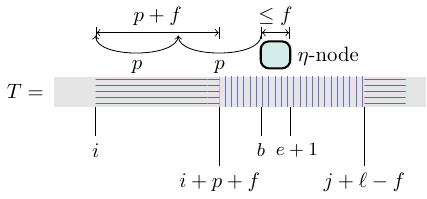}
  	}
	\caption{\emph{Left}: Overlapping \LCEI{}s $\intervalI = [i..i+\ell-1]$ and $\intervalJ = [j..j+\ell-1]$.
		\emph{Right}: Finding the generated substring~$T[b..e]$ of an $\eta$-node in a protected interval.
		Given that $p$ is the smallest period of $T[\intervalI\cup\intervalJ]$,
		it is sufficient to make $\marginLCE+\period$ characters on the left protected
		to find the generated substring of all $\eta$-nodes of \thInst[ {T[i..j+\ell-1]} ] in $T[i..i+\period+\marginLCE-1]$.
	}
	\label{figMergePeriod}
  \end{figure}

  \begin{figure}[t]
  	\centering{%
		\Bild{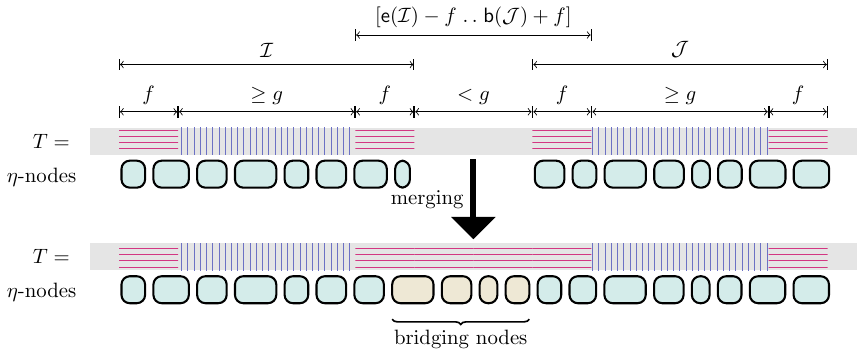}
  	}
	\caption{Merging \thInst[ {T[\intervalI]} ] and \thInst[ {T[\intervalJ]} ] with $\ibeg{\intervalJ} - \gapLCE \le \iend{\intervalI} \le \ibeg{\intervalJ}-1$.
		The substring $T[\iend{\intervalI}-\marginLCE..\ibeg{\intervalJ}+\marginLCE]$ is marked protected for the sake of the bridging nodes.
	}
	\label{figEspMerge}
  \end{figure}

In our algorithm for sparse suffix sorting, a special problem emerges when two computed \LCEI{}s overlap.
For instance, this can happen when the \ac{LCE} of a position~$i \in \Pos$ with a position~$j \in \Pos$ overlaps, i.e.,
$[i..i+\lce[i,j]-1]
\cap
[j..j+\lce[i,j]-1] 
\not= \emptyset$.
The algorithm would proceed with merging both overlapping \LCEI{}s to satisfy \cref{itemIntervalgap}.
However, the merged \LCEI{} cannot respect \namecref{invRedInterval}~\ref{invRedInterval} and~\ref{invWorkingSpace} in general (consider that each interval has a length of $3\gapLCE$, and both intervals overlap with $2\gapLCE$ characters).
In the case of overlapping, we exploit the periodicity caused by the overlap to make an \thTreeF{} fit into both intervals 
(while still assuring that \cref{itemIntervalcover} and \cref{itemIntervalgap} hold, and that we can restore the text).

\block{\ref{itTextSpaceMergeOverlap} Interval Overlapping}
In a more general setting, suppose that the intervals $\intervalI := [i..i+\ell-1]$ and $\intervalJ := [j..j+\ell-1]$ with $T[\intervalI] = T[\intervalJ]$
overlap, without loss of generality $i < j$.
Given $\ell > 2\gapLCE$, our task is to create \thInst[\substr{T}{i}{j+\ell-1}] (e.g., needed to comply with \cref{itemIntervalcover}).
Since $T[\intervalI] = T[\intervalJ]$, the substring $T[i..j+\ell-1]$ has a period~$\period$ with $1 \le \period \le j-i$, i.e.,
$T[i..j+\ell-1] = \textX^k \textY$, where $\abs{\textX} = \period$ and $\textY$ is a (proper) prefix of $\textX$, for an integer~$k$ with $k \ge 2$
($k > 1$ since $j \le i+\ell-1$, otherwise $i > j$ or $\intervalI \cap \intervalJ = \emptyset$).
First, we compute the smallest period $\period \le j - i$ of $T[i..j+\ell-1]$ in $\Oh{\ell}$ time~\cite{kolpakov99maximalRepetitions}. %
By definition, each substring of $T[i+\period..j+\ell-1]$ appears also $\period$~characters earlier.
We treat the substring $T[i..i+\period+\marginLCE-1]$ as a reference and therefore mark it protected.
Keeping the original characters in $T[i..i+\period+\marginLCE-1]$, we can restore the generated substrings of every $\eta$-node by an arithmetic progression.
This can be seen by two facts:
First, the length of the generated substring of an $\eta$-node is at most $3^\eta = \alpha \lg^2 n / \lg \sigma \le \marginLCE/2$.
Second, given an $\eta$-node with the generated substring~$T[b..e]$ with $i+\period+\marginLCE \le e \le j+\ell-1$,
we find an integer~$k$ with $k \ge 0$ such that $T[b..e] = T[b-\period^k..e-\period^k]$ and $[b-\period^k..e-\period^k] \subseteq [i..i+\period+\marginLCE-1]$ (since $e-b \le \marginLCE/2$).
Hence, we can make the interval~$[i+\period+\marginLCE+1..j+\ell-1-\marginLCE]$ \emph{recyclable}, 
which is at least as large as~$\marginLCE$, since $\abs{\intervalI \cup \intervalJ} \ge j-i + 2\gapLCE \ge \period + 2\gapLCE$ is at least $\period + 3\marginLCE$ for a sufficiently large $\gapLCE$.
The partitioning into protected and recyclable intervals is illustrated in \cref{figMergePeriod}.

For the actual merging operation, we elaborate an approach that respects \cref{invRedInterval,invWorkingSpace}:

\block{\ref{itTextSpaceMergeGap} Merging with a Gap}
We introduce a merge operation that supports the merging of two \thTreeF{}s whose \LCEI{}s have a gap of less than $\gapLCE$ characters.
The difference to \cref{lemmaTreeCombine} is that we additionally build new $\eta$-nodes on the gap between both trees.
The $\eta$-nodes whose generated substrings intersect with the gap are called \intWort{bridging} nodes.

Let \thInst[ {T[\intervalI]} ] and \thInst[ {T[\intervalJ]} ] be built on two \LCEI{}s $\intervalI$ and $\intervalJ$ with $1 \le \ibeg{\intervalJ} - \iend{\intervalI} \le \gapLCE$.
Our task is to compute the merged tree $\thInst[\substr{T}{\ibeg{\intervalI}}{\iend{\intervalJ}}]$.
We do that by (a) reprocessing \Oh{\lcontext+\rcontext} nodes at every height of both trees (according to \cref{lemmaTreeCombine}),
and (b) building the bridging nodes connecting both trees.
Like with the non-surrounded nodes, the generated substring of a bridging node can be a unique substring of the text.
This means that overwriting $T[\iend{\intervalI}-\marginLCE..\ibeg{\intervalJ}+\marginLCE]$ would invalidate the generated substrings of the bridging nodes and of some (formerly) non-surrounded nodes.
Therefore, we mark the interval~$[\iend{\intervalI}-\marginLCE..\ibeg{\intervalJ}+\marginLCE]$ as protected.
By doing so, we can use the characters of $T[\iend{\intervalI}-\marginLCE..\ibeg{\intervalJ}+\marginLCE]$ 
to 
\begin{itemize}
	\item create the bridging $\eta$-nodes, and to 
	\item reparse the non-surrounded nodes of both trees~(\cref{figEspMerge}).
\end{itemize}
The bridging nodes and their ancestors take $\oh{\lg n \lg^* n}$ words of additional space 
since building \linebreak[4] \thInst[ {T[\iend{\intervalI}+1..\ibeg{\intervalJ}-1]} ] with $\abs{\ibeg{\intervalJ} - \iend{\intervalI}} = \Oh{\gapLCE}$
takes
$(\gapLCE/2^\eta) \lg n = \oh{\gapLCE \lg \sigma} = \oh{\lg^*n \lg^2 n}$ bits (or $\oh{\lg^* n \lg n}$ words) of space.
By choosing $\gapLCE$ and $\alpha$ sufficiently large, we can store the bridging nodes in a recyclable interval while maintaining \cref{invWorkingSpace} for the merged \LCEI{}.
Finally, the time bound for this merging strategy is given in the following \lcnamecref{lemmaConcatTSESP}:

\begin{corollary}\label{lemmaConcatTSESP}
	Given two \LCEI{}s $\intervalI$ and $\intervalJ$ 
	with 
	$\ibeg{\intervalI} \le \ibeg{\intervalJ} \le \iend{\intervalI} + \gapLCE$,
we can build \thInst[\substr{T}{\ibeg{\intervalI}}{\iend{\intervalJ}}]
in $\Oh{\gapLCE \lg^* n + \lookupTime\numNodes{\gapLCE}  + \gapLCE\eta\wordpack + \lookupTime \lg^*n \lg n}$ time.
\end{corollary}
\begin{proof}
	We adapt the merging of two HSP trees (\cref{lemmaTreeCombine}) for the \thTreeF{}s.
	The difference to \cref{lemmaTreeCombine} is that we reparse an $\eta$-node by rebuilding its local surrounding consisting of
	$\Oh{(\lcontext+\rcontext) 3^\eta}$ nodes that take
	$\alpha (\lcontext + \rcontext) \lg^2 n/ \lg \sigma \le \marginLCE$ words for a sufficiently large~$\alpha$.
	According to \cref{invWorkingSpace}, there are at least $\marginLCE$ words of space left in a recyclable interval 
	to recompute an $\eta$-node, and to create the bridging nodes in the fashion of \cref{corTSESPTextSpace}.
	Both creating and recomputing takes overall 
$\Oh{\gapLCE \lg^* n + \lookupTime\numNodes{\gapLCE}  + \gapLCE\eta\wordpack}$ time.
\end{proof}

There is one problem left before we can prove the main result of the paper:
The sparse suffix sorting algorithm of \cref{secAbstAlgo} creates \LCEI{}s 
on substrings smaller than $\gapLCE$ between two \LCEI{}s temporarily when applying \cref{ruleCreate}.
We cannot afford to build such tiny \thTreeF{}s, since they cannot respect \cref{invRedInterval} and \cref{invWorkingSpace}.
Due to \cref{ruleMerge}, we eventually merge a temporarily created \DynLCE{} with a \DynLCE{} on a long \LCEI{}.
Instead of temporarily creating an \thTreeF{} covering less than $\gapLCE$ characters, we apply the new merge operation of \cref{lemmaConcatTSESP} directly, merging two trees that have a gap of less than $\gapLCE$ characters.
With this and the other properties stated above, we come to the final proof:

\begin{proof}[Proof of \cref{thmSparseSuffixSorting}]
	The analysis is split into suffix comparison, tree generation and tree merging:
	\begin{itemize}
		\item 
			Suffix comparisons are done as in \cref{thmAbstractExtraSpace}.
			LCE queries on \thTreeF{}s and HSP trees are conducted in the same time bounds (compare \cref{lemmaLCE} with \cref{corTSESPLCE}).
		\item 
			All positions considered for creating the \thTreeF{}s belong to $\Comlcp$.
			Constructing the \thTreeF{}s costs 
			$\Oh{ \abs{\Comlcp} \lg^*n + \lookupTime\numNodes{\abs{\Comlcp}}  + \abs{\Comlcp} \lg \lg n}$ overall time, due to
			\cref{corTSESPTextSpace}.
		\item 
			Merging in the fashion of \cref{lemmaConcatTSESP} does not affect the overall time:
			Since a merge of two trees introduces less than $\gapLCE$ new text positions to an \LCEI{},
			we conclude with the same analysis as in \cref{thmExtraSpace} that the time for merging is upper bounded by the construction time.
	\end{itemize}
	Plugging the times for suffix comparisons, tree construction and merging in \cref{thmAbstractExtraSpace} yields the overall time
	\[
		\Oh{\abs{\Comlcp} \lg^* n + \lookupTime\numNodes{\abs{\Comlcp}}  + \abs{\Comlcp} \lg\lg n } =
		\Oh{\abs{\Comlcp} \tuple{\lookupTime (\lg \sigma)^{0.7} / (\lg n)^{1.2} + \lg\lg n }  } =
                \Oh{\abs{\Comlcp} (\sqrt{\lg \sigma} + \lg \lg n ) },
	\]
	since $\lookupTime = \Oh{\lg n}$.
	The time for searching and sorting the suffixes is 
	\(
		\Oh{ m \lg m \lg^* n \lg n}
	\).
	The auxiliary data structures used are $\SAVL(\Suf(\Pos))$, the search tree~\LCEITree{} for the \LCEI{}s,
	and the list of recyclable intervals, each taking $\Oh{m}$ words of space.
\end{proof}

\section{Conclusion}
In the first part, we introduced the HSP trees based on the ESP technique as a new data structure that (a) answers LCE queries, and (b) can merge with another HSP tree to form a larger HSP tree.
With these properties, HSP trees are an eligible choice for the mergeable LCE data structure needed for the sparse suffix sorting algorithm presented here.
 
In the second part, we developed a truncated version with a trade-off parameter determining the height at which to cut off the lower nodes.
Setting the trade-off parameter adequately, the truncated HSP tree fits into text-space.
As a result of independent interest, we obtained an LCE data structure with a trade-off parameter, like other already known solutions.
Although not shown here, an ESP tree can similarly (a) answer LCE queries, (b) be merged, and (c) be truncated. 
However, 
answering LCE queries or merging two ESP trees is by a factor of~$\Oh{\lg n}$ slower than when the operations are performed with HSP trees.

\JO[In the appendix, we noted that the maximum number of fragile nodes in an ESP tree of a string of length~$n$ can be at least~\Om{\lg^2 n}, ]
{We also noted that the maximum number of fragile nodes in an ESP tree of a string of length~$n$ can be at least~\Om{\lg^2 n}, }
which invalidates the upper bound of~\Oh{\lg n \lg^*n} on the maximal number of fragile nodes postulated in~\cite{Cormode2007sed}.
This result also invalidates theoretical results that depend on the ESP technique 
(e.g., for approximating the edit distance with moves~\cite{Cormode2007sed} or the LZ77 factorization~\cite{Cormode05substring}, or for building indexes~\cite{TakabatakeNKTS16,FukunagaTIS16,MaruyamaNKS13,TakabatakeTS14}).
We could quickly provide a new upper bound of~\Oh{\lg^2 n \lg^* n}, but it remains an open problem to refine our bounds.
Luckily, our proposed HSP technique can be used as a substitution for the ESP technique, since HSP trees and ESP trees share the same bounds for construction time and space usage.
By switching to the HSP technique, we regain the promised \Oh{\lg n \lg^* n} number of fragile nodes.
It is easy to see that this result also recovers the postulated \Oh{\lg n \lg^* n} approximation bound on the edit distance matching problem~\cite{Cormode2007sed,TakabatakeNKTS16}:
Given \etInst[\textT] of a string~$\textT$ of length~$n$, 
it is assumed by \citet[Theorem 7]{Cormode2007sed} that changing/deleting a character of~$\textT$, or inserting a character in~$\textT$ changes \Oh{\lg^*n \lg n} nodes in \etInst[\textT].
Although we only provided proofs that \ac{preapp} characters to~$\textT$ changes \Oh{\lg^* n \lg n} nodes of \hInst[\textT],
it is easy to generalize this result by applying a merge operation:
Given that we insert a character~$\Char{c} \in \Sigma$ between $T[i]$ and $T[i+1]$, 
the trees $\hInst[\textT]$ and $\hInst[ {\textT[1..i]\Char{c}\textT[i+1..]} ]$ differ in at most \Oh{\lg^* n \lg n} nodes, 
since appending \Char{c} to \hInst[ {\textT[1..i]} ] and merging \hInst[ {\textT[1..i]\Char{c}} ] with \hInst[ {\textT[i+1..]} ] 
changes \Oh{\lg^* n \lg n} nodes. 
The same can be observed when deleting or changing the $i$-th character.

In the light of the theoretical improvements of the HSP over the ESP, it is interesting to evaluate how HSP behaves practically.
Especially, we are interested in how well HSP behaves in the context of grammar compression~\cite{Bannai16} like the ESP-index~\cite{MaruyamaNKS13,TakabatakeTS14}
on highly repetitive texts, where a more stable behavior of the repetitive nodes could lead to an improved compression ratio.

From the theoretical point of view, it would be interesting to compute the sparse suffix sorting with a 
trade-off parameter adjusting space and query time such that this parameter can be chosen from a continuous domain like the result we presented for the LCE query data structure.

In the case that we can impose a restriction on the set of suffixes to sort,
\citet{Karkkainen96sparse} presented a sparse suffix sorting algorithm running in optimal \Oh{n} time while using \Oh{m} words of space,
given that $\Pos$ is a set of equally spaced text positions.
We wonder whether it is also possible to gain a benefit when only every $i$-th entry of $\SA{}$ is needed,
i.e., the order of each $i$-th lexicographically smallest suffix
for an arithmetic progression $i = c, 2c, 3c,\ldots$ with a constant integer~$c \ge 2$.
Related to this problem is the suffix selection problem, i.e., to find the $i$-th lexicographically smallest suffix for a given integer~$i$.
Interestingly, \citet{Franceschini07suffix} showed that the suffix selection problem can be solved in \Oh{n} time in the comparison model, whereas 
suffix sorting is solved in \Ot{n \lg n} time within the same model.

\ifthenelse{\boolean{withAlgo}}{%
	\block{Mergeable Rank-Support}
	Remembering the note after \cref{lemmaTruncatedLCEPointer},
	we are unaware whether a rank-support data structures can be mergeable.
	Given two bit vectors~\bv{1} and~\bv{2}, both with a rank-support data structure,
	the task is to compute a rank-support data structure on the concatenation of \bv{1} and \bv{2}
	in sub-linear time in the total lengths of both bit vectors.

	\block{Working Space Aware Compressed Bit Vectors}
	Although there are bit vectors with rank-support that can be stored in compressed space (e.g.~\cite{pagh01low}),
	there is, to the best of our knowledge, no bit vector representation that can be constructed within compressed space or online.
	
\newcommand{\BST}{\ensuremath{\mathcal{B}}}
\newcommand{\BSPT}{\ensuremath{\UnaryOperator[T,\Pos]{\BST}}}
\newcommand{\pll}[1]{\ensuremath{\operatorname{ll}(#1)}}
\newcommand{\plr}[1]{\ensuremath{\operatorname{lr}(#1)}}
\newcommand{\prl}[1]{\ensuremath{\operatorname{rl}(#1)}}
\newcommand{\prr}[1]{\ensuremath{\operatorname{rr}(#1)}}
\newcommand{\plllcp}[1]{\ensuremath{\operatorname{cp}(#1, \mathtt{ll})}}
\newcommand{\plrlcp}[1]{\ensuremath{\operatorname{cp}(#1, \mathtt{lr})}}
\newcommand{\prllcp}[1]{\ensuremath{\operatorname{cp}(#1, \mathtt{rl})}}
\newcommand{\prrlcp}[1]{\ensuremath{\operatorname{cp}(#1, \mathtt{rr})}}
\section{Open Problems}

The LZ77 factorization of \cref{sec77} could be performed on the sparse suffix tree, i.e.,
the suffix tree, in which all leaves not corresponding to the $m$ selected suffix are omitted.
Then the factorization can only create referencing factors at the $m$ selected positions.
The referencing factors have a referred position within~$\Pos$.

Each node of $\SAVL(\Suf(\Pos))$ represents a suffix starting at a position of~$\Pos$.
Besides a starting position, a node~$v$ additionally stores $\max \menge{ \lcp[u,v] : u \text{~ancestor of~} v}$.
The node~$u$ is the lowest ancestor who has $v$ in either its left or right subtree.
In case such an ancestor exists (the stored LCP value is greater than zero),
$v$ additionally records in a bit whether it is in $u$'s left or right subtree, such that $v$ and this bit uniquely determine the position of~$u$ in the tree.
Although we can transform $\SAVL(\Suf(\Pos))$ to $\SSA[T,\Pos]$ with a simple in-order traversal as shown in \cref{thmAbstractExtraSpace},
it is not obvious whether $\SAVL(\Suf(\Pos))$ holds enough information to determine the contents of $\SLCP[T,\Pos]$.
A conjecture is that the stored LCP values in $\SAVL(\Suf(\Pos))$ are actually a permutation of $\SLCP[T,\Pos]$.

Instead of using the \SAVLT{} of \citet{Irving2003sbs}, we can also
devise a data structure that is a more natural choice for computing the sparse suffix sorting:
A balanced binary search tree (e.g., an AVL or red-black tree) with suffixes as keys, where each node is augmented with four LCP values.
We call this data structure a \ac{BSPT}, and write $\BSPT$ for the \ac{BSPT} containing suffixes beginning at the text positions of a set~$\Pos$.
\todo{Can \BSPT be abstracted to be a tree on strings, not necessarily suffixes?}

In the following we identify nodes with the starting positions of their corresponding suffixes, i.e., a node~$v \in \BSPT$ corresponds to the suffix~$T[v..]$.
A node~$v$ stores the LCP values

\begin{textminipage}{0.45\linewidth}
\begin{itemize}
	\item $\plllcp{v} := \lcp(T[v..], T[\pll{v}..])$, 
	\item $\plrlcp{v} := \lcp(T[v..], T[\plr{v}..])$, 
	\item $\prllcp{v} := \lcp(T[v..], T[\prl{v}..])$, and 
	\item $\prrlcp{v} := \lcp(T[v..], T[\prr{v}..])$, 
\end{itemize}
\end{textminipage}
\begin{minipage}{0.4\linewidth}
	\tikzset{%
		mininode/.style={%
			, inner sep=1pt
			, circle
			, draw
		},
		subtree/.style={%
			, draw
			, shape border uses incircle
			, minimum width =1em
			, minimum height=1.7em
			, isosceles triangle stretches
			, inner sep=0pt
			, isosceles triangle
			, shape border rotate=90
			, yshift=-1.5em
		},
		pic/.style={%
			, sibling distance=9em
			, level distance=2em
		}
	}
	\begin{tikzpicture}[pic]
		\node[circle,draw]  {$v$}
			child{ node[mininode] {\phantom{}} {node[subtree] (A) {}  }}
			child{ node[mininode] {\phantom{}} {node[subtree] (B) {}  }}
			;
		\node [mininode] at (A.south west) {};
		\node [below left=1em of A.south west] (pll) {\pll{v}};
		\draw (A.south west) -- (pll);

		\node [mininode] at (A.south east) {};
		\node [below left=1em of B.south west] (prl) {\prl{v}};
		\draw (B.south west) -- (prl);

		\node [mininode] at (B.south west) {};
		\node [below right=1em of A.south east] (plr) {\plr{v}};
		\draw (A.south east) -- (plr);

		\node [mininode] at (B.south east) {};
		\node [below right=1em of B.south east] (prr) {\prr{v}};
		\draw (B.south east) -- (prr);
	\end{tikzpicture}
\end{minipage}

where
\begin{itemize}
	\item $\pll{v}$ (resp.\ $\plr{v}$) denotes the leftmost (resp.\ rightmost) node in the subtree rooted at the \emph{left} child of $v$, and
	\item $\prl{v}$ (resp.\ $\prr{v}$) denotes the leftmost (resp.\ rightmost) node in the subtree rooted at the \emph{right} child of $v$.
\end{itemize}
If $v$ does not have a left (resp.\ right) child, set $\pll{v},\plr{v} \gets v$ (resp.\ $\prl{v}, \prr{v} \gets v$) for convenience.
Using a pointer based structure for the tree, $\BSPT$ occupies \Oh{\abs{\Pos}} space.
It has the following properties:
\todo{Check this!}
\begin{itemize}
	\item $\SSA[T,\Pos]$ and $\SLCP[T,\Pos]$ can be computed by an in-order traversal of the tree. 
		We have $\SAVL[T,\Pos][i] = v$ for a node~$v$ with in-order number~$i$.
		Additionally, $\SLCP[T,\Pos][i] = \plrlcp{v}$ (resp.\ $\SLCP[T,\Pos][i+1] =  \prrlcp{v}$) if $v$ has a left (resp.\ right) child, 
	\item Accessing $\SSA[T,\Pos]$ or $\SLCP[T,\Pos]$ can be supported in $\Oh{\lg \abs{\Pos}}$ time, given that we augment each node with its subtree size.
\end{itemize}

\begin{figure}[ht]
	\centering{%
		\newcommand{\PatternSeachCase}{1}
		\providecommand{\PatternSeachCase}{1}
\begin{tikzpicture}
\newcommand{\PH}{\phantom{\ensuremath{\textY_i}}}
	\matrix 
	[ matrix of nodes
	, font=\ttfamily
	, inner sep=1pt
	, minimum width=2.5em
	, minimum height=2em
	, nodes={draw=none, fill=none,text height=1em,text depth=0.5em}
	, nodes in empty cells,
	]
	(warray) {%
		{$\cdots$} & {\pll{v}}         & \PH{} & {$\cdots$} & \PH{} & {\plr{v}} & {$v$}       & {\prl{v}} & \PH{} & {$\cdots$} & \PH{} & \prr{v}           & {$\cdots$} \\
		\PH{}      & $\textY[1]$       & \PH{} & \PH{}      & \PH{} & \PH{}     & $\textY[1]$ & \PH{}     & \PH{} & \PH{}      & \PH{} & $\textY[1]$ \\
		\PH{}              & $\PH{}$          & \PH{} & \PH{}      & \PH{} & \PH{}     & $\vdots$    & \PH{}     & \PH{} & \PH{}      & \PH{} & $\vdots$ \\
\PH{}              & $\vdots$          & \PH{} & \PH{}      & \PH{} & \PH{}     & $\PH{}$    & \PH{}     & \PH{} & \PH{}      & \PH{} & $\textY[\rho]$ \\
\PH{}              & $\PH{}$          & \PH{} & \PH{}      & \PH{} & \PH{}     & $\PH{}$    & \PH{}     & \PH{} & \PH{}      & \PH{} & \PH{}  \\
\PH{}              & $\PH{}$          & \PH{} & \PH{}      & \PH{} & \PH{}     & $\PH{}$    & \PH{}     & \PH{} & \PH{}      & \PH{} & \PH{}  \\
\PH{}              & $\textY[\lambda]$ & \PH{} & \PH{}      & \PH{} & \PH{}     & $\PH{}$    & \PH{}     & \PH{} & \PH{}      & \PH{} & \PH{}  \\
\PH{}              & $\PH{}$           & \PH{} & \PH{}      & \PH{} & \PH{}     & $\PH{}$     & \PH{}     & \PH{} & \PH{}      & \PH{} & \PH{}  \\
\PH{}              & $\PH{}$          & \PH{} & \PH{}      & \PH{} & \PH{}     & $\PH{}$    & \PH{}     & \PH{} & \PH{}      & \PH{} & \PH{}  \\
\PH{}              & $\PH{}$           & \PH{} & \PH{}      & \PH{} & \PH{}     & $\PH{}$     & \PH{}     & \PH{} & \PH{}      & \PH{} & \PH{}  \\
\PH{}              & $\PH{}$           & \PH{} & \PH{}      & \PH{} & \PH{}     & $\PH{}$     & \PH{}     & \PH{} & \PH{}      & \PH{} & \PH{}  \\
};

 		\pgfdeclarelayer{bg}    %
 		\pgfsetlayers{bg,main}  %

		\begin{pgfonlayer}{bg}

\fill [color=gray!30] (warray-1-1.north west) rectangle node {} (warray-1-13.south east);
\node [left=0pt of warray-1-1.west] {$\SSA[T,\Pos] = $};

\fill [color=gray!20] (warray-2-2.north west) rectangle node {} (warray-9-2.south east);
\fill [color=gray!20] (warray-2-7.north west) rectangle node {} (warray-10-7.south east);
\fill [color=gray!20] (warray-2-8.north west) rectangle node {} (warray-11-8.south east);
\fill [color=gray!20] (warray-2-12.north west) rectangle node {} (warray-5-12.south east);

\fill [color=solarizedRed!50] (warray-2-2.north west) rectangle node {} (warray-7-2.south east);
\fill [color=solarizedRed!50] (warray-2-12.north west) rectangle node {} (warray-4-12.south east);

\draw [|<->|] ([xshift=-0.5em] warray-2-2.north west) -- node [midway,auto,anchor=east] {$\lambda$} ([xshift=-0.5em] warray-7-2.south west);
\draw [|<->|] ([xshift=0.5em] warray-2-12.north east) -- node [midway,auto,anchor=west] {$\rho$} ([xshift=0.5em] warray-4-12.south east);

\ifthenelse{\PatternSeachCase = 1}{%
\fill [color=solarizedGreen!50] (warray-2-2.north west) rectangle node {} (warray-8-2.south east);
\fill [color=solarizedGreen!50] (warray-2-7.north west) rectangle node {} (warray-8-7.south east);
\draw [|<->|] ([xshift=0.5em] warray-2-2.north east) -- node [midway,auto,anchor=west] {$\plllcp{v}$} ([xshift=0.5em] warray-8-2.south east);
\draw [|<->|] ([xshift=-0.5em] warray-2-7.north west) -- node [midway,auto,anchor=east] {$\plllcp{v}$} ([xshift=-0.5em] warray-8-7.south west);
\node at (warray-8-2.center) {$\circ$};
\node at (warray-8-7.center) {$\circ$};

\node [left=2em of warray-8-7](greaterLLUp) {$\textT[\pll{v}+\lambda]$};
\node  [above=0em of greaterLLUp] (greaterLL) {$\textT[p_v+\lambda] = $};
\node [below=0em of greaterLLUp] (greaterLLDown) {$\not=\textY[\lambda+1]$};
\node at (warray-8-7.center) {$\circ$};
\draw (warray-8-7.center) -- (greaterLL);
\draw (warray-8-2.center) -- (greaterLLUp);
\draw [|<->|] (warray-8-7.south east) -- node [midway,auto,anchor=south] {$\textY$ must occur here.} (warray-8-12.south east);
}{}

\ifthenelse{\PatternSeachCase = 2}{%
\fill [color=solarizedRed!50] (warray-2-7.north west) rectangle node {} (warray-9-7.south east);
\fill [color=solarizedGreen!50] (warray-2-2.north west) rectangle node {} (warray-7-2.south east);
\fill [color=solarizedGreen!50] (warray-2-7.north west) rectangle node {} (warray-7-7.south east);
\fill [color=solarizedBlue!50] (warray-2-8.north west) rectangle node {} (warray-10-8.south east);
\draw [|<->|] ([xshift=0.5em] warray-2-2.north east) -- node [midway,auto,anchor=west] {$\plllcp{v}$} ([xshift=0.5em] warray-7-2.south east);
\draw [|<->|] ([xshift=-0.5em] warray-2-7.north west) -- node [midway,auto,anchor=east] {$\lambda$} ([xshift=-0.5em] warray-7-7.south west);
\draw [|<->|] ([xshift=-0.5em] warray-8-7.north west) -- node [midway,auto,anchor=east] {$\ell$} ([xshift=-0.5em] warray-9-7.south west);
\draw [|<->|] ([xshift=1.5em] warray-2-8.north east) -- node [midway,auto,anchor=west] {$\prllcp{v}$} ([xshift=1.5em] warray-10-8.south east);
\node at (warray-10-7.center) {$\circ$};
\node [left=2em of warray-10-7,label={$\textY[\lambda+\ell+1] <$}] {$\textT[p_v+\lambda+\ell]$};
\draw ([xshift=-3em] warray-10-7.center) -- (warray-10-7.center);
}{}

\ifthenelse{\PatternSeachCase = 3}{%
\fill [color=solarizedGreen!50] (warray-2-2.north west) rectangle node {} (warray-5-2.south east);
\fill [color=solarizedGreen!50] (warray-2-7.north west) rectangle node {} (warray-5-7.south east);
\draw [|<->|] ([xshift=0.5em] warray-2-2.north east) -- node [midway,auto,anchor=west] {$\plllcp{v}$} ([xshift=0.5em] warray-5-2.south east);
\draw [|<->|] ([xshift=0.5em] warray-2-7.north east) -- node [midway,auto,anchor=west] {$\plllcp{v}$} ([xshift=0.5em] warray-5-7.south east);
\draw [dashed] ([xshift=0.5em] warray-4-2.south west) -- node [pos=0.75,auto,anchor=north] {$\lcp[\textY,\cdot] \ge \rho$} ([xshift=0.5em] warray-4-12.south east);
\node at (warray-6-2.center) {$\circ$};
\node [left=1em of warray-6-7,] (lessLL) {$\textT[v+\plllcp{v}]$};
\node [above=0em of lessLL] (lessLLUp) {$\textY[1+\plllcp{v}] \neq$};
\node at (warray-6-7.center) {$\circ$};
\draw (warray-6-7.center) -- (lessLL);
\draw (warray-6-2.center) -- (lessLLUp);
}{}
\end{pgfonlayer}
\end{tikzpicture}
	}%
	\caption{Setting of the proof of \cref{lemmaBPSTPatternMatching}\ref{itPatternSearchA}.}
	\label{figPatternSearchA}
\end{figure}
\begin{figure}[ht]
	\centering{%
		\newcommand{\PatternSeachCase}{2}
		\providecommand{\PatternSeachCase}{1}
\begin{tikzpicture}
\newcommand{\PH}{\phantom{\ensuremath{\textY_i}}}
	\matrix 
	[ matrix of nodes
	, font=\ttfamily
	, inner sep=1pt
	, minimum width=2.5em
	, minimum height=2em
	, nodes={draw=none, fill=none,text height=1em,text depth=0.5em}
	, nodes in empty cells,
	]
	(warray) {%
		{$\cdots$} & {\pll{v}}         & \PH{} & {$\cdots$} & \PH{} & {\plr{v}} & {$v$}       & {\prl{v}} & \PH{} & {$\cdots$} & \PH{} & \prr{v}           & {$\cdots$} \\
		\PH{}      & $\textY[1]$       & \PH{} & \PH{}      & \PH{} & \PH{}     & $\textY[1]$ & \PH{}     & \PH{} & \PH{}      & \PH{} & $\textY[1]$ \\
		\PH{}              & $\PH{}$          & \PH{} & \PH{}      & \PH{} & \PH{}     & $\vdots$    & \PH{}     & \PH{} & \PH{}      & \PH{} & $\vdots$ \\
\PH{}              & $\vdots$          & \PH{} & \PH{}      & \PH{} & \PH{}     & $\PH{}$    & \PH{}     & \PH{} & \PH{}      & \PH{} & $\textY[\rho]$ \\
\PH{}              & $\PH{}$          & \PH{} & \PH{}      & \PH{} & \PH{}     & $\PH{}$    & \PH{}     & \PH{} & \PH{}      & \PH{} & \PH{}  \\
\PH{}              & $\PH{}$          & \PH{} & \PH{}      & \PH{} & \PH{}     & $\PH{}$    & \PH{}     & \PH{} & \PH{}      & \PH{} & \PH{}  \\
\PH{}              & $\textY[\lambda]$ & \PH{} & \PH{}      & \PH{} & \PH{}     & $\PH{}$    & \PH{}     & \PH{} & \PH{}      & \PH{} & \PH{}  \\
\PH{}              & $\PH{}$           & \PH{} & \PH{}      & \PH{} & \PH{}     & $\PH{}$     & \PH{}     & \PH{} & \PH{}      & \PH{} & \PH{}  \\
\PH{}              & $\PH{}$          & \PH{} & \PH{}      & \PH{} & \PH{}     & $\PH{}$    & \PH{}     & \PH{} & \PH{}      & \PH{} & \PH{}  \\
\PH{}              & $\PH{}$           & \PH{} & \PH{}      & \PH{} & \PH{}     & $\PH{}$     & \PH{}     & \PH{} & \PH{}      & \PH{} & \PH{}  \\
\PH{}              & $\PH{}$           & \PH{} & \PH{}      & \PH{} & \PH{}     & $\PH{}$     & \PH{}     & \PH{} & \PH{}      & \PH{} & \PH{}  \\
};

 		\pgfdeclarelayer{bg}    %
 		\pgfsetlayers{bg,main}  %

		\begin{pgfonlayer}{bg}

\fill [color=gray!30] (warray-1-1.north west) rectangle node {} (warray-1-13.south east);
\node [left=0pt of warray-1-1.west] {$\SSA[T,\Pos] = $};

\fill [color=gray!20] (warray-2-2.north west) rectangle node {} (warray-9-2.south east);
\fill [color=gray!20] (warray-2-7.north west) rectangle node {} (warray-10-7.south east);
\fill [color=gray!20] (warray-2-8.north west) rectangle node {} (warray-11-8.south east);
\fill [color=gray!20] (warray-2-12.north west) rectangle node {} (warray-5-12.south east);

\fill [color=solarizedRed!50] (warray-2-2.north west) rectangle node {} (warray-7-2.south east);
\fill [color=solarizedRed!50] (warray-2-12.north west) rectangle node {} (warray-4-12.south east);

\draw [|<->|] ([xshift=-0.5em] warray-2-2.north west) -- node [midway,auto,anchor=east] {$\lambda$} ([xshift=-0.5em] warray-7-2.south west);
\draw [|<->|] ([xshift=0.5em] warray-2-12.north east) -- node [midway,auto,anchor=west] {$\rho$} ([xshift=0.5em] warray-4-12.south east);

\ifthenelse{\PatternSeachCase = 1}{%
\fill [color=solarizedGreen!50] (warray-2-2.north west) rectangle node {} (warray-8-2.south east);
\fill [color=solarizedGreen!50] (warray-2-7.north west) rectangle node {} (warray-8-7.south east);
\draw [|<->|] ([xshift=0.5em] warray-2-2.north east) -- node [midway,auto,anchor=west] {$\plllcp{v}$} ([xshift=0.5em] warray-8-2.south east);
\draw [|<->|] ([xshift=-0.5em] warray-2-7.north west) -- node [midway,auto,anchor=east] {$\plllcp{v}$} ([xshift=-0.5em] warray-8-7.south west);
\node at (warray-8-2.center) {$\circ$};
\node at (warray-8-7.center) {$\circ$};

\node [left=2em of warray-8-7](greaterLLUp) {$\textT[\pll{v}+\lambda]$};
\node  [above=0em of greaterLLUp] (greaterLL) {$\textT[p_v+\lambda] = $};
\node [below=0em of greaterLLUp] (greaterLLDown) {$\not=\textY[\lambda+1]$};
\node at (warray-8-7.center) {$\circ$};
\draw (warray-8-7.center) -- (greaterLL);
\draw (warray-8-2.center) -- (greaterLLUp);
\draw [|<->|] (warray-8-7.south east) -- node [midway,auto,anchor=south] {$\textY$ must occur here.} (warray-8-12.south east);
}{}

\ifthenelse{\PatternSeachCase = 2}{%
\fill [color=solarizedRed!50] (warray-2-7.north west) rectangle node {} (warray-9-7.south east);
\fill [color=solarizedGreen!50] (warray-2-2.north west) rectangle node {} (warray-7-2.south east);
\fill [color=solarizedGreen!50] (warray-2-7.north west) rectangle node {} (warray-7-7.south east);
\fill [color=solarizedBlue!50] (warray-2-8.north west) rectangle node {} (warray-10-8.south east);
\draw [|<->|] ([xshift=0.5em] warray-2-2.north east) -- node [midway,auto,anchor=west] {$\plllcp{v}$} ([xshift=0.5em] warray-7-2.south east);
\draw [|<->|] ([xshift=-0.5em] warray-2-7.north west) -- node [midway,auto,anchor=east] {$\lambda$} ([xshift=-0.5em] warray-7-7.south west);
\draw [|<->|] ([xshift=-0.5em] warray-8-7.north west) -- node [midway,auto,anchor=east] {$\ell$} ([xshift=-0.5em] warray-9-7.south west);
\draw [|<->|] ([xshift=1.5em] warray-2-8.north east) -- node [midway,auto,anchor=west] {$\prllcp{v}$} ([xshift=1.5em] warray-10-8.south east);
\node at (warray-10-7.center) {$\circ$};
\node [left=2em of warray-10-7,label={$\textY[\lambda+\ell+1] <$}] {$\textT[p_v+\lambda+\ell]$};
\draw ([xshift=-3em] warray-10-7.center) -- (warray-10-7.center);
}{}

\ifthenelse{\PatternSeachCase = 3}{%
\fill [color=solarizedGreen!50] (warray-2-2.north west) rectangle node {} (warray-5-2.south east);
\fill [color=solarizedGreen!50] (warray-2-7.north west) rectangle node {} (warray-5-7.south east);
\draw [|<->|] ([xshift=0.5em] warray-2-2.north east) -- node [midway,auto,anchor=west] {$\plllcp{v}$} ([xshift=0.5em] warray-5-2.south east);
\draw [|<->|] ([xshift=0.5em] warray-2-7.north east) -- node [midway,auto,anchor=west] {$\plllcp{v}$} ([xshift=0.5em] warray-5-7.south east);
\draw [dashed] ([xshift=0.5em] warray-4-2.south west) -- node [pos=0.75,auto,anchor=north] {$\lcp[\textY,\cdot] \ge \rho$} ([xshift=0.5em] warray-4-12.south east);
\node at (warray-6-2.center) {$\circ$};
\node [left=1em of warray-6-7,] (lessLL) {$\textT[v+\plllcp{v}]$};
\node [above=0em of lessLL] (lessLLUp) {$\textY[1+\plllcp{v}] \neq$};
\node at (warray-6-7.center) {$\circ$};
\draw (warray-6-7.center) -- (lessLL);
\draw (warray-6-2.center) -- (lessLLUp);
}{}
\end{pgfonlayer}
\end{tikzpicture}
	}%
	\caption{Setting of the proof of \cref{lemmaBPSTPatternMatching}\ref{itPatternSearchB} in case that $T[v+\lambda+\ell] < \textY[\lambda+\ell+1]$ and $\prllcp{v} \ge \lambda + \ell$. }
	\label{figPatternSearchB}
\end{figure}
\begin{figure}[ht]
	\centering{%
		\newcommand{\PatternSeachCase}{3}
		\providecommand{\PatternSeachCase}{1}
\begin{tikzpicture}
\newcommand{\PH}{\phantom{\ensuremath{\textY_i}}}
	\matrix 
	[ matrix of nodes
	, font=\ttfamily
	, inner sep=1pt
	, minimum width=2.5em
	, minimum height=2em
	, nodes={draw=none, fill=none,text height=1em,text depth=0.5em}
	, nodes in empty cells,
	]
	(warray) {%
		{$\cdots$} & {\pll{v}}         & \PH{} & {$\cdots$} & \PH{} & {\plr{v}} & {$v$}       & {\prl{v}} & \PH{} & {$\cdots$} & \PH{} & \prr{v}           & {$\cdots$} \\
		\PH{}      & $\textY[1]$       & \PH{} & \PH{}      & \PH{} & \PH{}     & $\textY[1]$ & \PH{}     & \PH{} & \PH{}      & \PH{} & $\textY[1]$ \\
		\PH{}              & $\PH{}$          & \PH{} & \PH{}      & \PH{} & \PH{}     & $\vdots$    & \PH{}     & \PH{} & \PH{}      & \PH{} & $\vdots$ \\
\PH{}              & $\vdots$          & \PH{} & \PH{}      & \PH{} & \PH{}     & $\PH{}$    & \PH{}     & \PH{} & \PH{}      & \PH{} & $\textY[\rho]$ \\
\PH{}              & $\PH{}$          & \PH{} & \PH{}      & \PH{} & \PH{}     & $\PH{}$    & \PH{}     & \PH{} & \PH{}      & \PH{} & \PH{}  \\
\PH{}              & $\PH{}$          & \PH{} & \PH{}      & \PH{} & \PH{}     & $\PH{}$    & \PH{}     & \PH{} & \PH{}      & \PH{} & \PH{}  \\
\PH{}              & $\textY[\lambda]$ & \PH{} & \PH{}      & \PH{} & \PH{}     & $\PH{}$    & \PH{}     & \PH{} & \PH{}      & \PH{} & \PH{}  \\
\PH{}              & $\PH{}$           & \PH{} & \PH{}      & \PH{} & \PH{}     & $\PH{}$     & \PH{}     & \PH{} & \PH{}      & \PH{} & \PH{}  \\
\PH{}              & $\PH{}$          & \PH{} & \PH{}      & \PH{} & \PH{}     & $\PH{}$    & \PH{}     & \PH{} & \PH{}      & \PH{} & \PH{}  \\
\PH{}              & $\PH{}$           & \PH{} & \PH{}      & \PH{} & \PH{}     & $\PH{}$     & \PH{}     & \PH{} & \PH{}      & \PH{} & \PH{}  \\
\PH{}              & $\PH{}$           & \PH{} & \PH{}      & \PH{} & \PH{}     & $\PH{}$     & \PH{}     & \PH{} & \PH{}      & \PH{} & \PH{}  \\
};

 		\pgfdeclarelayer{bg}    %
 		\pgfsetlayers{bg,main}  %

		\begin{pgfonlayer}{bg}

\fill [color=gray!30] (warray-1-1.north west) rectangle node {} (warray-1-13.south east);
\node [left=0pt of warray-1-1.west] {$\SSA[T,\Pos] = $};

\fill [color=gray!20] (warray-2-2.north west) rectangle node {} (warray-9-2.south east);
\fill [color=gray!20] (warray-2-7.north west) rectangle node {} (warray-10-7.south east);
\fill [color=gray!20] (warray-2-8.north west) rectangle node {} (warray-11-8.south east);
\fill [color=gray!20] (warray-2-12.north west) rectangle node {} (warray-5-12.south east);

\fill [color=solarizedRed!50] (warray-2-2.north west) rectangle node {} (warray-7-2.south east);
\fill [color=solarizedRed!50] (warray-2-12.north west) rectangle node {} (warray-4-12.south east);

\draw [|<->|] ([xshift=-0.5em] warray-2-2.north west) -- node [midway,auto,anchor=east] {$\lambda$} ([xshift=-0.5em] warray-7-2.south west);
\draw [|<->|] ([xshift=0.5em] warray-2-12.north east) -- node [midway,auto,anchor=west] {$\rho$} ([xshift=0.5em] warray-4-12.south east);

\ifthenelse{\PatternSeachCase = 1}{%
\fill [color=solarizedGreen!50] (warray-2-2.north west) rectangle node {} (warray-8-2.south east);
\fill [color=solarizedGreen!50] (warray-2-7.north west) rectangle node {} (warray-8-7.south east);
\draw [|<->|] ([xshift=0.5em] warray-2-2.north east) -- node [midway,auto,anchor=west] {$\plllcp{v}$} ([xshift=0.5em] warray-8-2.south east);
\draw [|<->|] ([xshift=-0.5em] warray-2-7.north west) -- node [midway,auto,anchor=east] {$\plllcp{v}$} ([xshift=-0.5em] warray-8-7.south west);
\node at (warray-8-2.center) {$\circ$};
\node at (warray-8-7.center) {$\circ$};

\node [left=2em of warray-8-7](greaterLLUp) {$\textT[\pll{v}+\lambda]$};
\node  [above=0em of greaterLLUp] (greaterLL) {$\textT[p_v+\lambda] = $};
\node [below=0em of greaterLLUp] (greaterLLDown) {$\not=\textY[\lambda+1]$};
\node at (warray-8-7.center) {$\circ$};
\draw (warray-8-7.center) -- (greaterLL);
\draw (warray-8-2.center) -- (greaterLLUp);
\draw [|<->|] (warray-8-7.south east) -- node [midway,auto,anchor=south] {$\textY$ must occur here.} (warray-8-12.south east);
}{}

\ifthenelse{\PatternSeachCase = 2}{%
\fill [color=solarizedRed!50] (warray-2-7.north west) rectangle node {} (warray-9-7.south east);
\fill [color=solarizedGreen!50] (warray-2-2.north west) rectangle node {} (warray-7-2.south east);
\fill [color=solarizedGreen!50] (warray-2-7.north west) rectangle node {} (warray-7-7.south east);
\fill [color=solarizedBlue!50] (warray-2-8.north west) rectangle node {} (warray-10-8.south east);
\draw [|<->|] ([xshift=0.5em] warray-2-2.north east) -- node [midway,auto,anchor=west] {$\plllcp{v}$} ([xshift=0.5em] warray-7-2.south east);
\draw [|<->|] ([xshift=-0.5em] warray-2-7.north west) -- node [midway,auto,anchor=east] {$\lambda$} ([xshift=-0.5em] warray-7-7.south west);
\draw [|<->|] ([xshift=-0.5em] warray-8-7.north west) -- node [midway,auto,anchor=east] {$\ell$} ([xshift=-0.5em] warray-9-7.south west);
\draw [|<->|] ([xshift=1.5em] warray-2-8.north east) -- node [midway,auto,anchor=west] {$\prllcp{v}$} ([xshift=1.5em] warray-10-8.south east);
\node at (warray-10-7.center) {$\circ$};
\node [left=2em of warray-10-7,label={$\textY[\lambda+\ell+1] <$}] {$\textT[p_v+\lambda+\ell]$};
\draw ([xshift=-3em] warray-10-7.center) -- (warray-10-7.center);
}{}

\ifthenelse{\PatternSeachCase = 3}{%
\fill [color=solarizedGreen!50] (warray-2-2.north west) rectangle node {} (warray-5-2.south east);
\fill [color=solarizedGreen!50] (warray-2-7.north west) rectangle node {} (warray-5-7.south east);
\draw [|<->|] ([xshift=0.5em] warray-2-2.north east) -- node [midway,auto,anchor=west] {$\plllcp{v}$} ([xshift=0.5em] warray-5-2.south east);
\draw [|<->|] ([xshift=0.5em] warray-2-7.north east) -- node [midway,auto,anchor=west] {$\plllcp{v}$} ([xshift=0.5em] warray-5-7.south east);
\draw [dashed] ([xshift=0.5em] warray-4-2.south west) -- node [pos=0.75,auto,anchor=north] {$\lcp[\textY,\cdot] \ge \rho$} ([xshift=0.5em] warray-4-12.south east);
\node at (warray-6-2.center) {$\circ$};
\node [left=1em of warray-6-7,] (lessLL) {$\textT[v+\plllcp{v}]$};
\node [above=0em of lessLL] (lessLLUp) {$\textY[1+\plllcp{v}] \neq$};
\node at (warray-6-7.center) {$\circ$};
\draw (warray-6-7.center) -- (lessLL);
\draw (warray-6-2.center) -- (lessLLUp);
}{}
\end{pgfonlayer}
\end{tikzpicture}
	}%
	\caption{Setting of the proof of \cref{lemmaBPSTPatternMatching}\ref{itPatternSearchC}.}
	\label{figPatternSearchC}
\end{figure}

Similar to the \SAVLT, the \ac{BSPT} can perform pattern matching within the same time bounds.
To show this, we need a small helper \lcnamecref{lemmaLCPtransitivity}:

\begin{lemma}[{\cite[Lemma 1]{Irving2003sbs}}]\label{lemmaLCPtransitivity}
	Given three strings~$\textX,\textY,\textZ$ with the lexicographic order $\textX \prec \textY \prec \textZ$, then 
	$\lcp[\textX,\textZ] = \min\tuple{\lcp[\textX,\textY], \lcp[\textY,\textZ]}$.
\end{lemma}

\begin{lemma}\label{lemmaBPSTPatternMatching}
  Given a text $T$ of length~$n$, and $\BSPT$ built on the suffixes of~$T$ whose a starting positions are in~$\Pos$,
we can find 
  $\argmax \menge{\lcp(\textY, T[p..]) \mid p \in \Pos }$ 
of a pattern~$\textY$ 
  in \Oh{\abs{\textY}\wordpack + \lg\abs{\Pos}} time.
\end{lemma}
\begin{proof}
	Analogously to the pattern matching algorithm on the enhanced suffix array~\cite[Figure 3]{manber93suffix}, 
	we perform a binary search by walking down the tree.
	We maintain two variables $\lambda,\rho \in \Int{1}{n}$ with the invariant that
$\lambda = \lcp(T[\pll{v}..], \textY)$ and $\rho = \lcp(T[\prr{v}..], \textY)$ on visiting a node $v$.
Starting at the root note, we initialize $\lambda \gets \lcp(T[\pll{\text{root}}..], \textY)$ and $\rho \gets \lcp(T[\prr{\text{root}}..], \textY)$, where $\text{root}$ is the root node of $\BSPT$.
Suppose that we access a node $v \in \BSPT$, and that $\lambda \ge \rho$
(otherwise exchange $\pll{v}$, $\plr{v}$, $\lambda$, etc., with $\prr{v}$, $\prl{v}$, $\rho$, etc., respectively).
We follow the case analysis of~\cite[Figure 2]{manber93suffix}:
\begin{enumerate}[(a)]
	\item Case $\plllcp{v} > \lambda$, see also \cref{figPatternSearchA}. 
		Then $\textY[\lambda+1] > T[v+\lambda] = T[\pll{v}+\lambda]$, and thus $\textY \succ T[v..]$.
		\begin{itemize}
	\item If $\prllcp{v} < \lambda$ then $\textY$ does not occur in $\textT$.
	\item Otherwise ($\prllcp{v} \ge \lambda$), 
			$\lambda \le \lcp[ {\textY[1..], T[\prl{v}..]} ]$.
			We set $\lambda \gets \lambda + \lcp[ {\textY[1+\lambda..], T[\prl{v}+\lambda..]} ]$,
			and descend to $v$'s right child. \label{itPatternSearchA}
		\end{itemize}
	\item Case $\plllcp{v} = \lambda$, see also \cref{figPatternSearchB}. We compute $\ell := \lcp(T[v+\lambda..], \textY[\lambda+1..])$.
		If both strings are equal, we found a match and return. \label{itPatternSearchB}
		Otherwise, we compare the first pair of mismatching characters $T[v+\lambda+\ell]$ and $\textY[\lambda+\ell+1]$.
		\begin{itemize}
			\item If $T[v+\lambda+\ell] < \textY[\lambda+\ell+1]$, then the corresponding suffix~$T[p_v..]$ of~$v$ is (lexicographically) smaller than~$\textY$.
			We abort the search if $\plrlcp{v} < \lambda + \ell$, because then the next lexicographically larger suffix~$T[p_{\prl{v}}..]$ stored in $\BSPT$ shares a shorter prefix with~$\textY$ than $T[p_v..]$ with $\textY$.
			Under the assumption that $\plrlcp{v} \ge \lambda + \ell$, 
			we set $\lambda \gets \lcp(T[\prl{v}+\lambda+\ell..], \textY[\lambda+\ell+1..])$,
			and descend to $v$'s right child.
			\item The case that $T[v+\lambda+\ell] > \textY[\lambda+\ell+1]$ is symmetrical: $T[p_v..]$ is larger than~$\textY$.
				We abort the search if $\prllcp{v} < \lambda + \ell$, because then $\lcp[ {T[\plr{v}..]},\textY] < \lcp[ {T[p_v..],\textY }]$.
			Under the assumption that $\prllcp{v} \ge \lambda + \ell$, 
			we set $\rho \gets \lcp(T[\plr{v}+\lambda+\ell..], \textY[\lambda+\ell+1..])$,
			and descend to $v$'s left child.
		\end{itemize}
	\item Case $\plllcp{v} < \lambda$, see also \cref{figPatternSearchC}. We have $\plllcp{v} \ge \rho$ since $\textY[1..\rho]$ is a common prefix of all 
	suffixes corresponding to the nodes of~$v$'s subtree (according to the assumption that~$\rho < \lambda$).
	In particular, $\rho \le \lcp[ {\textY[1..], T[\plr{v}..]} ]$.
	We set $\rho \gets \rho + \lcp[ {\textY[1+\rho..], T[\plr{v}+\rho..]} ]$, and descend to $v$'s left child. \label{itPatternSearchC}
\end{enumerate}
We never decrease $\rho$ and $\lambda$ (instead we abort in the case that $\textY$ is not a prefix of any suffixes).
Since both values are upper bounded by $\abs{\textY}$,
we compare \Oh{\abs{\textY}} characters in total.
In the word-packing model, this gives a running time of \Oh{\abs{\textY}\wordpack}.
\end{proof}

\begin{figure}[ht]
\tikzset{%
itria/.style={%
  draw,shape border uses incircle,
  isosceles triangle,shape border rotate=90,yshift=-3.6em},
itriaP/.style={%
  draw,shape border uses incircle,
  isosceles triangle,shape border rotate=90,yshift=-3.1em},
  pic/.style={%
sibling distance=2cm, 
level 2/.style={sibling distance =1.5cm}
  }
}
	\centering{%
\begin{tikzpicture}[pic,yscale=0.6,baseline=(current bounding box.center)]
\tikzset{%
itriaP/.style={%
  draw,shape border uses incircle,
  isosceles triangle,shape border rotate=90,yshift=-1.7em}}
\def\rddots{\cdot^{\cdot^{\cdot}}}
	\node[circle,draw]  {$u_{i+1}$}
		child{ node[circle,draw] {\phantom{}} {node[itriaP] {}  }}
		child{ node[circle,draw] {$u_{i}$}
			child{ node[] {\rotatebox{60}{$\cdots$}}
		child{ node[circle,draw] {$u_1$}
		child{ node[circle,draw] {$v$}
		}
		child{ node[circle,draw] {\phantom{}} {node[itriaP] {}  }}
		}
		child{ node[circle,draw] {\phantom{}} {node[itriaP] {}  }}
		}
		child{ node[circle,draw] {\phantom{}} {node[itriaP] {}  }}
		}
	;
\end{tikzpicture}
\hspace{3em}
\begin{tikzpicture}[pic,baseline=(current bounding box.center)]
	\node[circle,draw] (u) {$u$}
	child{ node[circle,draw] {$v$}
		child{ node[circle,draw] {\phantom{}} {node[itriaP] {$A$}  }}
		child{ node[circle,draw] {$w$} {node[itria] {$B$}  }}
}
		child{ node[circle,draw] {\phantom{}} {node[itriaP] {$C$}  }}
	;
\end{tikzpicture}
$\xrightarrow{\text{right rotation}}$
\begin{tikzpicture}[pic,baseline=(current bounding box.center)]
\node [circle,draw] (uP)  {$v$}
		child{ node[circle,draw] {\phantom{}} {node[itriaP] {$A$}  }}
	child{ node[circle,draw] {$u$}
		child{ node[circle,draw] {$w$} {node[itria] {$B$}  }}
		child{ node[circle,draw] {\phantom{}} {node[itriaP] {$C$}  }}
}
	;

\end{tikzpicture}
	}%
	\caption{\emph{Left}: Adding the left child~$v$ to the node~$u$. \emph{Middle} and \emph{Right}: Trees before (\emph{middle}) and after (\emph{right}) right rotating the triplet of nodes $(u,v,w)$.}
	\label{figRightRotation}
\end{figure}

\begin{figure}[ht]
	\centering{%
\begin{tikzpicture}
\newcommand{\PH}{\phantom{\ensuremath{\textY_i}}}
	\matrix 
	[ matrix of nodes
	, font=\ttfamily
	, inner sep=0pt
	, minimum width=3em
	, nodes={draw, fill=none,text height=1.2em,text depth=0.5em}
	]
	(warray) {%
		$\cdots$ & \pll{w} & $\cdots$ & \PH{} & $w$ & \PH{} & $\cdots$ & \prr{w} & $u$ & $\cdots$ \\
	};
	\node [above of=warray-1-2] (lab2) {$=\pll{u}$};
	\node [above of=warray-1-8] (lab8) {$=\plr{u}$};
	\draw (warray-1-2) -- (lab2);
	\draw (warray-1-8) -- (lab8);
	\node [left=0em of warray-1-1] {$\SSA[T,\Pos]=$};

	\draw [|<->|] ([yshift=-1em] warray-1-8.south west) -- node [midway,auto,anchor=north] {$\plrlcp{u}$} ([yshift=-1em] warray-1-9.south east);
	\draw [|<->|] ([yshift=-3em] warray-1-5.south west) -- node [midway,auto,anchor=north] {$\prrlcp{w}$} ([yshift=-3em] warray-1-8.south east);
	\draw [|<->|] ([yshift=-1em] warray-1-2.south west) -- node [midway,auto,anchor=north] {$\plllcp{w}$} ([yshift=-1em] warray-1-5.south east);
	\draw [|<->|,dashed] ([yshift=-2.6em] warray-1-2.south west) -- node [pos=0.1,auto,anchor=north] {$\plllcp{u}$} ([yshift=-2.6em] warray-1-9.south east);

\end{tikzpicture}
	}%
	\caption{Updating the value $\plllcp{u}$ after the right rotation of \cref{figRightRotation}. }
	\label{figRightRotationFixLCP}
\end{figure}

It is left to show that new suffixes can be inserted into $\BSPT$ efficiently.
Let $1 \le p \le \abs{T}, p \notin \Pos$ be a text position that we want to add to $\BSPT$.
We locate the insertion point in $\BSPT$, insert a new leaf, update all invalidated LCP values, and re-balance the tree, if necessary.
With the aid of \cref{lemmaBPSTPatternMatching}, we can update the LCP values efficiently:
\begin{description}
\item[Insertion]
	Suppose that our goal is to insert the left child $v$ of a node $u_1$ (inserting the right child is analogous by symmetry), see also \cref{figRightRotation}.
		First we update the stored LCP values in $\BSPT$, and subsequently perform rotations (if necessary).
        Let $u_1, u_2, \ldots, u_i$ be the maximal sequence of the ancestors of $v$ 
        such that $u_j$ is left child of $u_{j+1}$ for every integer~$j$ with $1 \le j \le i-1$.
		According to \cref{lemmaLCPtransitivity} we update 
		$\plllcp{u_j} \gets \lcp(T[u_j..], T[u..]) = \min \{ \lcp(T[u..], T[u_1..]), \plllcp{u_j} \}$ for every~$j$ with $1 \le j \le i$, and
        $\prllcp{u_{i+1}} \gets \lcp(T[u_{i+1}..], T[v..])$, 
		where $u_{i+1}$ is the parent of $u_i$ (we omit~$u_{i+1}$ if $u_i$ is the root node).
		In total, we need to compute the two LCP queries $\lcp(T[v..], T[u_1..])$ and $\lcp(T[u_{i+1}..], T[v..])$, 
		whose values have already been computed after locating the insertion point~$u_1$ as described in \cref{lemmaBPSTPatternMatching}.
	\item[Rotation]
		Suppose that we perform a right rotation on the nodes~$u,v,w$, where $w$ is the right child of~$v$, which is the left child of~$u$.
		This setting is also depicted in \cref{figRightRotation}.
        The right rotation makes (a) $w$ the left child of $u$, and (b) $u$ the right child of $v$.
		The operations (a) and (b) invalidate the values $\plllcp{u}$ and $\prrlcp{v}$, respectively.
		With an application of \cref{lemmaLCPtransitivity}, we can restore these values by setting
        \begin{enumerate}[(a)]
			\item $\plllcp{u} \gets \min \{ \plrlcp{u}, \prrlcp{w}, \plllcp{w} \}$ (cf. \cref{figRightRotationFixLCP}), and
	\item $\prrlcp{v} \gets \min \{ \prrlcp{v}, \plrlcp{u}, \prrlcp{u} \}$.
        \end{enumerate}
		By symmetry, left rotations are done analogously.
\end{description}

\begin{lemma}
	Updating $\BSPT$ to $\BST(T,\Pos \cup\menge{p})$ can be done in \Oh{\ell\wordpack + \lg n} time, where $p \in [1..n]$ and $\ell = \abs{T[p..]}$.
\end{lemma}

}{}

\bibliographystyle{\myBibStyle}
\bibliography{papermeta/literature}
\clearpage
\Jvar{}

\end{document}